\documentclass[journal, doublecolumn]{IEEEtran}
\usepackage{mathrsfs}
\usepackage{bbding}
\usepackage{amsfonts}
\usepackage{amsmath}
\usepackage{graphicx}
\usepackage{subfigure}
\usepackage{latexsym}
\usepackage{amssymb}
\usepackage{amsmath}
\usepackage{enumerate}
\usepackage{color}
\usepackage{boxedminipage}
\usepackage{color}
\usepackage{calc}
\usepackage{amsmath}
\usepackage{amssymb}
\usepackage{graphicx}
\usepackage{esint}
\usepackage[table]{xcolor} 
\usepackage{amsthm}
\usepackage[noblocks]{authblk}

\newtheorem{Assumption}{Assumption}

\newtheorem{Algorithm}{Algorithm} 

\newtheorem{Definition}{Definition}
\newtheorem{Lemma}{Lemma}
\newtheorem{Problem}{Problem}

\newtheorem{Remark}{Remark}
\newtheorem{Theorem}{Theorem}
\newtheorem{Corollary}{Corollary}

\newcommand{\bs}{\begin{small}}
\newcommand{\bsc}{\end{small}}
\newcommand{\txtblue}{\textcolor{black}}
\newcommand{\diag}{\text{diag}}
\usepackage[font=small]{caption}
\newcommand{\bgbox}{\vspace{0.3cm} \begin{mdframed}}
\newcommand{\edbox}{\end{mdframed} \vspace{0.2cm}}

\allowdisplaybreaks[4]

\begin{document}
\title{\txtblue{MIMO Amplify-and-Forward Precoding for Networked Control Systems}}

\author[1]{Fan Zhang}
\author[2]{Vincent K. N. Lau}
\author[1]{Gong Zhang}
\affil[1]{Hong Kong Research Center, Huawei Technologies Co. Ltd. (\{zhang.fan2, nicholas.zhang\}@huawei.com)}
\affil[2]{Department of ECE,  Hong Kong University of Science and Technology (eeknlau@ust.hk)}

\maketitle

\begin{abstract}
In this paper, we consider a MIMO networked control system (NCS) \txtblue{in which a sensor amplifies and forwards  the observed MIMO   plant state  to a remote controller} via a MIMO fading channel. We focus on the MIMO \txtblue{amplify-and-forward (AF)} precoding design at the sensor to minimize a weighted average state estimation error at the remote controller  subject to an average communication power gain constraint of the sensor. The MIMO  \txtblue{AF} precoding design is formulated as  an infinite horizon average cost Markov decision process (MDP). To deal with the curse of dimensionality associated with the MDP, we propose a novel continuous-time perturbation approach and derive an asymptotically optimal closed-form priority function for the MDP.  Based on this, we derive  a closed-form first-order optimal  dynamic MIMO \txtblue{AF} precoding solution, and the solution has an \emph{event-driven control} structure. Specifically, the sensor activates the strongest eigenchannel to deliver a dynamically weighted combination of the plant  states  to the controller when the accumulated state estimation error exceeds a dynamic threshold. We further establish technical conditions for  ensuring the stability of the MIMO NCS, and show that the mean square error  of the plant state estimation   is $\mathcal{O}\left(\frac{1}{\bar{F}}\right)$, where $\bar{F}$ is the  \txtblue{maximum AF gain} of the MIMO \txtblue{AF} precoding. 
\end{abstract}

\vspace{-0.4cm}

\section{Introduction}

\subsection{Background}
Networked control systems (NCSs) have drawn great attention in recent years due to their  growing applications in industrial  automation, smart transportation, remote robotic control, etc. \cite{coup1}.   In this paper, we consider an NCS consisting of a \emph{multiple-input multiple-output (MIMO) dynamic system} (a potentially unstable plant), a \emph{multiple-antenna sensor} and a  \emph{multiple-antenna controller}, and they form a closed-loop control   as illustrated in Fig. \ref{artNSC}. In the NCS, the sensor senses the  state of the MIMO plant and transmits  the plant state  to the remote controller over a MIMO fading channel so as to stabilize the MIMO plant. The performance of the NCS is closely related to  the communication  resource allocation (e.g., power and precoding control) at the sensor over the MIMO channel. There are many existing works on  MIMO precoding   for  multi-antenna communication  systems. In \cite{mimo1},  the authors analyze the achievable capacity region of a MIMO point-to-point system.   In \cite{selec1}, \cite{optMISO3}, antenna selection, dynamic link adaption and joint precoding are considered to increase the data rate  or minimize the  mean squared error (MSE) for MIMO systems. However, these solution frameworks focus on  optimizing physical layer  performance (e.g.,  throughput, MSE), and they are not directly related to the end-to-end performance of the NCS.  The optimization objective of the MIMO precoding problem in an NCS should be directly related to the NCS performance. Furthermore, the precoding  should be adaptive to both the channel fading matrix (which reveals good transmission opportunities) and the MIMO plant state information (which reveals the urgency of the information streams).

NCS is in fact a very challenging problem because it embraces both \emph{information theory} (for modeling the dynamics of the physical channel) and \emph{control theory} (for modeling the dynamic systems under imperfect state feedback control).  Most of the  existing works focus on the study of stabilization of an NCS under various information structures or communication scenarios (see the survey papers \cite{coup1},  \cite{survey2} and the references therein). \txtblue{For example,  in [6, Chap. 1--9],  the authors   focus on designing   encoding and decoding schemes  at the sensor  and  the controller under various types of information structures to achieve  plant stabilization.} In  \cite{dmc2}, the authors consider noiseless digital channel between the sensor and the controller and  give a lower bound on the channel rate for ensuring the NCS stability.   In  \cite{packetloss2}, the authors model the communication channel as a packet loss   erasure channel and give a lower bound on the successful transmission probability to ensure the NCS stability. In \cite{awgn1},  \txtblue{\cite{ref1}}, the authors study the stabilization of NCS over additive white Gaussian noise (AWGN) channels and obtains a minimum channel capacity requirement  for stabilization. In \cite{rate2}, the authors consider memoryless  Gaussian channel between the sensor and the controller and establish  a sequential rate distortion  framework to design encoder and decoder in order to achieve NCS stability. In \cite{Qli2}, the authors study  multi-input networked stabilization with a fading channel between the controller and plant. In all these works, the key focus is  on achieving the NCS  stability, which is only a weak form of control performance. \txtblue{In [6, Chap. 10-12], the authors focus on stochastic optimization problems for NCSs. However, the per-stage cost only depends on the plant state and plant control actions, which fails to capture the communication cost. Furthermore, the stochastic optimizations therein  are solved using numerical methods in dynamic programming theories \cite{mdpcref2}.} There are also some  works on communication resource optimization for NCSs.   In \cite{nodualeffect}, \cite{onoff2}, a  sensor scheduling scheme is proposed to minimize the linear quadratic Gaussian (LQG) cost (reflecting the plant performance) and the communication cost (penalizing the information exchange between the sensor and the controller). However, the communication channel  in \cite{nodualeffect} is a simplified on-off error-free  model, and the result therein cannot be extended to the MIMO fading channels.    In \cite{mdp1}, the authors consider sensor power control by solving a discrete-time MDP formulation which minimizes the average state estimation error and average power cost. \txtblue{In \cite{ref2}, the authors consider similar  control and communication optimization by solving a continuous-time infinite horizon discounted total cost problem.} The optimal solutions in \cite{mdp1} and \cite{ref2} are obtained  by solving the associated optimality equations, which is well-known to be very challenging \cite{mdpcref2}. The  optimality equations  therein are  solved using  the conventional  numerical value iteration algorithm (VIA)  \cite{mdpcref2}, \cite{hjbs}, which  suffers from slow convergence and lack of insights.  In addition, \cite{mdp1} assumes  single-antenna fading channels and the solution cannot be extended to the MIMO fading channels.  \txtblue{Moreover, the works that consider  MIMO transmission between the sensor and the controller in  NCSs propose  either (i) static precoding where the precoder is not adaptive to the system states (e.g., the encoder and decoder structure in  \cite{rate2} and \cite{ref3} only depends on the variance of the Gaussian source and variance of the measurement   noise),  (ii) dynamic precoding but solutions are based on numerical solutions (e.g., \cite{onoff2}, \cite{mdp1}), or (iii) dynamic precoding based on closed-form heuristic schemes (e.g., \cite{para1}, \cite{para2}). On the other hand, there are some papers  (e.g.,   \cite{ci2}--\cite{ci4}) that consider the optimization  in NCSs from the perspective of the team decision problem, where structural coding and encoding schemes are proposed and the schemes  are based on a sufficient statistic that is obtained by compressing the common information at the distributed decision makers. At one step of solving the team decision problem, the coordination strategy decision problem  (c.f., Chap. 12.3  of \cite{book1}) is formulated as a (partially observed) Markov decision process (POMDP), and the MDP/POMDP is solved using  numerical value iteration algorithms \cite{mdpcref2}, \cite{hjbs} with huge computational complexity.}

 \vspace{-0.4cm}
\subsection{Our Contribution}

In this paper, we consider a MIMO NCS where a sensor delivers the MIMO plant states to a remote controller over a MIMO wireless fading channel using the \txtblue{amplify-and-forward (AF)} precoding  as illustrated in Fig. \ref{artNSC}. Using the separation principle of control and communications \cite{rate2}, \cite{dualeffect}, the  MIMO \txtblue{AF} precoding is chosen to  minimize the  average weighted MIMO plant state estimation error at the remote controller   subject to the average communication  power gain constraint of the sensor. Specifically, the  MIMO \txtblue{AF}  precoding problem is formulated as  an infinite horizon average cost  MDP.  To address the challenge of curse of dimensionality and lack of design insights for numerical solutions to MDP,  we   propose a novel continuous-time perturbation approach and obtain an asymptotically optimal closed-form priority function for solving the associated optimality condition of the MDP. Based on the structural properties of the MIMO \txtblue{AF} precoding, we show that the solution has an  \emph{event-driven control} structure.  Specifically, the sensor  only needs to activate the \emph{strongest eigenchannel} to  transmit a dynamically weighted combination of the MIMO plant states  over the MIMO wireless channel  when the accumulated state estimation error exceeds a dynamic trigger threshold \txtblue{(which depends on the instantaneous  plant state estimation error and the instantaneous MIMO fading channel matrix as well as the state estimation error covariance). Furthermore,  we establish the   closed-form first-order optimal characterization of the dynamic trigger threshold via the closed-form priority function.} In addition, we  derive sufficient conditions regarding the communication resource needed to stabilize the MIMO plant in the MIMO NCS.   We show that the achievable MSE of the plant state estimation is  $\mathcal{O}\left(\frac{1}{\bar{F}}\right)$, where $\bar{F}$ is the \txtblue{maximum AF gain}  of the MIMO \txtblue{AF} precoding. Finally, we compare the  proposed scheme with various state-of-the-art baselines and show that significant performance gains can be achieved with low complexity.

\emph{Notations:}  Bold font is used to denote matrices and vectors. $\mathbf{A}^T$, $\mathbf{A}^\dagger$ and $\mathbf{A}^\ddagger$  denote the transpose, conjugate transpose and element-wise complex conjugate of  $\mathbf{A}$ respectively.    $\mathrm{Tr}\left(  \mathbf{A}\right)$ represents the trace of   $\mathbf{A}$. $\mathbf{I}$ represents  identity matrix with appropriate dimension. $\|\mathbf{A}\|_F$ denotes the Frobenius norm of   $\mathbf{A}=[a_{kl}]$. $\mu_{max}(\mathbf{A})$   represents the largest  eigenvalue of  a symmetric matrix $\mathbf{A}$. $\|\mathbf{A}\|$ represents the Euclidean norm of a vector $\mathbf{A}$. $|x|$ represents the absolute value of a scaler $x$.   $\mathbb{S}^n$ ($\mathbb{S}_+^n$) represents the set of $n \times n$ dimensional  (positive definite) symmetric matrices. $\nabla_\mathbf{x} f(\mathbf{x})$ denotes the column gradient vector with the $k$-th element being  $\frac{\partial f(\mathbf{x})}{\partial x_k}$. $\nabla_\mathbf{x}^2 f(\mathbf{x})$ denotes the Hessian matrix of $f(\mathbf{x})$. $f\left(x\right)=\mathcal{O}\left(g\left(x \right) \right)$ as $x\rightarrow a$ means $\lim_{x \rightarrow a}\frac{f(x)}{g(x)}<\infty$.   $\text{Re}\{ x \}$ represents the real part of $x$. $\mathbf{x}\sim\mathcal{N}(0, \mathbf{X}$) ($\mathbf{x}\sim\mathcal{CN}(0, \mathbf{X})$) means that the real-valued (complex-valued) random variable $\mathbf{x}$ is circularly-symmetric Gaussian distributed with zero mean and covariance $\mathbf{X}$. Denote $B\setminus A=\{x\in B|x\notin A\}$.
\begin{figure}[t]
\centering
  \includegraphics[width=3.2in]{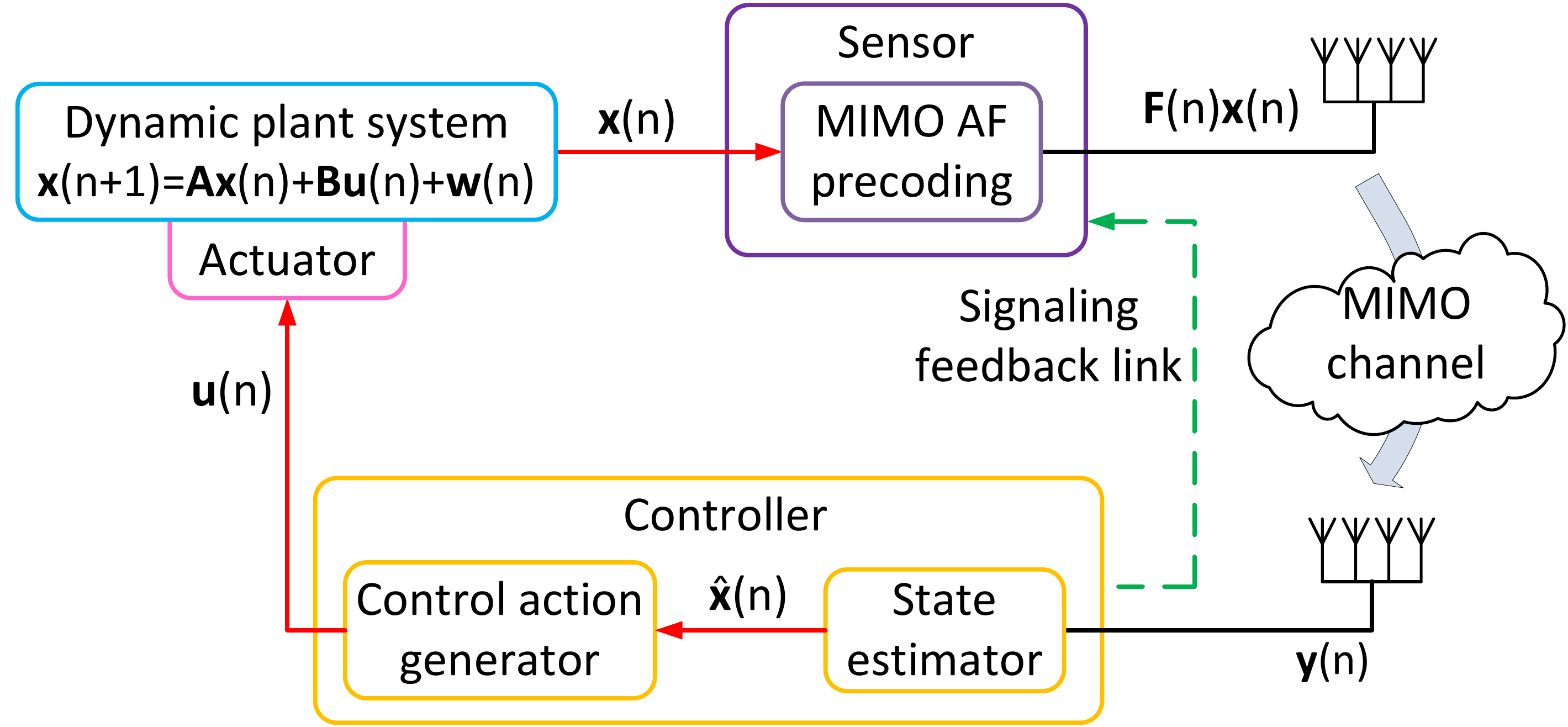}
  \caption{A typical architecture of  a MIMO NCS with MIMO \txtblue{amplify-and-forward (AF)} precoding  and analog state transmission over the MIMO fading channel.}\vspace{-0.6cm}
  \label{artNSC}
\end{figure}

\section{System Model}	\label{ncssystem}
Fig. \ref{artNSC} shows a MIMO networked control system (NCS), which consists of a MIMO plant, a multiple-antenna sensor and a multiple-antenna controller, and they form a closed-loop control.  Furthermore, we consider a slotted system, where the time dimension is partitioned into decision slots indexed by $n$ with slot duration $\tau$.  The  sensor has perfect state observation of the MIMO plant state $\mathbf{x}(n)$ at any time slot $n$.  The  controller is  geographically separated from the sensor, and there is  a MIMO  wireless  channel connecting them. At time slot $n$,  the sensor transmits  $\mathbf{x}(n)$   to the remote controller over a wireless MIMO fading channel using the MIMO \txtblue{amplify-and-forward} precoding  $\mathbf{F}(n)$. The received signal at the controller is $\mathbf{y}(n)$, which  is  passed to the state estimator at the controller to obtain a state estimate $\hat{\mathbf{x}}(n)$ based on the local information. Then, $\hat{\mathbf{x}}(n)$ is passed to the control action generator  to generate a control action $\mathbf{u}(n)$. The actuator which is co-located with the plant uses control action $\mathbf{u}(n)$ for plant actuation.

Such an NCS with a wireless fading channel covers a lot of practical application scenarios. For example, in  intelligent automobiles  \cite{automobile}  the sensors (e.g., air bag sensor, fuel pressure sensor, engine sensor) are located  all over the vehicle body and collect the real-time information that reflects the  operation conditions. This information is sent to the central processor inside the vehicle body. The processor generates control actions that control various subsystems of the vehicle.  
\vspace{-0.3cm}

\subsection{Stochastic MIMO Dynamic System }
We consider a continuous-time stochastic  plant system with dynamics $\dot{\mathbf{x}}(n)=\widetilde{\mathbf{A}}\mathbf{x}(n)+\widetilde{\mathbf{B}}\mathbf{u}(n)+\widetilde{\mathbf{w}}(n)$, $t\geq0$,  $\mathbf{x}(0)=\mathbf{x}_0$, where $\mathbf{x}(n) \in \mathbb{R}^{L \times 1}$ is the plant state process, $\mathbf{u}(n) \in \mathbb{R}^{M \times 1}$ is the plant control action,  $\widetilde{\mathbf{A}}\in \mathbb{R}^{L\times L}$, $\widetilde{\mathbf{B}} \in \mathbb{R}^{L\times M}$, and $\widetilde{\mathbf{w}}(n) \sim \mathcal{N}(0, \widetilde{\mathbf{W}})$ is an additive plant disturbance with  zero mean and covariance $\widetilde{\mathbf{W}} \in \mathbb{R}^{L\times L}$.   Without loss of generality, we assume $\widetilde{\mathbf{W}}$ is diagonal\footnote{For non-diagonal  $\widetilde{\mathbf{W}}$, we can pre-process the plant  using the \emph{whitening transformation} procedure \cite{white}. \txtblue{Specifically, let the  eigenvalue decomposition  of $\widetilde{\mathbf{W}}$ be $\mathbf{M}\widetilde{\mathbf{W}}\mathbf{M}^T=\mathbf{T}$, where $\mathbf{M}$ is a unitary matrix and $\mathbf{T}$ is diagonal. We have $\mathbf{x}_\mathbf{M}(n+1)=\mathbf{A}_\mathbf{M}\mathbf{x}_\mathbf{M}(n)+\mathbf{B}_\mathbf{M}\mathbf{u}(n)+\mathbf{w}_\mathbf{M}(n)$, where $\mathbf{x}_\mathbf{M}=\mathbf{M}\mathbf{x}$, $\mathbf{A}_\mathbf{M}=\mathbf{M}\mathbf{A}\mathbf{M}^T$, $\mathbf{B}_\mathbf{M}=\mathbf{M}\mathbf{B}$,  $\mathbf{w}_\mathbf{M}=\mathbf{M}\mathbf{w}$ and  $\mathbf{E}\big[\mathbf{w}_\mathbf{M}\mathbf{w}_\mathbf{M}^T\big]=\mathbf{T}$. Therefore, the  optimization in the NCS  based on the original  plant state $\mathbf{x}$ can be  transformed to  an equivalent optimization based on the transformed plant state  $\mathbf{x}_\mathbf{M}$ with diagonal plant noise covariance.}}. Since the sensor  samples the plant state once per time slot (with duration  $\tau$), the state dynamics of the sampled discrete-time stochastic  plant system  is given by  \cite{sampleplant}
\begin{align}	\label{plantain}
	\mathbf{x}(n+1)=\mathbf{A}\mathbf{x}(n)+\mathbf{B}\mathbf{u}(n)+\mathbf{w}(n), \quad  n=0,1, 2 \dots
\end{align}
where $\mathbf{A}=\exp( \widetilde{\mathbf{A}}\tau )$, $\mathbf{B}= \widetilde{\mathbf{A}}^{-1}\big(\exp( \widetilde{\mathbf{A}}\tau )-\mathbf{I}\big)\widetilde{\mathbf{B}}$, and $\mathbf{w}(n) = \int_0^\tau	\exp( \widetilde{\mathbf{A}}s ) \widetilde{\mathbf{w}} ((n+1)\tau-s)d s $ is a  random noise with zero mean and covariance $\mathbf{W} = \int_0^\tau \exp( \widetilde{\mathbf{A}}s ) \widetilde{\mathbf{W}}  \exp( \widetilde{\mathbf{A}}s )d s$.  We have the following assumptions on the plant model:
\begin{Assumption}[Stochastic Plant Model]	\label{assumpstabl}
	We assume that the plant system $\big(\mathbf{A}, \mathbf{B}\big)$ is controllable.~\hfill\IEEEQED
\end{Assumption}

\vspace{-0.3cm}

\subsection{MIMO Wireless  Channel Model}
The communication channel between the sensor and the controller is modeled  as a MIMO wireless fading channel. We assume that the sensor is equipped with $N_t$ antennas. Using multiple-antenna techniques, the sensor can deliver $L$ parallel plant state  streams to the receiver through spatial multiplexing.  Let $\mathbf{F} \in \mathbb{C}^{N_t \times L}$ be the MIMO \txtblue{amplify-and-forward\footnote{\txtblue{Note that we consider the AF precoding due to its computational simplicity. In \cite{ncsmimo2}, \cite{fd1}, it has also been shown that a static AF precoding is optimal in the sense that it achieves the necessary condition for stability of  linear time-invariant  systems.}} (AF)}  precoding matrix at the sensor.  The controller is equipped with $N_r$ antennas and we assume $L\leq \min\{N_t, N_r\}$ \cite{datastream}. The received signal $\mathbf{y} \in \mathbb{C}^{N_r \times 1}$ at the controller is given by
\begin{equation}	\label{recvsig}
	\mathbf{y}(n)=  \mathbf{H}(n)\mathbf{F}(n) \mathbf{x}(n)+ \mathbf{z}(n)
\end{equation}
where $\mathbf{H}(n)\in \mathbb{C}^{N_r \times N_t}$ is the channel fading matrix (CSI) from the sensor to the controller and  $\mathbf{z}(n)   \sim \mathcal{CN}(0, \mathbf{I})$ is an additive channel noise.   Furthermore, we have the following assumptions on the CSI $\mathbf{H}(n)$.
\begin{Assumption}	[MIMO Channel  Model]	\label{CSIassum}	
	$\mathbf{H}(n)$ remains constant within each decision slot  and is  i.i.d. over  slots. Specifically, each element of $\mathbf{H}(n)$ follows a complex Gaussian distribution with zero mean and unit variance.~\hfill\IEEEQED
\end{Assumption}

\vspace{-0.3cm}
\subsection{Information Structures at the Sensor and the Controller}	\label{informs}
Let $a_0^n=\left\{a(0), \dots, a(n)\right\}$ denote the history of the realizations of variable $a$ up to time $n$. The knowledge at the sensor and the controller at time slot $n$ are represented by the \emph{information structures} $I_S(n)$ and $I_C(n)$ which are given below, respectively:
\txtblue{\begin{align}
	I_S(n)&=\big\{\underbrace{\mathbf{x}_0, \mathbf{w}_0^{n-1}}_{\text{plant-related states}}, \underbrace{\mathbf{H}_0^{n},  \mathbf{z}_0^{n-1}}_{\text{com.-related states}}\big\},	\notag	\\
	I_C(n)&= \big\{\underbrace{\mathbf{u}_0^{n-1}}_{\text{plant-related states}}, \underbrace{\mathbf{E}_0^{n}, \mathbf{y}_0^{n}}_{\text{com.-related states}}\big\},\quad n=1, 2, \dots	\notag
\end{align}}and $I_S(0)=\left\{\mathbf{x}_0, \mathbf{H}(0)\right\}$ and $I_C(0)=\left\{\mathbf{E}(0), \mathbf{y}(0)\right\}$ and  we denote $\mathbf{E}(n) \triangleq \mathbf{H}(n)\mathbf{F}(n)$.  There are several observations on the information structures  at the sensor and the controller:

\begin{Remark}[Observations on  $I_S(n)$ and $I_C(n)$]	\
\begin{itemize}
	\item \txtblue{For  $I_S(n)$ at the sensor, $\mathbf{x}_0$ is the initial plant state, $\mathbf{w}_0^{n-1}$ can be obtained using $\big(\mathbf{x}_0^n, \mathbf{u}_0^{n-1}\big)$ (according to (\ref{plantain})), which are locally available at the sensor,   $\mathbf{H}_0^{n}$ can be obtained by the uplink training pilots from the controller due to   the reciprocity property of the wireless channel \cite{uplink},  $\mathbf{z}_0^{n-1}$ can be obtained using $\big(\mathbf{x}_0^{n-1},  \mathbf{E}_0^{n-1},   \mathbf{y}_0^{n-1}\big)$ (according to (\ref{recvsig})), where   $\mathbf{F}(n)$ in $\mathbf{E}(n)=\mathbf{H}(n)\mathbf{F}(n)$ is the locally generated precoding action, and $\mathbf{y}_0^{n-1}$  can be obtained by     causal   feedback from the controller as illustrated in Fig. \ref{artNSC}.}  Therefore, the information  in $I_S(n)$ can be obtained locally at the controller. We will discuss the implementation considerations in Section \ref{implantsdsd} regarding the associated signaling feedback from the controller to the sensor.
	\item For $I_C(n)$ at the controller, $\mathbf{u}_0^{n-1}$ are the past plant control actions, $\mathbf{E}_0^{n}$ can be locally measured using the  dedicated  pilots from the sensor \cite{training}, and $\mathbf{y}_0^{n}$ are the received signals  over the wireless fading channel. Therefore, the information  in $I_C(n)$ can be obtained locally at the controller. 
	~\hfill~\IEEEQED
\end{itemize}
\end{Remark}

\section{MIMO  AF Precoding Problem Formulation}

\txtblue{In this section, we first   define the MIMO  AF precoding policy and establish the no dual effect property in our NCS. Next, we  give the optimal plant control policy based on the no dual effect property, which is  the certainty equivalent (CE) controller.}  We then  formulate the  MIMO  \txtblue{AF} precoding  problem for the MIMO NCS and utilize the special  problem structure to derive the optimality conditions. 
\vspace{-0.3cm}
\subsection{MIMO \txtblue{AF} Precoding Policy and Optimal CE Controller} 
Let $\mathcal{F}_S(n)=\sigma\left(\left\{I_S(m):  m\in [0,n]\right\}\right)$ be the minimal $\sigma$-algebra containing the set $\left\{I_S(m):  \right. \\  \left. m\in[0,n] \right\}$ and $\left\{\mathcal{F}_S(n)\right\}$ be the associated \emph{filtration}  at the sensor.  At time slot $n$, the sensor determines the MIMO  \txtblue{AF} precoding  action $\mathbf{F}(n)$  according to the following  policy:
\begin{Definition}	[MIMO  \txtblue{AF} Precoding  Policy]	\label{powerpoli}
	A   \emph{MIMO  \txtblue{AF} precoding  policy} $\Omega$ for the sensor is $\mathcal{F}_S(n)$-adapted at time slot $n$, meaning that $\mathbf{F}(n)$ is adaptive to all the available information  at the sensor up to time slot $n$ (i.e., $\left\{I_S(m): m\in [0,n]\right\}$). Furthermore,  the precoding action $\mathbf{F}(n)$ satisfies the following \txtblue{AF gain constraint} of the sensor, i.e., $\mathrm{Tr}\left(\mathbf{F}^\dagger (n)\mathbf{F}(n)\right)\leq \bar{F}$ for all $n$, \txtblue{where $\bar{F}$ is the maximum AF gain of the sensor}.~\hfill~\IEEEQED
\end{Definition}

\txtblue{As indicated in  \cite{onoff2} and \cite{mdp1}, the  joint communication  and plant control optimization problem is challenging,  because the design of the communication policy and the plant control policy  are coupled together\footnote{\txtblue{The coupling is because the communication control action will affect the state estimation accuracy at the controller, which will in turn affect the plant state evolution \cite{onoff2}, \cite{mdp1}.}}.  However, by establishing   the no dual effect property (e.g., \cite{rate2}, \cite{dualeffect}), we can obtain the optimal plant control policy for the joint optimization problem, which  is given by the  CE controller.  Specifically, let $\hat{\mathbf{x}}(n)=\mathbb{E}\big[\mathbf{x}(n)\big|I_C(n)\big]$ be the  plant state estimate at the controller and  $\boldsymbol{\Delta}(n)=\mathbf{x}(n)-\hat{\mathbf{x}}(n)$  be the state estimation error. The no dual effect property is established as follows:
\begin{Lemma}	[No Dual Effect Property]	\label{lemmadual}
	Under the MIMO AF precoding policy in Definition \ref{powerpoli}, we have the following no dual effect property in our NCS:
	\begin{align}
		\mathbb{E}\big[\boldsymbol{\Delta}^T(n)\boldsymbol{\Delta}(n)\big| I_C(n)\big]=\mathbb{E}\big[\boldsymbol{\Delta}^T(n)\boldsymbol{\Delta}(n)\big| I_S(n)\big], \quad \forall n	\label{nodualeffecrdef}
	\end{align}	
\end{Lemma}
\begin{proof}
	please refer to Appendix A.
\end{proof}
\quad Using the no dual effect property in Lemma \ref{lemmadual} and Prop. 3.1 of \cite{rate2} (or  Theorem 1  in Section III of \cite{dualeffect})}, the optimal plant control policy  is given by the certainty equivalent (CE) controller:
\begin{align}	\label{statsrr}
	\mathbf{u}^\ast(n)=\boldsymbol{\Psi} \hat{\mathbf{x}}(n), \quad \forall n
\end{align}
where $\boldsymbol{\Psi}= -(\mathbf{B}^T\mathbf{Z} \mathbf{B} + \mathbf{R})^{-1}\mathbf{B}^T\mathbf{Z} \mathbf{A}$ is the feedback gain matrix,  $\mathbf{Z}$ satisfies the following  discrete-time algebraic Riccati equation\footnote{We assume that  $(\mathbf{A}, \mathbf{Q}^{1/2})$ is observable  as in the classical LQG control theories. This assumption together with Assumption \ref{assumpstabl} ensures that the DARE has a unique symmetric positive semidefinite solution \cite{poweronl}.} (DARE): $\mathbf{Z} = \mathbf{A}^T\mathbf{Z}\mathbf{A} - \mathbf{A}^T\mathbf{Z} \mathbf{B}(\mathbf{B}^T\mathbf{Z} \mathbf{B} + \mathbf{R})^{-1}\mathbf{B}^T\mathbf{Z}\mathbf{A} + \mathbf{Q}$, and $\mathbf{Q} \in \mathbb{S}^{L}_+$ and $\mathbf{R}\in \mathbb{S}^{m}_+$ are the  weighting matrices for the plant state deviation cost and  plant control cost of the LQG control associated with  the CE controller \cite{rate2}, \cite{dualeffect}.

\txtblue{We need to design a MIMO \txtblue{AF} precoding policy such that the MIMO plant  system state  is bounded. Specifically, we have the following definition on the admissible MIMO \txtblue{AF} precoding  policy:
\begin{Definition}	\emph{(Admissible  MIMO \txtblue{AF} Precoding Policy):}	\label{admisscontrolpol}
	A MIMO \txtblue{AF} precoding   policy $\Omega$  is admissible if the plant state   process    is stable under $\Omega$ and the CE controller in (\ref{statsrr}),  i.e.,  $\lim_{n \rightarrow \infty}\mathbb{E}^{\Omega}\big[\left\|\mathbf{x}(n)\right\|^2 \big]< \infty$ under $\mathbf{u}^\ast$ in (\ref{statsrr}).~\hfill~\IEEEQED
\end{Definition}}

\vspace{-0.3cm}

\subsection{MIMO \txtblue{AF} Precoding Problem Formulation and Optimality Conditions}

%
%
%
%

Under the CE controller in (\ref{statsrr})   a given  admissible MIMO precoding policy $\Omega$,  the  optimization objective  of the MIMO NCS is reduced to an  \emph{average state estimation error} which is given by:
\begin{align}	\label{cost1}
	\overline{D}\left(\Omega\right) = \limsup_{N \rightarrow \infty} \frac{1}{N} \mathbb{E}^{\Omega} \left[ \sum_{n=0}^{N-1}\boldsymbol{\Delta}^T(n)\mathbf{S} \boldsymbol{\Delta}(n)\txtblue{\tau}\right]
\end{align}
where $\mathbf{S}\in \mathbb{S}^L_+$ is a constant weighting matrix.  Similarly, the average communication power gain cost at the sensor is given by
\begin{align}
	\overline{P}\left(\Omega\right) = \limsup_{N \rightarrow \infty} \frac{1}{N} \mathbb{E}^{\Omega}\left[\sum_{n=0}^{N-1} \text{Tr}\left(\mathbf{F}^\dagger(n) \mathbf{F}(n)\right)\txtblue{\tau}\right]
\end{align}
We consider the following MIMO  \txtblue{AF}  precoding  optimization:
\begin{Problem}	\emph{(MIMO  \txtblue{AF}  Precoding Optimization for MIMO NCS):}	\label{probformu}
\begin{align} 
	& \min_{\Omega}	\quad  \overline{D}\left(\Omega\right) +  \lambda \overline{P}\left(\Omega\right)	\label{perstagecost} \\
	 =& \limsup_{N \rightarrow \infty} \frac{1}{N}\mathbb{E}^{\Omega}\left[  \sum_{n=0}^{N-1} \big(\boldsymbol{\Delta}^T(n)\mathbf{S} \boldsymbol{\Delta}(n)+ \lambda  \text{Tr}\left(\mathbf{F}^\dagger(n) \mathbf{F}(n)\right)\big)\txtblue{\tau}\right] \notag 	
\end{align}
where  $\lambda \in \mathbb{R}^+$ is  the  communication power price. The system state is  $\boldsymbol{\chi}(n)\triangleq \big\{\boldsymbol{\Delta}(n-1), \boldsymbol{\Sigma}(n), \mathbf{H}(n)\big\}$, where  $\boldsymbol{\Sigma}(n)=\mathbb{E}\big[\big(\mathbf{x}(n)-\hat{\mathbf{x}}^-(n)\big]\big)\big(\mathbf{x}(n)-\hat{\mathbf{x}}^-(n)\big)^T \big|I_C(n-1) \big]$ is the one-step state prediction error covariance and $\hat{\mathbf{x}}^-(n)\triangleq \mathbb{E}\big[\mathbf{x}(n)\big|I_C(n-1)\big]$ is the one-step plant state prediction. The state dynamic of  $\boldsymbol{\chi}(n)$ is given by\footnote{\txtblue{Here, we adopt the  augmented complex Kalman filter   in \cite{ackf} to obtain the plant state estimate,  which is the   minimum MSE estimator  for a complex-valued plant state  measurement (i.e., $\mathbf{y}$ in our problem). Please refer to Appendix B for the iterative equation of the plant state estimate $\hat{\mathbf{x}}(n)$ in order to obtain the dynamics of $\boldsymbol{\Delta}$.}}
\begin{align}
	&\boldsymbol{\Delta}(n) =\big(\mathbf{I} - \mathbf{K}^a(n) \mathbf{E}^a(n)\big) \notag \\
	&\hspace{1.5cm}\cdot \big(\mathbf{A} \boldsymbol{\Delta}(n-1)+ \mathbf{w}(n-1)\big) - \mathbf{K}^a(n)\mathbf{z}^a(n)		\label{kalman2}	\\
	& \boldsymbol{\Sigma}(n+1) = \mathbf{A} \Big(\boldsymbol{\Sigma}(n) - \boldsymbol{\Sigma}(n)(\mathbf{E}^a (n))^\dagger   \notag \\
	& \hspace{0cm}  \cdot \big(\mathbf{E}^a (n) \boldsymbol{\Sigma}(n)   (\mathbf{E}^a (n))^\dagger+\mathbf{I}\big)^{-1} \mathbf{E}^a (n) \boldsymbol{\Sigma}(n)  \Big) \mathbf{A}^T + \mathbf{W} \label{kalman31}
\end{align}
with initial conditions $\boldsymbol{\Delta}(n)=0$, and $\boldsymbol{\Sigma}(0)=0$, where $\mathbf{K}^a(n) = \boldsymbol{\Sigma}(n) (\mathbf{E}^a (n))^\dagger \big(\mathbf{E}^a (n) \boldsymbol{\Sigma}(n)   (\mathbf{E}^a (n))^\dagger  +\mathbf{I}\big)^{-1}$ is the Kalman gain, $\mathbf{E}^a(n)=\left( \begin{smallmatrix} 
  \mathbf{E} (n)\\
  \mathbf{E}^\ddagger(n)
\end{smallmatrix} \right)$ is an  augmented $2 N_r \times L$ matrix and $\mathbf{z}^a (n)=\left( \begin{smallmatrix} 
  \mathbf{z} (n)\\
  \mathbf{z}^\ddagger(n)
\end{smallmatrix} \right)$ is an augmented  $2 N_r \times 1$ noise vector.~\hfill\IEEEQED
\end{Problem}

Given an admissible MIMO \txtblue{AF} precoding policy $\Omega$, the system state  process $\{\boldsymbol{\chi}(n)\}$ is a controlled Markov chain with the following transition probability:
\txtblue{\begin{align}	\label{totaltrans1}
	& \Pr\big[\boldsymbol{\chi}(n+1)\big| \boldsymbol{\chi}(n),\mathbf{F}(n)\big] 	\notag \\
	 = & \Pr\big[\boldsymbol{\Delta}(n)\big| \boldsymbol{\chi}(n),\mathbf{F}(n)  \big]  \notag \\
	& \cdot  \Pr\big[\boldsymbol{\Sigma}(n+1)\big|\boldsymbol{\Sigma}(n),\mathbf{H}(n), \mathbf{F}(n)  \big]\Pr\big[\mathbf{H}(n+1) \big]  
\end{align}
where $\Pr[\boldsymbol{\Delta}(n)|\boldsymbol{\chi}(n),\mathbf{F}(n)   ]$ and   $ \Pr[\boldsymbol{\Sigma}(n+1)|\boldsymbol{\Sigma}(n),\mathbf{H}(n), \\ \mathbf{F}(n) ]$  are  the  state estimation error and the error covariance transition probabilities associated with the dynamics in (\ref{kalman2}) and (\ref{kalman31}).} Hence, Problem \ref{probformu} is an infinite horizon average cost MDP with system state $\boldsymbol{\chi}(n)$ and per-stage cost  \bs$\big(\boldsymbol{\Delta}^T(n)\mathbf{S} \boldsymbol{\Delta}(n)+ \lambda  \text{Tr}(\mathbf{F}^\dagger(n) \mathbf{F}(n))\big)\tau$\bsc. Exploiting the i.i.d. property of the MIMO fading channel, the optimality condition of Problem \ref{probformu}  is given by the following reduced Bellman equation \txtblue{according to Prop. 4.6.1 of \cite{mdpcref2} and  Lemma 1 of  \cite{mdpsurvey}}: 
\begin{Theorem}	[Sufficient Conditions for Optimality]	\label{mdpvarify}
	If there exists $\left(\theta^\ast, V^\ast(\boldsymbol{\Delta}, \boldsymbol{\Sigma})\right)$ that satisfies the following \emph{optimality equation} (i.e., reduced Bellman equation) for given $\boldsymbol{\Delta}, \boldsymbol{\Sigma}$:
	\begin{align}
		& \theta^\ast \tau +  V^\ast \left(\boldsymbol{\Delta}, \boldsymbol{\Sigma}\right) \label{OrgBel} \\
		 = & \mathbb{E}\bigg[\min_{F \in \Omega(\boldsymbol{\chi}) }  \Big[\big((\boldsymbol{\Delta}')^T\mathbf{S}(\boldsymbol{\Delta}') + \lambda \mathrm{Tr}\big( \mathbf{F} \mathbf{F}^\dagger\big) \big)\tau \notag \\
		 & + \sum_{\boldsymbol{\Delta}', \boldsymbol{\Sigma}'}\Pr\left[\boldsymbol{\Delta}', \boldsymbol{\Sigma}'\big|\boldsymbol{\chi}, \mathbf{F}  \right] V^\ast \left(\boldsymbol{\Delta}', \boldsymbol{\Sigma}'\right)\Big]\bigg| \boldsymbol{\Delta}, \boldsymbol{\Sigma}\bigg], 	\notag 
	\end{align}
and for all admissible MIMO  \txtblue{AF} precoding policies $\Omega$, $V^\ast (\boldsymbol{\Delta}, \boldsymbol{\Sigma})$ satisfies the following transversality condition:
	 \begin{align}	\label{transodts}
	\lim_{N\rightarrow \infty} \frac{1}{N}\mathbb{E}^{\Omega}\left[ V^\ast\left(\boldsymbol{\Delta}(N), \boldsymbol{\Sigma}(N) \right) |\boldsymbol{\chi}(0)\right]=0
\end{align}
Then, we have  the following results:
\begin{itemize}
	\item $\theta^\ast \tau=\min_{\Omega}  \overline{D}\left(\Omega\right) +  \lambda \overline{P}\left(\Omega\right)$ is the optimal cost of Problem \ref{probformu}.
	\item Suppose there exists an admissible stationary MIMO  \txtblue{AF} precoding   policy $\Omega^\ast$ with  $\Omega^\ast\left(\boldsymbol{\chi} \right) = \mathbf{F}^\ast$, where  $\mathbf{F}^\ast$ attains the minimum of the R.H.S. in (\ref{OrgBel}) for given $\boldsymbol{\chi}$. Then, $\Omega^\ast$ is the optimal MIMO  \txtblue{AF} precoding  policy for  Problem \ref{probformu}.	
\end{itemize}
\end{Theorem}
\begin{proof}
	please refer to Appendix B.
\end{proof}

Note that the optimal MIMO  \txtblue{AF} precoding  $\mathbf{F}^\ast(n)$ adapts to $\boldsymbol{\chi}(n)$, which consists of both the plant state information $\left(\boldsymbol{\Delta}(n-1), \boldsymbol{\Sigma}(n) \right)$ and the CSI $\mathbf{H}(n)$. Unfortunately, the Bellman equation in (\ref{OrgBel}) is very difficult to solve because it involves a huge number of fixed point equations w.r.t. $\left(\theta^\ast, V^\ast(\boldsymbol{\Delta}, \boldsymbol{\Sigma})\right)$. Numerical solutions such as value iteration or policy iteration  \cite{mdpcref2}  have exponential complexity w.r.t. $L$ (the dimension of $\mathbf{x}$) and are not scalable. 

%

\vspace{-0.3cm}

\section{Closed-Form First-Order Optimal MIMO  \txtblue{AF} Precoding}	\label{pertuabationana}
In this section, we shall establish a continuous-time perturbation approach to derive an approximate closed-form priority function. We show that the approximate priority function is asymptotically accurate for small $\tau$. Based on that, we derive the closed-form MIMO \txtblue{AF} precoding solution and show that the solution has an \emph{event-driven control}  structure.  We also derive an achievable upper bound of the mean square state estimation  error  in the NCS,  and discuss how the system parameters affect this upper bound.
\vspace{-0.3cm}

\subsection{Continuous-Time Approximation}
We first consider a perturbation of the  priority function $V^\ast(\boldsymbol{\Delta}, \boldsymbol{\Sigma})$ in Theorem \ref{mdpvarify} w.r.t. the slot duration $\tau$. Based on that, the optimality condition in Theorem \ref{mdpvarify} reduces to the partial differential equation (PDE) as below. 
\begin{Lemma}	\emph{(Perturbation Analysis for Solving the Optimality Equation):}	\label{perturbPDE}
	If there exists $\left(\theta, V(\boldsymbol{\Delta}, \boldsymbol{\Sigma})\right)$ \txtblue{where}\footnote{\txtblue{$V\in \mathcal{C}^2$ means that $V(\boldsymbol{\Delta}, \boldsymbol{\Sigma})$ is second order differentiable w.r.t. to each variable in $(\boldsymbol{\Delta}, \boldsymbol{\Sigma})$.}} \txtblue{$V\in \mathcal{C}^2$}  that satisfies

	\begin{itemize}
		\item the following multi-dimensional PDE:
	\begin{align}	
		\theta  &=  \boldsymbol{\Delta}^T \mathbf{S} \boldsymbol{\Delta}  + \mathbb{E}\left[ \min_{F\in \Omega(\boldsymbol{\chi})}\left[ \lambda \mathrm{Tr}\big( \mathbf{F} \mathbf{F}^\dagger\big)  \right.\right. \notag \\
		& \left.\left. -2 \text{Re}\left\{\nabla_{\boldsymbol{\Delta}}^T V \boldsymbol{\Sigma} \mathbf{F}^\dagger \mathbf{H}^\dagger \mathbf{H} \mathbf{F} \boldsymbol{\Delta}/\tau  \right\}\right]\bigg| \boldsymbol{\Delta}, \boldsymbol{\Sigma}   \right]  +\nabla_{\boldsymbol{\Delta}}^T V  \widetilde{\mathbf{A}} \boldsymbol{\Delta}   \notag \\
		&  + \frac{1}{2}\text{Tr}\left( \nabla_{\boldsymbol{\Delta}}^2 V \widetilde{\mathbf{W}} \right)+ \text{Tr}\left( \frac{\partial V}{\partial \boldsymbol{\Sigma}}  \widetilde{\mathbf{W}} \right) \label{bellman2}
	\end{align}

		\item and $V(\boldsymbol{\Delta}, \boldsymbol{\Sigma})=\mathcal{O}(\left\|\boldsymbol{\Delta}\right\|^2)$,
	\end{itemize}
	  Then,  for any $\boldsymbol{\Delta}, \boldsymbol{\Sigma}$, 
	\begin{align}	\label{errortermtaur}
		V^\ast \left(\boldsymbol{\Delta}, \boldsymbol{\Sigma} \right)  =  V(\boldsymbol{\Delta}, \boldsymbol{\Sigma}) + \mathcal{O}\left(\tau\right)
	\end{align}
	where $\mathcal{O}\left(\tau\right)$  is  the asymptotically small error term.
\end{Lemma}
\begin{proof}
	please refer to Appendix C.
\end{proof}

As a result, solving the  optimality equation in (\ref{OrgBel}) is transformed into a calculus problem of solving the PDE in (\ref{bellman2}). Furthermore, the difference between the solution of the PDE (i.e., $V(\boldsymbol{\Delta},\boldsymbol{\Sigma})$) in (\ref{bellman2}) and the priority function in (\ref{OrgBel}) (i.e., $V^\ast(\boldsymbol{\Delta}, \boldsymbol{\Sigma})$) is $\mathcal{O}\left(\tau\right)$ for sufficiently small slot duration $\tau$. In the next subsections, we shall focus on solving the PDE in (\ref{bellman2}) by leveraging the well-established theories of  differential equations.

Let $\widetilde{\Omega}^\ast$ be the minimizer of the R.H.S.  of (\ref{bellman2}) and  $\widetilde{\theta}^\ast=\limsup_{N \rightarrow \infty} \frac{1}{N} \sum_{n=0}^{N-1} \mathbb{E}^{\widetilde{\Omega}^\ast}  \left[\boldsymbol{\Delta}^T(n)\mathbf{S}\boldsymbol{\Delta}(n)  + \lambda    \mathrm{Tr}\big(\mathbf{F}(n) \mathbf{F}^\dagger(n)    \big) \right]$ be the associated  objective function value. The performance gap between $\widetilde{\theta}^\ast$ and the optimal   cost $\theta^{\ast}$ in (\ref{OrgBel}) is established in the following theorem:
\begin{Theorem}	[Performance Gap between $\widetilde{\theta}^\ast$ and $\theta^{\ast}$]	\label{perfgap}
	If $V(\boldsymbol{\Delta}, \boldsymbol{\Sigma})=\mathcal{O}(\left\|\boldsymbol{\Delta}\right\|^2)$ and $\widetilde{\Omega}^\ast$ is admissible, then the performance gap between $\widetilde{\theta}^\ast$ and $\theta^{\ast}$ is given by
	\begin{align}		\label{perfgapequ}
		\widetilde{\theta}^\ast - \theta^{\ast} =\mathcal{O}(\tau), \qquad \text{as } \tau \rightarrow 0
	\end{align}
	\end{Theorem}
\begin{proof}
Please refer to Appendix D.		
\end{proof}

Theorem \ref{perfgap} suggests that  $\widetilde{\theta}^\ast \rightarrow \theta^{\ast} $, as $\tau \rightarrow 0$. In other words, the  MIMO  \txtblue{AF} precoding    policy  $\widetilde{\Omega}^\ast$ is asymptotically optimal as $\tau \rightarrow 0$.

In this next subsections, we focus on finding the priority function $V(\boldsymbol{\Delta}, \boldsymbol{\Sigma})$ to solve the PDE  in (\ref{bellman2}). To do this, we first derive the structural properties of the optimal MIMO \txtblue{AF} precoding solution $\mathbf{F}^\ast$ that minimize the R.H.S. of (\ref{bellman2}). Based on that, we derive asymptotically accurate closed-form approximate priority function $V(\boldsymbol{\Delta}, \boldsymbol{\Sigma})$.

\subsection{Structural Properties of the  MIMO \txtblue{AF} Precoding Solution $\widetilde{\Omega}^\ast$}	\label{signalingsec}
In this section, we give the  MIMO \txtblue{AF} precoding solution for given $V(\boldsymbol{\Delta}, \boldsymbol{\Sigma})$. Let $\sigma^\ast$ and $\mathbf{U} \in \mathbb{C}^{N_t\times N_t}$ be the largest squared singular value and the associated  left  singular  matrix of $\mathbf{H}$, respectively. Let $\boldsymbol{\Xi}\triangleq  \boldsymbol{\Delta}\nabla_{\boldsymbol{\Delta}}^T V  \boldsymbol{\Sigma}\big/ \tau$, denote $\nu^\ast\triangleq \mu_{max}({\boldsymbol{\Xi}+\boldsymbol{\Xi}^T})$ and $\mathbf{q}_1$ be the associated unit column eigenvector. Then, the MIMO  \txtblue{AF} precoding  solution that minimizes the R.H.S. of the PDE in (\ref{bellman2}) is given in the following theorem:
\begin{Theorem}		\emph{(Structural Properties of the Optimal MIMO \txtblue{AF} Precoding Policy):}	\label{thmpower}
	For any given state realization $\boldsymbol{\chi}$, the optimal MIMO \txtblue{AF} precoding $\widetilde{\Omega}^\ast\left(\boldsymbol{\chi}\right)=\mathbf{F}^\ast$ that  minimizes the R.H.S.  the PDE in (\ref{bellman2}) is  given by
	\begin{itemize}
		\item \textbf{Dormant Mode:} If $\lambda   > \sigma^\ast \nu^\ast $, then $\mathbf{F}^\ast=0$.
		\item \textbf{Active Mode:}  If $\lambda   < \sigma^\ast \nu^\ast $, then $\mathbf{F}^\ast=\sqrt{\bar{F}}\mathbf{U} \boldsymbol{\Upsilon}$. $\boldsymbol{\Upsilon}\in \mathbb{R}^{N_t \times L}$ is a dynamic power splitting matrix with only non-zero row (first row) given by $\mathbf{q}_1^T$.
\end{itemize}
\end{Theorem}
\begin{proof}
	please refer to Appendix E.
\end{proof}

\txtblue{It can be observed that the MIMO AF precoding  policy  $\widetilde{\Omega}^\ast$ has an \emph{event-driven control structure} with a dynamically changing threshold  $\sigma^\ast \nu^\ast$. Specifically, the sensor either transmits using the maximum communication resource  or shuts down, depending on whether the dynamic threshold  is larger than $\lambda$ or not. Furthermore, the dynamic threshold is adaptive to the   plant state estimation error  $\boldsymbol{\Delta}$, the state estimation error covariance $\boldsymbol{\Sigma}$   and the CSI $\mathbf{H}$.} Note that the optimal MIMO \txtblue{AF} precoding $\mathbf{F}^\ast$ only activates the strongest eigenchannel $\sigma^\ast$ (with power splitting across the plant  states) to deliver a dynamically weighted combination of the plant   states $\{x_1, x_2, \dots,  x_L\}$. The power splitting dynamic weight (first row of $\boldsymbol{\Upsilon}$) is adaptive to the plant-related  states $(\boldsymbol{\Delta}, \boldsymbol{\Sigma})$ of the MIMO plant system.

\vspace{-0.3cm}

\subsection{Closed-Form Approximate Priority Function}	\label{closedformsolution}

Based on $\mathbf{F}^\ast$ in Theorem \ref{thmpower}, the priority function $V(\boldsymbol{\Delta}, \boldsymbol{\Sigma})$ is given by the solution of the PDE in (\ref{bellman2}). However,  obtaining the solution to the PDE is very challenging due to the multi-dimensional nonlinear and coupling structure. Numerical solutions such as value iteration \cite{mdpcref2} suffer from the curse of dimensionality issue and lack of design insights. 


We shall adopt the asymptotic analysis  techniques \cite{dominmethod} to derive an asymptotic solution of the PDE.  The solution is summarized in the following  lemma:
\begin{Lemma}	[Asymptotic Solution of the PDE]	\label{solpde}
The asymptotic solution of the PDE in (\ref{bellman2}) is given as follows\footnote{As discussed in Theorem \ref{thmpower}, the optimal MIMO \txtblue{AF} precoding is only related to the partial gradient $\nabla_{\boldsymbol{\Delta}} V$. Therefore, we  focus on deriving $\nabla_{\boldsymbol{\Delta}} V$ for the PDE in (\ref{bellman2}).}:
	\begin{itemize}
		\item if $\|\boldsymbol{\Delta}\|$ is  sufficiently   small (Low Urgency Regime), 
		\begin{align}	\label{valuefunc1}
			\nabla_{\boldsymbol{\Delta}} V= \left(\Phi_1(\boldsymbol{\Sigma})+\Phi_1^T(\boldsymbol{\Sigma})\right)\boldsymbol{\Delta}+\mathcal{O}(\|\boldsymbol{\Delta}\|^3)\mathbf{1} 
		\end{align}
		  where  the full expression of $\Phi_1(\boldsymbol{\Sigma})\in \mathbb{R}^{L\times L}$ is given   in (\ref{appeox12}) in Appendix E. Furthermore, we have $\|\Phi_1(\boldsymbol{\Sigma})\|_{F}=\mathcal{O}\big(\|\boldsymbol{\Sigma}\|_F\big)$.
		 \item   if $\|\boldsymbol{\Delta}\|$ is sufficiently  large (High Urgency Regime), 
		\begin{align}	
			\nabla_{\boldsymbol{\Delta}} V= \left(\Phi_2(\boldsymbol{\Sigma})+\Phi_2^T(\boldsymbol{\Sigma})\right)\boldsymbol{\Delta} +\mathcal{O}\left(\frac{1}{\|\boldsymbol{\Delta}\|^3}\right)\mathbf{1} \label{valuefunc2}
		\end{align}
		where  the full expression of $\Phi_2(\boldsymbol{\Sigma})\in \mathbb{R}^{L\times L}$ is given   in (\ref{appeox22}) in Appendix E. Furthermore, we have $\|\Phi_2(\boldsymbol{\Sigma})\|_{F}=\mathcal{O}\big(\|\boldsymbol{\Sigma}\|_F\big)$.
		\end{itemize}
\end{Lemma}
\begin{figure}[t]
\centering
\subfigure[Shape of the decision region boundary w.r.t. state estimation error $\Delta_1$ and error covariance $\Sigma_{11}$.]{
\includegraphics[width=2.6in]{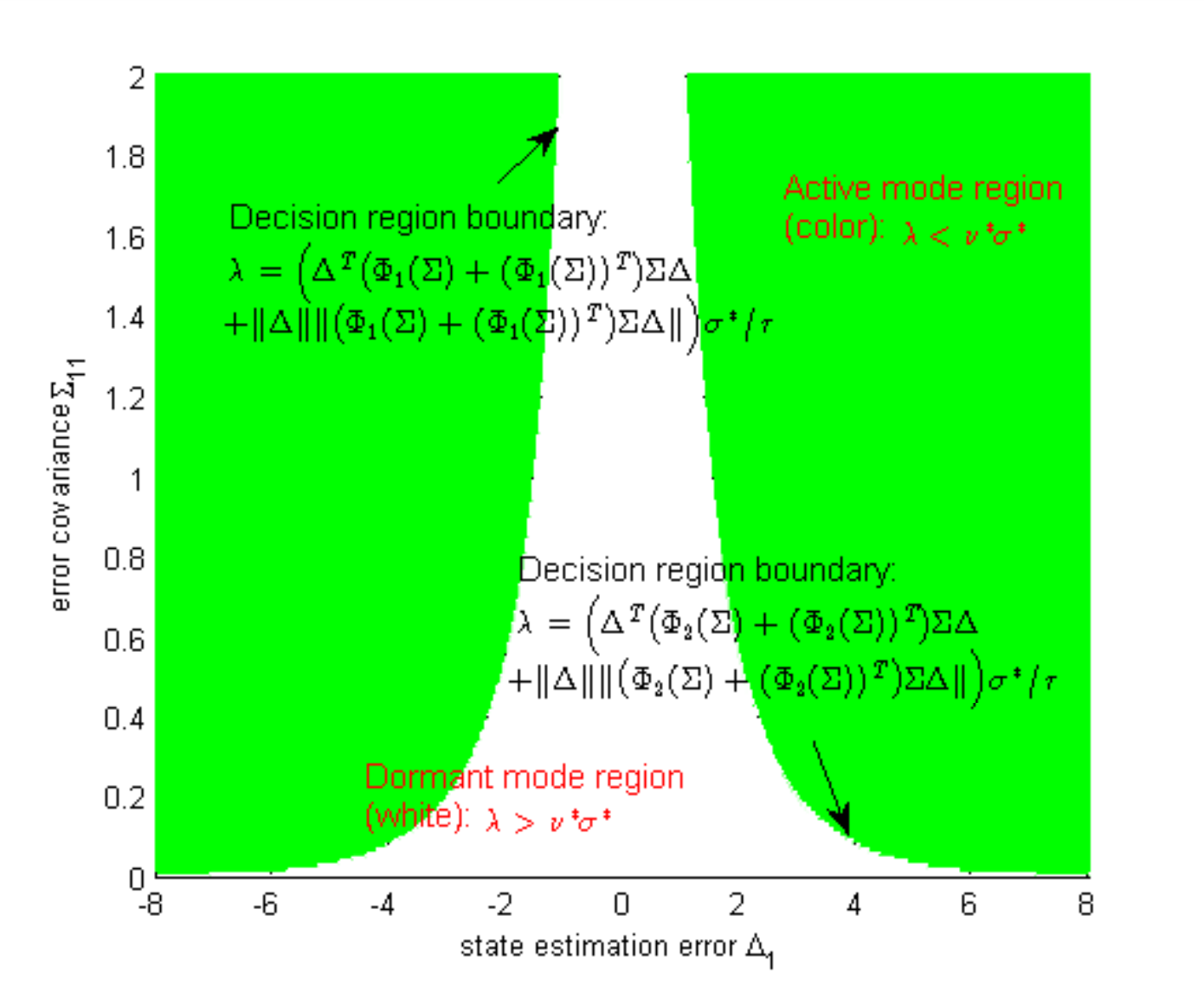}}
\hspace{10cm}
\subfigure[Decision region w.r.t. $\Delta_1$ and $\Sigma_{11}$ under instability  of $\widetilde{\mathbf{A}}$ with  $\widetilde{\mathbf{A}}=0.5\mathbf{I}, \mathbf{I}, 1.5\mathbf{I}$ and $\lambda=30$.]{
\includegraphics[width=2.6in]{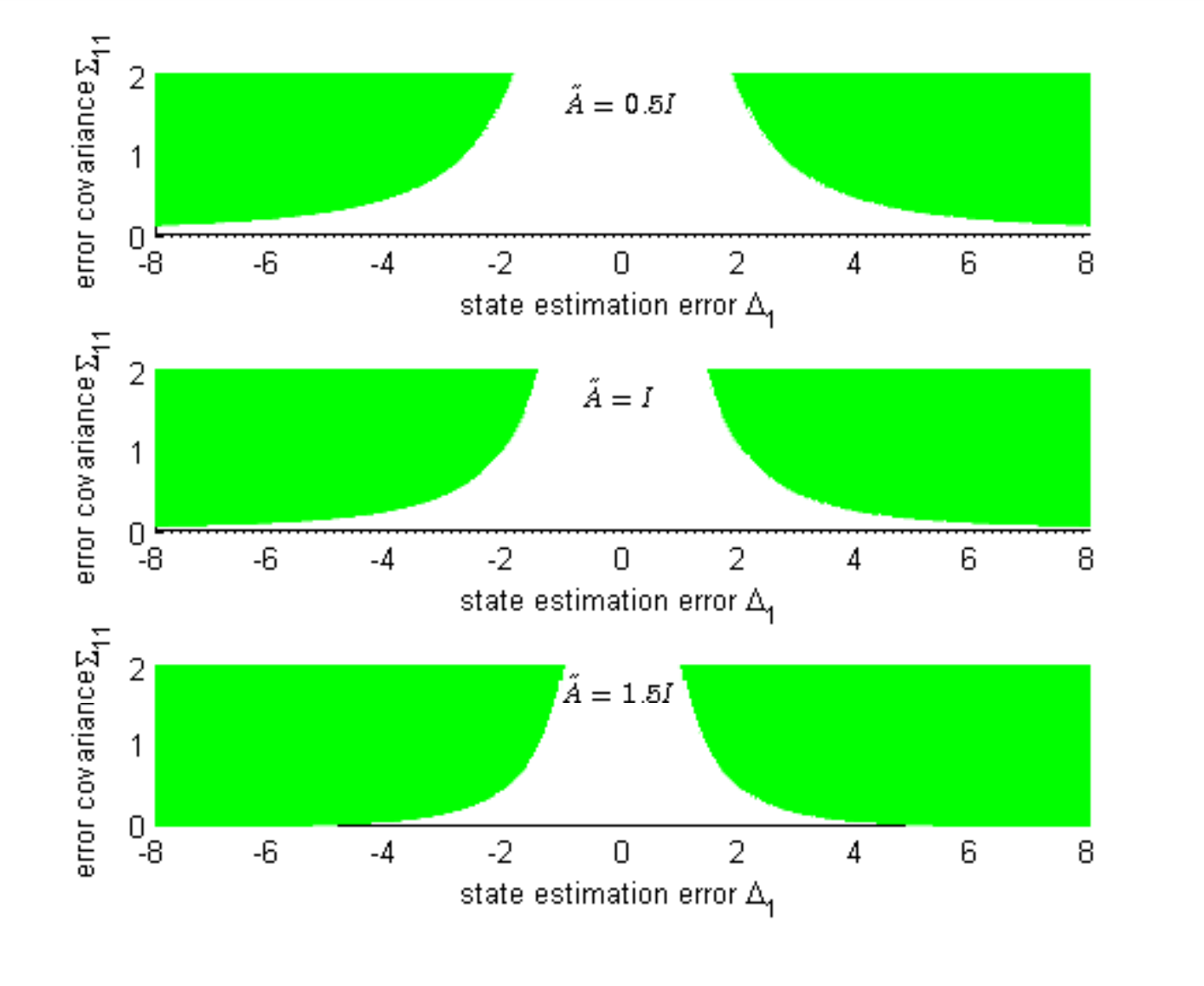}}
\hspace{0.1cm}
\subfigure[Decision region w.r.t. $\Delta_1$ and $\Sigma_{11}$ under communication power price $\lambda$ with  $\lambda=15, 20, 30$ and $\widetilde{\mathbf{A}}=\mathbf{I}$.]{
\includegraphics[width=2.6in]{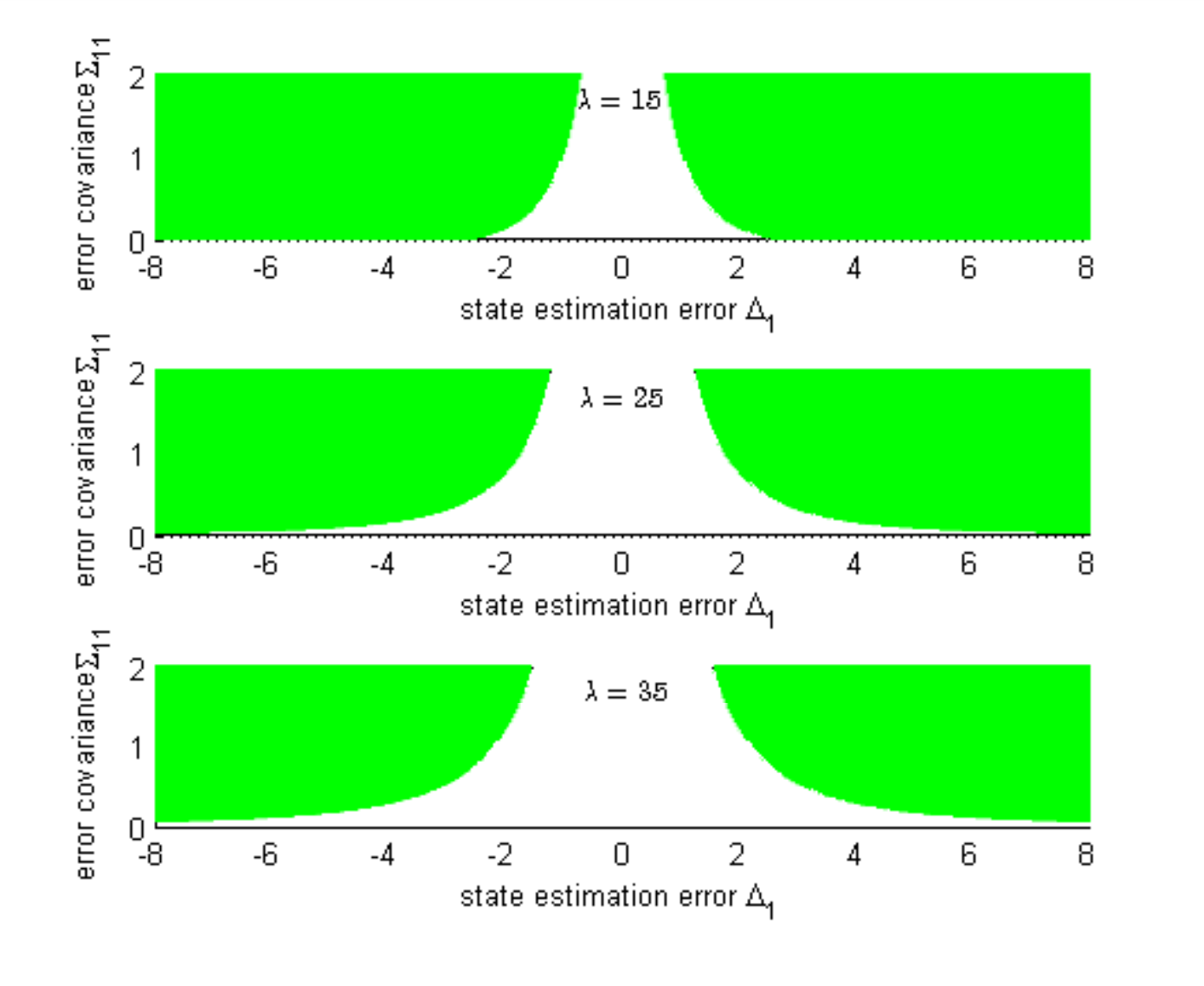}}
\caption{Decision region between the dormant/active mode w.r.t.  $\Delta_1$ and $\Sigma_{11}$. The system variables and parameters are configured as follows: $\Delta_2=1$, $\eta_{th}=2.5$, $\boldsymbol{\Sigma}=[\Sigma_{ij}]$ with $\Sigma_{12}=\Sigma_{21}=0.1$ and $\Sigma_{22}=0.5$, $\sigma^\ast=2$, $\widetilde{\mathbf{B}}=\mathbf{Q}=\mathbf{R}=\widetilde{\mathbf{W}}=\mathbf{I}$, $\bar{F}=1$, $N_t=2$, $N_r=2$, $\tau=0.05$.}	\label{decregion2}\vspace{-0.7cm}
\end{figure}

\begin{proof}
	please refer to Appendix F.
\end{proof}

As a result, we adopt the following approximation for the solution of the PDE  in (\ref{bellman2}):
\begin{align}	\label{approxvaluefunc}
	\nabla_{\boldsymbol{\Delta}} V\approx\left\{
	\begin{aligned}
		 & \left(\Phi_1(\boldsymbol{\Sigma})+\Phi_1^T(\boldsymbol{\Sigma})\right)\boldsymbol{\Delta},  \quad \text{if } \|\boldsymbol{\Delta} \|<\eta_{th}	\\
		 &  \left(\Phi_2(\boldsymbol{\Sigma})+\Phi_2^T(\boldsymbol{\Sigma})\right)\boldsymbol{\Delta} ,  \quad \text{if } \|\boldsymbol{\Delta} \|>\eta_{th} 
	\end{aligned}
	\right.
\end{align}
where  $ \eta_{th} >0$ is a solution parameter.

\begin{Remark}	[Structure of Decision Region]	\label{remark55}
	The  decision region between the active/dormant modes of $\mathbf{F}^\ast$ in Theorem \ref{thmpower} is jointly determined by the MIMO fading channel $\sigma^\ast$, the state estimation error $\boldsymbol{\Delta}$ and the one-step state prediction error covariance $\boldsymbol{\Sigma}$. The decision region has the following properties:
	\begin{itemize}
		\item 	\textbf{Shape of the Decision Region Boundary between the Active/Dormant Modes:} 
		Fig. \ref{decregion2}(a)  shows the shape of the decision region boundary between the active mode (when $\lambda<\nu^\ast \sigma^\ast$) and dormant mode (when $\lambda>\nu^\ast \sigma^\ast$). The dynamic threshold  $\nu^\ast$ grows w.r.t. $\|\boldsymbol{\Delta}\|^2$ and $\Sigma_{kk}$  at the order of $\mathcal{O}\left(\|\boldsymbol{\Delta}\|^2\right)$ and   $\mathcal{O}\left(\Sigma_{kk}^2\right)$ for all $k$, respectively. \txtblue{This is reasonable because  large state estimation error or large error covariance means  there is {\em urgency} in delivering information to the controller, which leads to activation of  the sensor transmission more frequently.}		
		\item	\textbf{Impact of Plant Dynamics on Decision Region:} 
		The active mode region enlarges as the instability degree of the plant dynamics $\widetilde{\mathbf{A}}$ increases as shown in Fig. \ref{decregion2}(b).  This is reasonable because unstable plant means it is more difficult for stabilization and hence, active mode covers  a larger  region to reach a lower plant estimation cost. 
		\item	\textbf{Impact of Communication Power Price on Decision Region:} 
		The active mode region  enlarges as the communication power price $\lambda$ decreases as shown in Fig. \ref{decregion2}(c). This means that for a smaller power price, it is appropriate to have  a large decision region for active mode so as to reach a low joint plant and communication cost.~\hfill~\IEEEQED
	\end{itemize}
\end{Remark}

\vspace{-0.3cm}
\subsection{Implementation Considerations of the MIMO \txtblue{AF} Precoding Solution} \label{implantsdsd}
Fig. \ref{approxqua} illustrates a sample path of the state evolutions and the transitions between the active and dormant modes under  the MIMO \txtblue{AF}  precoding solution in Theorem \ref{thmpower}. It can be observed that the state estimation error $\|\boldsymbol{\Delta}\|$ increases during the dormant modes and is reset during the active mode (event-driven when $\lambda   < \sigma^\ast \nu^\ast $). As such, the solution in Theorem \ref{thmpower} has an \emph{event-driven  control} structure with \emph{aperiodic reset} of $\|\boldsymbol{\Delta}\|$. We summarize the solution as follows:

\begin{figure}
\centering
\subfigure[Evolutions of the dynamic threshold $\sigma^\ast \nu^\ast$ and  transitions between the active and dormant modes.]{
\includegraphics[width=2.2in]{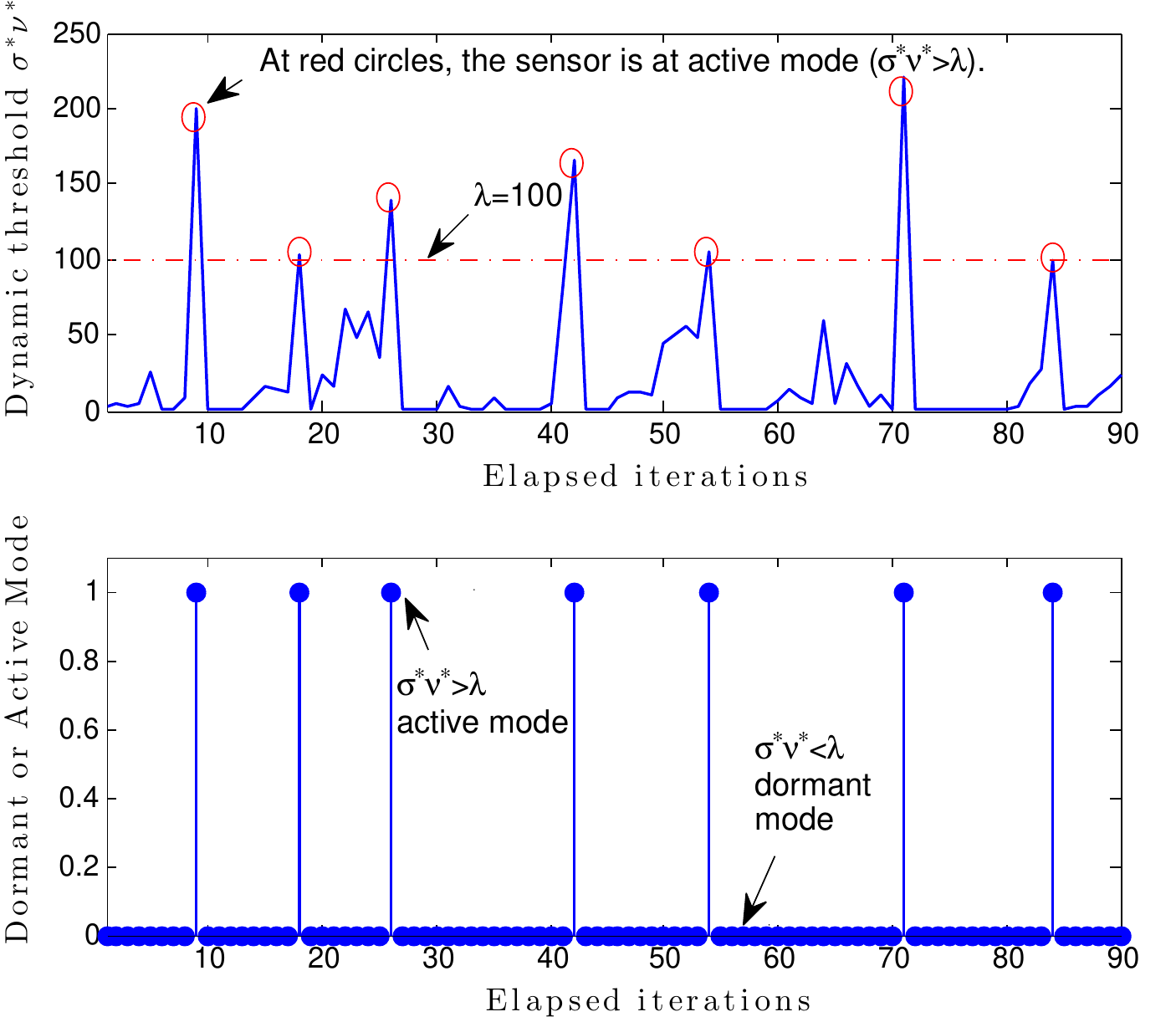}}
\hspace{0.5cm}
\centering
\subfigure[Evolutions of state estimation error $\|\boldsymbol{\Delta}\|^2$ and virtual state estimation error $\|\widetilde{\boldsymbol{\Delta}}\|^2$ in (\ref{dynvirtual}).]{
\centering\includegraphics[width=2.25in]{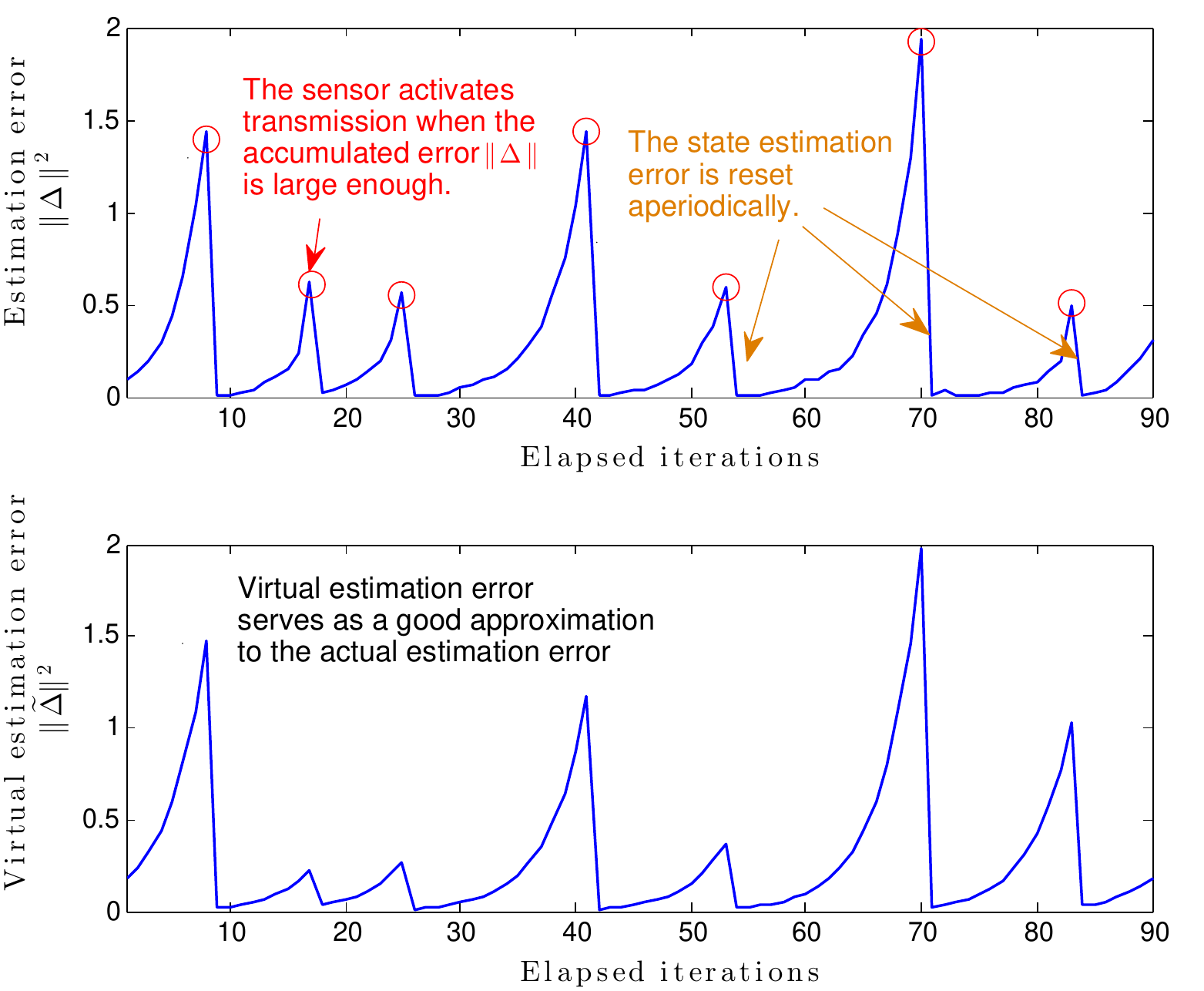}}
\caption{Illustrations of the evolutions of the  dynamic threshold, transitions between the active and dormant modes, and evolutions of the  state estimation error. The system parameters are configured as in Fig. \ref{decregion2} with $\widetilde{\mathbf{A}}=2\mathbf{I}$, $\bar{F}=1$ and $\lambda=100$.}
\label{approxqua}\vspace{-0.5cm}
\end{figure}

\begin{Algorithm}	\emph{(Dynamic MIMO \txtblue{AF} Precoding with Aperiodic Reset):}	\label{algo11}
\begin{figure}
\centering
  \includegraphics[width=3.2in]{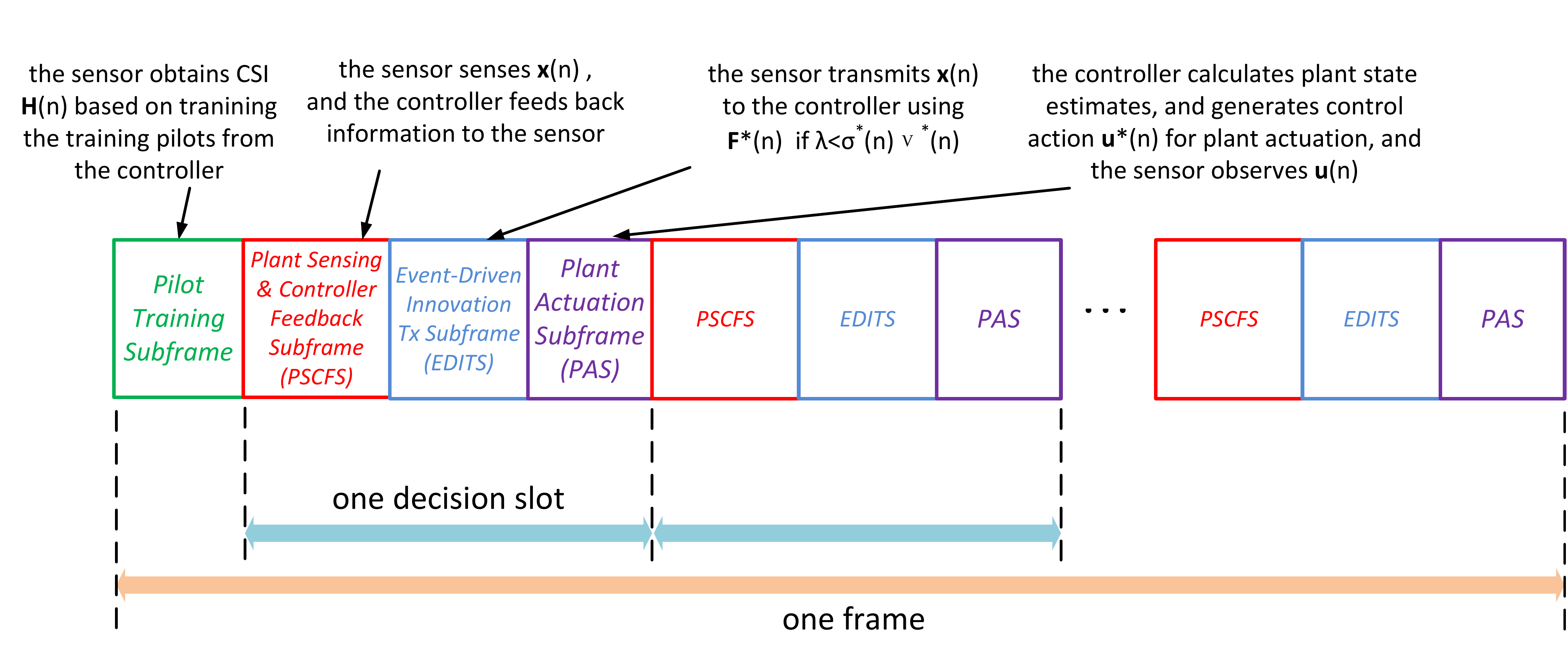}
  \caption{Illustrations of the frame structure.   The \txtblue{uplink training pilot from the controller} is transmitted at the beginning  of a frame once every coherence time, and the event-driven \txtblue{AF} precoding at the controller will be triggered at  every slot if $\lambda   < \sigma^\ast \nu^\ast$, where  $\sigma^\ast$ depends on the MIMO channel fading matrix and $\nu^\ast$ depends on the state estimation error $\boldsymbol{\Delta}$ and error covariance $\boldsymbol{\Sigma}$.}
  \label{evedriNSCasda1}\vspace{-0.8cm}
\end{figure}

	
	The time slots are  grouped into frames as illustrated in Fig. \ref{evedriNSCasda1}. \txtblue{The controller  transmits uplink training pilots at the beginning of a frame to the sensor, and  the sensor  estimates  the MIMO channel fading matrix $\mathbf{H}$.} At the beginning of the $n$-th slot,  
	\begin{itemize}
		\item	\textbf{Step 1 [Plant State Sensing of the Sensor and Information Feedback of the Controller]:}  The sensor  samples the plant state $\mathbf{x}(n)$.  If the rank of the feedback gain matrix  in (\ref{statsrr}) is less than  $L$ (i.e., $\text{rank}(\boldsymbol{{\Psi}})<L$), the controller will feed back a $(L-\text{rank}(\boldsymbol{\Psi}))$--dimensional  vector\footnote{specifically, $\widetilde{\mathbf{u}}^0(n-1)=\boldsymbol{\Psi}^0 \hat{\mathbf{x}}(n-1)$ where $\boldsymbol{\Psi}^0 \in \mathbb{R}^{(L-\text{rank}(\boldsymbol{\Psi}))\times L}$ and the rows of $\boldsymbol{\Psi}^0$ are the basis that spans the null space of $\boldsymbol{\Psi}$.} $\widetilde{\mathbf{u}}^0(n-1)$ to the sensor, which is a projection of $\hat{\mathbf{x}}(n-1)$ on the null space of $\boldsymbol{\Psi}$.  Otherwise, the controller does not need to feed back. 			
		\item	\textbf{Step 2 [Event-Driven \txtblue{AF} Precoding  and Plant State Transmission at the Sensor]: }  Based on ${\mathbf{u}}(n-1)$ from the plant and  the feedback $\widetilde{\mathbf{u}}^0(n-1)$  from the controller (if $\text{rank}(\boldsymbol{\Psi})<L$),   the sensor first calculates\footnote{Based on  $\mathbf{u}(n-1)$ and $\widetilde{\mathbf{u}}^0(n-1)$, the sensor first calculates $\hat{\mathbf{x}}(n-1)$. Then, it calculates $\boldsymbol{\Delta}(n-1)=\mathbf{x}(n-1)-\hat{\mathbf{x}}(n-1)$, $\boldsymbol{\Sigma}(n)$ using (\ref{kalman31}), which are use to further calculate  $\nu^\ast(n)$ in  Theorem \ref{thmpower}.} the dynamic threshold $\nu^\ast(n)$ according to Theorem \ref{thmpower}. If $\lambda   > \sigma^\ast(n) \nu^\ast(n)$, the sensor is in dormant mode at the current slot. Otherwise, the sensor calculates $\mathbf{F}^\ast(n)$ according to Theorem \ref{thmpower} and transmits the  $\mathbf{x}(n)$ using $\mathbf{F}^\ast(n)$.		
		\item	\textbf{Step 3 [Plant State Estimation and Plant Actuation]: } The controller calculates  the plant state estimate $\hat{\mathbf{x}}(n)$ based on the received signal $\mathbf{y}(n)$ and the local information, and generates plant control action $\mathbf{u}^\ast(n)$ according to (\ref{statsrr}). The actuator uses $\mathbf{u}^\ast(n)$ to drive the plant to a new state. The sensor observes the plant control action $\mathbf{u}^\ast(n)$.~\hfill~\IEEEQED
	\end{itemize}
\end{Algorithm}


Observe that when $\text{rank}(\boldsymbol{\Psi})<L$, the controller is required to feed back $\widetilde{\mathbf{u}}^0(n-1)$ to the sensor  every time slot. This is needed for the sensor to obtain $ \boldsymbol{\Delta}(n-1)$ in order to calculate the dynamic threshold $\nu^\ast(n)$. However, this feedback may be undesirable from the signaling overhead perspective. In fact, the sensor can approximate $\boldsymbol{\Delta}(n)$ using a \emph{virtual state estimation error} $\widetilde{\boldsymbol{\Delta}}(n)$ with the following dynamics:
\begin{align}	\label{dynvirtual}
	\widetilde{\boldsymbol{\Delta}}(n) =\big(\mathbf{I} - \mathbf{K}^a(n) \mathbf{E}^a(n)\big)\mathbf{A} \boldsymbol{\widetilde{\Delta}}(n-1) 
\end{align}
Note that the R.H.S. of (\ref{dynvirtual}) is the conditional mean drift of $\boldsymbol{\Delta}(n)$ in (\ref{kalman2}) and hence, $\widetilde{\boldsymbol{\Delta}}(n)$ tracks the mean of the actual $\boldsymbol{\Delta}(n)$. As a result, the sensor can use $\widetilde{\boldsymbol{\Delta}}(n)$ (which can be obtained locally at the sensor) instead of $\boldsymbol{\Delta}(n)$ to compute the MIMO \txtblue{AF} precoding $\mathbf{F}^\ast$ in Theorem \ref{thmpower}  as illustrated in Fig. \ref{approxqua}b, and no feedback from the controller is needed in Step 1 of Algorithm \ref{algo11}. 

\vspace{-0.3cm}

\subsection{Performance  Analysis}
We are interested to analyze the achievable system performance (MSE of the plant state estimation) using the proposed event-driven MIMO \txtblue{AF} precoding solution  $\mathbf{F}^\ast$ in Theorem \ref{thmpower}, and how the system parameters such as the \txtblue{maximum AF}  gain $\bar{F}$ and the average power price $\lambda$ affects the MSE. The result is summarized below.
\begin{Theorem}	[Achievable MSE under $\widetilde{\Omega}^\ast$]	\label{stab}
	For any given $\bar{F}>0$,  the MSE under $\widetilde{\Omega}^\ast$ is bounded, i.e., $\lim_{n \rightarrow \infty}\mathbb{E}^{\widetilde{\Omega}^\ast}\big[\left\|\boldsymbol{\Delta}(n)\right\|^2 \big]< \infty$. Furthermore, the MSE  satisfies:
	\begin{align}	\label{msebound}
		\mathbb{E}^{\widetilde{\Omega}^\ast}\left[\left\|\boldsymbol{\Delta} \right\|^2\right]\leq   \text{Tr}\big(   \mathbf{P} - \mathbf{P}G(\mathbf{P}, \bar{F}, \lambda)\mathbf{P} \big)  
	\end{align}
	where $G(\mathbf{P}, \bar{F}, \lambda)\triangleq \mathbb{E}\left[\int_{\lambda/\nu^\ast }^\infty \left(\frac{2\bar{F}x \mathbf{q}_1\mathbf{q}_1^T }{1+2\bar{F}x \mathbf{q}_1^TP \mathbf{q}_1}\right)f_{\sigma^\ast}(x)dx\right] $ and $f_{\sigma^\ast}(x)$ is the PDF of $\sigma^\ast$ (given in equ. (6) of \cite{eigenvaluesds}) and $\mathbf{P}$ satisfies the following fixed-point equation:
	\begin{align}	\label{fxptequ}
		\mathbf{P} = \mathbf{A}\left(\mathbf{P} - \mathbf{P}G(\mathbf{P}, \bar{F}, \lambda)\mathbf{P}\right)\mathbf{A}^T + \mathbf{W}
	\end{align}
\end{Theorem}

\begin{proof}
	please refer to Appendix G.
\end{proof}

\txtblue{Theorem \ref{stab} not only gives an upper bound of the MSE under $\widetilde{\Omega}^\ast$, but also leads to the result that $\widetilde{\Omega}^\ast$ is an admissible policy according to Definition \ref{admisscontrolpol}, as  shown below:
\begin{Corollary}	\emph{(Admissibility of the MIMO AF Precoding Policy $\widetilde{\Omega}^\ast$):}	\label{remarkwwewxx1}
	For any given $\bar{F}>0$, $\widetilde{\Omega}^\ast$ is an admissible policy according to Definition \ref{admisscontrolpol}. That is,  the plant state process  under $\widetilde{\Omega}^\ast$ and $\mathbf{u}^\ast$ in (\ref{statsrr}) is bounded, i.e.,  $\lim_{n \rightarrow \infty}\mathbb{E}^{\widetilde{\Omega}^\ast}\big[\left\|\mathbf{x}(n)\right\|^2 \big]< \infty$ under  $\mathbf{u}^\ast$ in (\ref{statsrr}).
\end{Corollary}
\begin{proof}
	please refer to Appendix H.
\end{proof}}

Therefore, $\bar{F}>0$ a sufficient condition for the stability of the  NCS and $\widetilde{\Omega}^\ast$ in Theorem \ref{thmpower} is an admissible policy (according to Definition \ref{admisscontrolpol}). In the following corollary, we discuss the impact of key system parameters on the MSE performance:
\begin{Corollary}	\emph{(Impact of System Parameters on MSE Performance):} \label{remarkwwewxx}	\
	\begin{itemize}
		\item	\textbf{MSE Upper Bound in (\ref{msebound}) vs Normalization Parameter $\bar{F}$:} The MSE upper bound in (\ref{msebound}) decreases at the order of $\mathcal{O}\left(\frac{1}{\bar{F}}\right)$ as $\bar{F}$ increases. 
		\item	\textbf{MSE Upper Bound in (\ref{msebound}) vs Communication Power Price $\lambda$:} The MSE upper bound in (\ref{msebound}) increases at the order of $\mathcal{O}\left(\frac{\exp(\lambda)}{\lambda^d}\right)$ (where $d\triangleq \min\{N_t, N_r\}$) as $\lambda$ increases.~\hfill~\IEEEQED
	\end{itemize}	
\end{Corollary}
\begin{proof}
	please refer to Appendix I.
\end{proof}

The above results illustrate that while $\bar{F}> 0$ is sufficient to maintain NCS stability, the reward of using a larger $\bar{F}$ is to further suppress the MSE at the order of   $\mathcal{O}\left(\frac{1}{\bar{F}}\right)$. On the other hand, the MSE increases exponentially fast as the average power price $\lambda$ increases. 
\txtblue{\begin{Remark}[Extension to Complex-Valued Plant State]
	Our proposed  solution framework  can be easily extended to the case with a complex-valued plant state.  Specifically,  the dynamics of the continuous-time stochastic  plant system before sampling is given by 
\begin{align}	\label{cmasd1}
	\dot{\mathbf{x}}(n)=\widetilde{\mathbf{A}}\mathbf{x}(n)+\widetilde{\mathbf{B}}\mathbf{u}(n)+\widetilde{\mathbf{w}}(n)
\end{align}
where $\mathbf{x}(n) \in \mathbb{C}^{L \times 1}$ is the plant state process, $\mathbf{u}(n) \in \mathbb{C}^{M \times 1}$ is the plant control action, $\widetilde{\mathbf{A}}\in \mathbb{C}^{L\times L}$, $\widetilde{\mathbf{B}} \in \mathbb{C}^{L\times M}$, and $\widetilde{\mathbf{w}}(n) \sim \mathcal{CN}(0, \widetilde{\mathbf{W}})$ is an additive plant disturbance with  zero mean and covariance $\widetilde{\mathbf{W}} \in \mathbb{R}^{L\times L}$.  Similarly, we can obtain the optimal CE controller using the no dual effect property as in Lemma \ref{lemmadual}. We then  formulate a MIMO precoding  AF optimization problem as follows:
\bs\begin{align} 
	\min_{\Omega}	\quad  & \limsup_{N \rightarrow \infty} \frac{1}{N}\mathbb{E}^{\Omega}\Bigg[  \sum_{n=0}^{N-1} \bigg(\left(\boldsymbol{\Delta}^a\right)^\dagger(n)\mathbf{S}^a \boldsymbol{\Delta}^a(n)\notag \\
&	\hspace{1cm}+ \lambda  \text{Tr}\left(\mathbf{F}^\dagger(n) \mathbf{F}(n)\right)\bigg)\tau\Bigg] \notag 	
\end{align}\bsc
where\footnote{\txtblue{Note that the squared estimation error in the  per-stage cost can be written in an equivalent form as $\boldsymbol{\Delta}^\dagger(n)\mathbf{S} \boldsymbol{\Delta}(n)=\frac{1}{2}\left(\boldsymbol{\Delta}^a\right)^\dagger(n)\mathbf{S}^a \boldsymbol{\Delta}^a(n)$ for all $n$.}}  $\boldsymbol{\Delta}^a(n)=\left[ \begin{smallmatrix} 
\boldsymbol{\Delta}(n) \\
 \boldsymbol{\Delta}^\ddagger(n)
\end{smallmatrix} \right]$ is an augmented $2L \times 1$ plant state error and $\mathbf{S}^a=\diag\left(\mathbf{S}, \mathbf{S}^\ddagger\right)$. The system state is  $\boldsymbol{\chi}(n)\triangleq \big\{\boldsymbol{\Delta}^a(n-1), \boldsymbol{\Sigma}^a(n), \mathbf{H}(n)\big\}$, where  $\boldsymbol{\Sigma}^a(n)=\mathbb{E}\big[\big(\mathbf{x}^a(n)-\hat{\mathbf{x}}^{a-}(n)\big]\big)\big(\mathbf{x}^a(n)-\hat{\mathbf{x}}^{a-}(n)\big)^\dagger \big|I_C(n-1) \big]$ is the one-step state prediction error covariance and $\hat{\mathbf{x}}^{a-}(n)\triangleq \mathbb{E}\big[\mathbf{x}^a(n)\big|I_C(n-1)\big]$ is the one-step plant state prediction. The state dynamic of  $\boldsymbol{\chi}(n)$ is given by
\begin{align}
	\boldsymbol{\Delta}^a(n) 	= & \big(\mathbf{I} - \mathbf{K}^a(n) \mathbf{E}^a(n)\big)\big(\mathbf{A}^a\boldsymbol{\Delta}^a(n-1)+ \mathbf{w}^a(n-1)\big)\notag\\
	& - \mathbf{K}^a(n)\mathbf{z}^a(n)		\notag	\\
	\boldsymbol{\Sigma}^a(n+1) = & \mathbf{A}^a\big(\boldsymbol{\Sigma}^a (n) - \boldsymbol{\Sigma}^a (n)(\mathbf{E}^a (n))^\dagger \big(\mathbf{E}^a (n) \boldsymbol{\Sigma}^a (n) \notag \\
	&\cdot  (\mathbf{E}^a (n))^\dagger  \left(\mathbf{A}^a\right)^\dagger+\mathbf{I}\big)^{-1}\mathbf{E}^a (n) \boldsymbol{\Sigma}^a (n)  \big)  + \mathbf{W}^a \notag
\end{align}
with initial conditions $\boldsymbol{\Delta}^a(n)=\mathbf{0}$ and $\boldsymbol{\Sigma}^a(0)=\mathbf{0}$, where $\mathbf{K}^a(n) =\boldsymbol{\Sigma}^a (n) (\mathbf{E}^a (n))^\dagger \big(\mathbf{E}^a (n) \boldsymbol{\Sigma}^a (n)    (\mathbf{E}^a (n))^\dagger    +\mathbf{I}\big)^{-1}$ is the Kalman gain, $\mathbf{E}^a(n)=\diag\left(\mathbf{E}(n), \mathbf{E}^\ddagger(n)\right)$, $\mathbf{A}^a=\diag\left(\mathbf{A}, \mathbf{A}^\ddagger\right)$, and $\mathbf{W}^a=\diag\left(\mathbf{W}, \mathbf{W}\right)$. Using the calculations for solving the PDE as in Lemma 4, we can  obtain the associated  closed-form  priority function and then obtain the  optimal event-driven MIMO AF precoding solution as in Theorem \ref{thmpower}, which is adaptive to the plant-related states $\big(\boldsymbol{\Delta}^a, \boldsymbol{\Sigma}^a\big)$ and CSI $\mathbf{H}$.~\hfill~\IEEEQED\end{Remark}}

\vspace{-0.2cm}

\section{Simulations}\label{simsection}

\begin{figure}[t]
  \centering
  \includegraphics[width=2.6in]{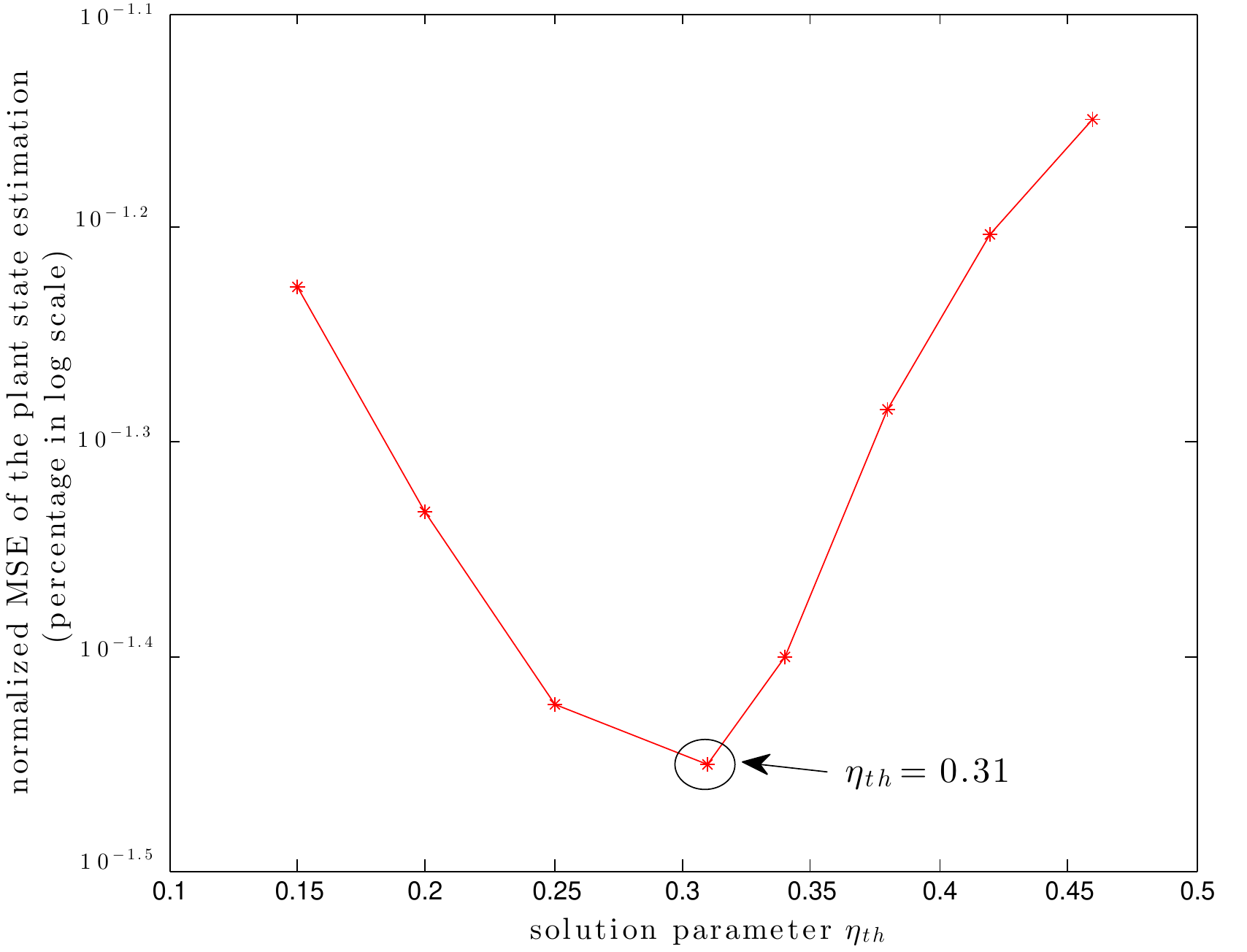}
  \caption{Normalized  MSE of the plant state estimation  versus $\eta_{th}$ under the MIMO \txtblue{AF} precoding  scheme  in Algorithm \ref{algo11}  at $\bar{F}=2$ and $\lambda=1500$.}\vspace{-0.2cm}
  \label{compareees2} \vspace{-0.5cm}
\end{figure}

In this section, we compare the  performance  of the proposed   MIMO \txtblue{AF} precoding  scheme with the following  four baselines. Baseline 1 refers to \emph{MIMO \txtblue{AF} precoding with  equal power across data streams (AP-EPDS)} \cite{baseline1}, where  $F=\sqrt{\frac{\bar{F}}{L}}\mathbf{U} \widetilde{\boldsymbol{\Upsilon}}$ and the $(l,l)$-th  element  in $\widetilde{\boldsymbol{\Upsilon}}$ is one for all $l=1,\dots, L$ and the other elements are zero.  Baseline 2 refers to \emph{MIMO \txtblue{AF}  precoding  for error-free channel (AP-EFC)} \cite{onoff2},  where   the sensor at each time slot determines whether to transmit  by minimizing the average  weighted state estimation error and the average number of channel uses, and  adopts the BF-EPDS if it transmits. Baseline 3 refers to \emph{MIMO \txtblue{AF}  precoding for SISO packet-dropout channel with special information structure (AP-SPSIS)} \cite{mdp1}, where    the sensor at each time slot determines whether to transmit  by minimizing the average weighted state estimation error and the average power cost, and  adopts the BF-EPDS if it transmits. The power action  depends on $\boldsymbol{\Theta}(n)$ and the CSI, where $\boldsymbol{\Theta}(n)\triangleq \mathbf{A}\boldsymbol{\Delta}(n-1)+\mathbf{w}(n-1)$. The solutions for Baseline 2 and  3 are obtained using the brute-force VIA. Baseline 4 refers to \emph{dynamic MIMO AF  precoding using approximate dynamic programming (DAP-ADP)} \cite{adp}, \cite{adpp}. We consider the quadratic  approximation of the  priority function $\widetilde{V}_r=r_1 \boldsymbol{\Delta}^T\boldsymbol{\Sigma}\boldsymbol{\Delta}+\mathbf{r}_2^T \boldsymbol{\Delta}$ (or equivalently $\nabla_{\boldsymbol{\Delta}}\widetilde{V}_r=r_1\boldsymbol{\Sigma}\boldsymbol{\Delta}+\mathbf{r}_2$), where $r\in \mathbb{R}$ and $\mathbf{r}_2 \in  \mathbb{R}^{L\times 1}$ are tunable parameters and  $ \boldsymbol{\Delta}^T\boldsymbol{\Sigma}\boldsymbol{\Delta}$ and $\boldsymbol{\Delta}$ are basis functions in the ADP. We adopt the average cost temporal-difference iteration learning  algorithm \cite{adp}, \cite{adpp}   to update $(r_1,\boldsymbol{r}_2)$ at each time slot. The MIMO AF precoding solution under ADP is similar to that in Theorem \ref{thmpower} with $V$ replaced by $\widetilde{V}_r$. We consider a MIMO NCS with parameters: $\mathbf{\widetilde{\mathbf{A}}}=\left( \begin{smallmatrix} 
  1  & 2\\
  -1 & 3 
\end{smallmatrix} \right)$, $\widetilde{\mathbf{B}}=\left( \begin{smallmatrix} 
  1 & 0.2\\
  0.1 & 1 
\end{smallmatrix} \right)$, $\widetilde{\mathbf{W}}=\text{diag}(1,2)$, $\mathbf{Q}=\text{diag}(1,2)$, $\mathbf{R}=\text{diag}(1,0.2)$,  $N_t=3$, $N_r=2$, and $\tau=0.05$s.

\begin{figure}
\begin{minipage}[t]{0.45\textwidth}
 \hspace{0.65cm}
  \includegraphics[width=2.6in]{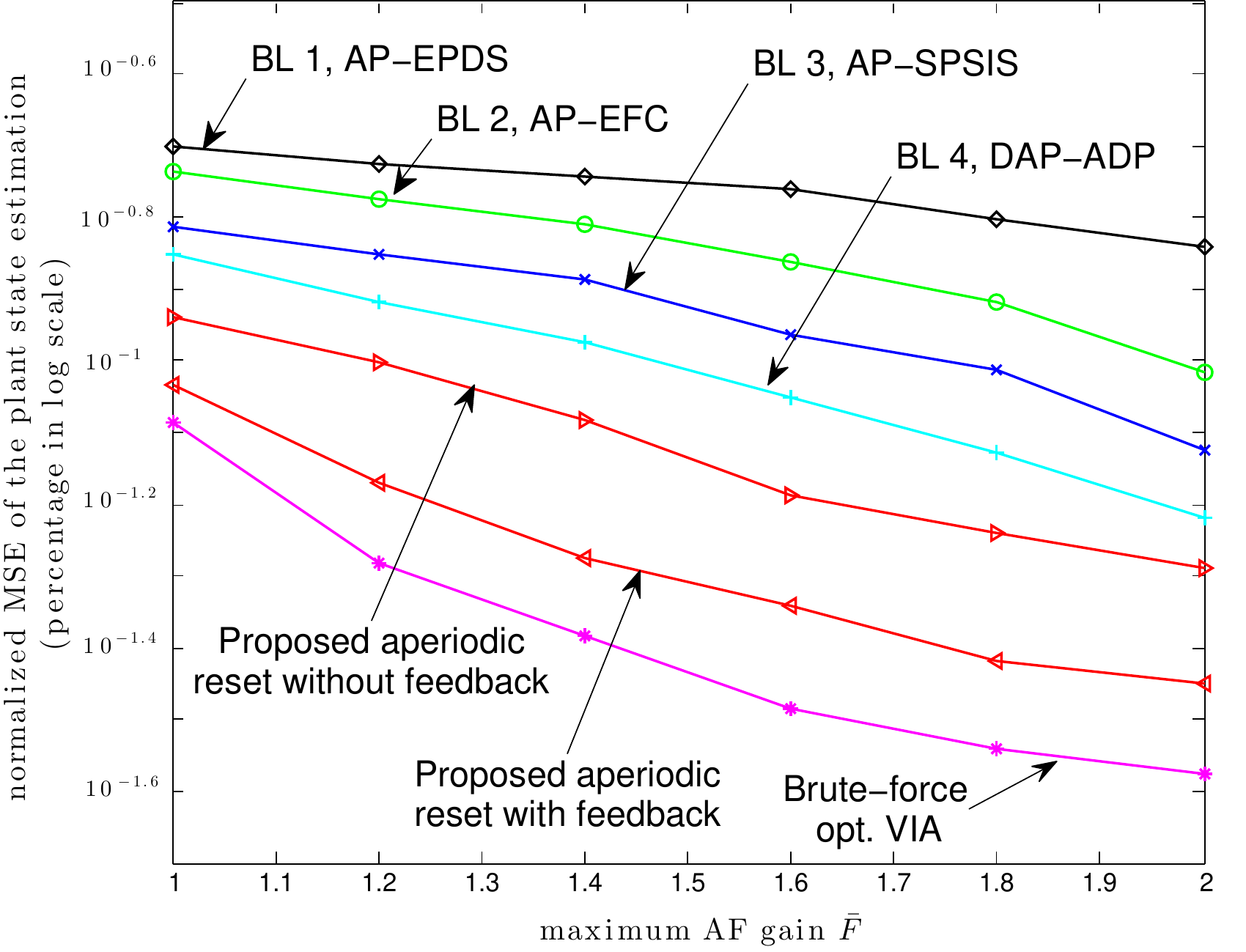}
  \caption{\txtblue{Normalized MSE of the plant state estimation  versus maximum AF  gain $\bar{F}$ at $\lambda=1500$.}}
  \label{comparsdeees2}
\end{minipage}
\hspace{0.03\textwidth}
\begin{minipage}[t]{0.45\textwidth}
  \centering
   \includegraphics[width=2.6in]{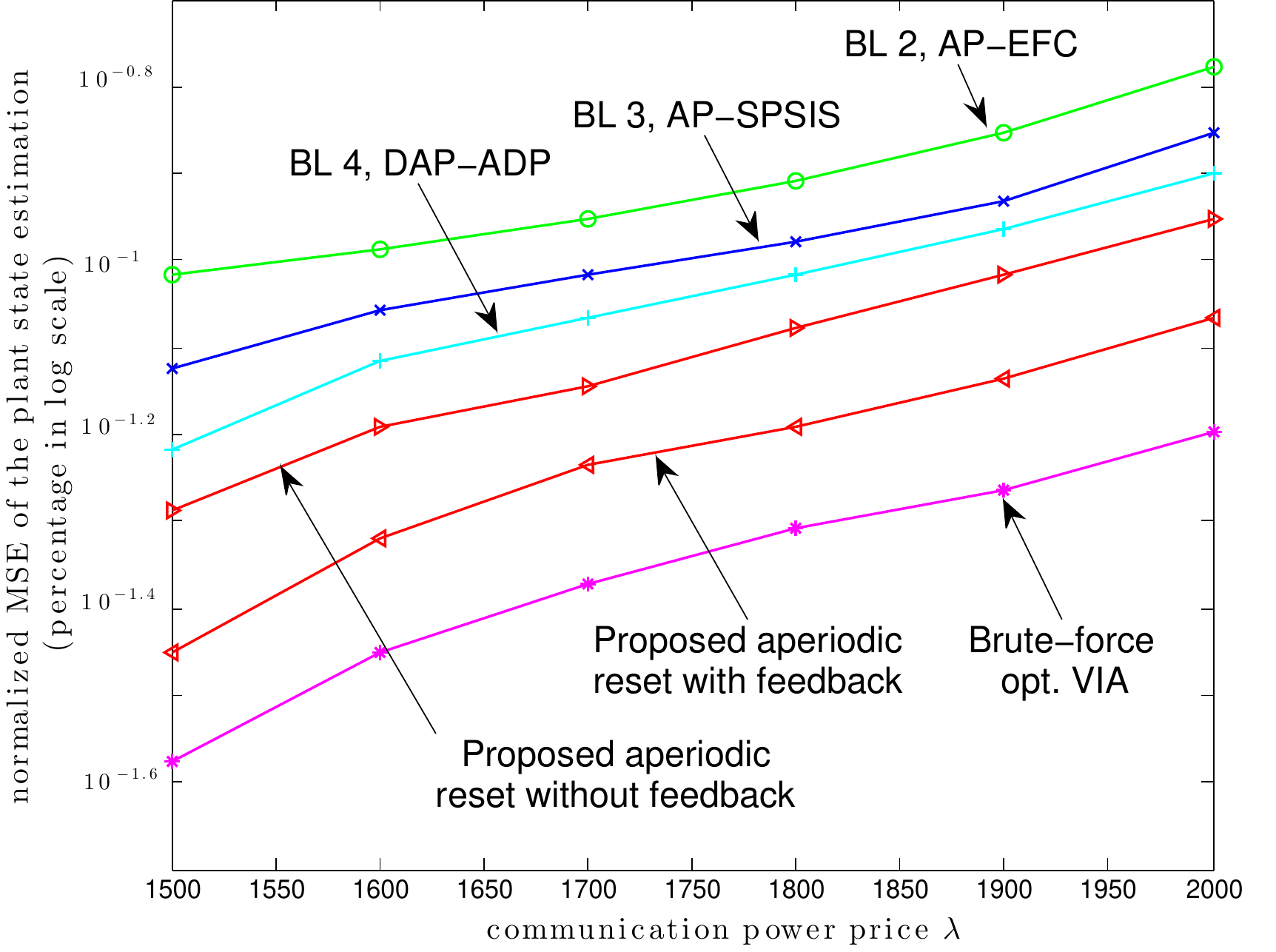}
  \caption{\txtblue{Normalized MSE of the plant state estimation  versus communication power price $\lambda$ at $\bar{F}=2$.}}
  \label{asddsadsd}\end{minipage}	\vspace{-0.5cm}
\end{figure}

\begin{table*}
	\centering
\begin{tabular}{|l | c c c c c c c c c c|}
	\hline
		  {\txtblue{$\lambda$}} & \multicolumn{1}{c|}{\txtblue{400}}  & \multicolumn{1}{c|}{\txtblue{800}}  & \multicolumn{1}{c|}{\txtblue{1000}} & \multicolumn{1}{c|}{\txtblue{1500}}  & \multicolumn{1}{c|}{\txtblue{2000}}  & \multicolumn{1}{c|}{\txtblue{3000}} & \multicolumn{1}{c|}{\txtblue{4000}}    & \multicolumn{1}{c|}{\txtblue{6000}}    \\
	\hline
		{\txtblue{Avg.   Pow. Gain Cost}}&   \multicolumn{1}{c|}{ \txtblue{0.1173}} &   \multicolumn{1}{c|}{ \txtblue{0.1088}} &   \multicolumn{1}{c|}{\txtblue{0.0920}} &    \multicolumn{1}{c|}{\txtblue{0.0873}}  &    \multicolumn{1}{c|}{\txtblue{0.0846}}    &    \multicolumn{1}{c|}{\txtblue{0.0817}}  &    \multicolumn{1}{c|}{\txtblue{0.0803}}   &    \multicolumn{1}{c|}{\txtblue{0.0784}}     	 \\ \hline
		 {\txtblue{Avg.   Abs. Pow. Cost (W)}}  &   \multicolumn{1}{c|}{ \txtblue{2.6306}} &   \multicolumn{1}{c|}{\txtblue{3.0451}} &   \multicolumn{1}{c|}{\txtblue{4.8488}} &    \multicolumn{1}{c|}{\txtblue{5.6796}}  &    \multicolumn{1}{c|}{\txtblue{6.3459}}    &    \multicolumn{1}{c|}{\txtblue{7.2648}}&    \multicolumn{1}{c|}{\txtblue{8.0645}}    &    \multicolumn{1}{c|}{\txtblue{9.1787}}   	 \\ \hline
	         \end{tabular}
	\caption{ \txtblue{Average  power gain cost  and average absolute power cost  under various  communication  power prices $\lambda$ at  $\tau=0.05$s and $\bar{F}=1$.}}
	\label{oneone}	\vspace{-0.5cm}
\end{table*}

\vspace{-0.3cm}
\subsection{Choice of the Solution Parameter $\eta_{th}$ in (\ref{approxvaluefunc})} \label{labelsdsa}

Fig.~\ref{compareees2} illustrates the normalized  MSE of the plant state estimation  versus different values of  $\eta_{th}$ under the MIMO \txtblue{AF} precoding  scheme  in Algorithm \ref{algo11} at  a \txtblue{maximum AF}  gain  $\bar{F}=2$ and communication power price $\lambda=1500$.  It can be observed that the average  normalized MSE achieves the minimum when $\eta_{th}$ is around 0.31. Therefore, we choose $\eta_{th}= 0.31$  when $\bar{F}=2$ and $\lambda=1500$. The optimal choices of $\eta_{th}$ at other \txtblue{maximum AF}  gains and communication power prices  can be obtained using similar methods.

\vspace{-0.5cm}

\subsection{Performance Comparisons}
Fig. \ref{comparsdeees2} illustrates the normalized MSE of the plant state estimation  versus  the \txtblue{maximum AF}  gain   $\bar{F}$ at communication power price $\lambda=1500$. It can be observed that there is significant performance gain of the proposed schemes (aperiodic reset with and without controller feedback) compared with all the baselines. This gain is contributed by the plant state and CSI adaptive dynamic MIMO \txtblue{AF} precoding. Furthermore,  the performance of the proposed scheme (aperiodic reset with controller feedback) is very close to that of the brute-force optimal  VIA  \cite{mdpcref2}.  Fig. \ref{asddsadsd} illustrates the normalized MSE of the plant state estimation  versus communication power price $\lambda$ at \txtblue{maximum AF}  gain $\bar{F}=2$. It can be observed that there is significant performance gain of the proposed schemes compared with all the baselines across a wide range of $\lambda$. \txtblue{Table \ref{oneone} illustrates the one-to-one association of the power price $\lambda$ and the absolute average power cost.}\vspace{-0.3cm}
\subsection{\txtblue{Comparison with the Brute-Force Optimal VIA}}
\txtblue{We evaluate the performance loss of our proposed closed-form MIMO  \txtblue{AF} precoding  policy $\widetilde{\Omega}^\ast$ with the optimal brute-force VIA \cite{mdpcref2} for solving Problem \ref{probformu}. Specifically, we focus on the normalized MSE performance under different power prices  and  the performance loss is defined as follows:
	\bs\begin{align}
		&Perf. \ Loss \\
		= &\frac{(Perf. \ under \ \widetilde{\Omega}^\ast) -(Optimal \ Perf. \ using \ VIA)}{Optimal \ Perf. \ using \ VIA}\notag
	\end{align}\bsc 	\quad We illustrate the performance loss results in Table \ref{tables1} and Table \ref{tables2}. Specifically, Table \ref{tables1} shows the performance loss under various  maximum AF  gains  $\bar{F}$ at power price $\lambda=1500$. It can be observed that the performance loss values are under 3\% under various   maximum AF  gains $\bar{F}$. Table \ref{tables2} shows the performance loss under various communication  power prices $\lambda$ at  maximum AF  gain  $\bar{F}=2$. The performance loss values are under 4\% under various   power prices $\lambda$. Therefore, based on the above numerical results, our proposed  closed-form MIMO \txtblue{AF} precoding  solution achieves a very low performance loss under various system parameter settings.}

\vspace{-0.3cm}
\subsection{Complexity Comparisons}
Table \ref{tables} illustrates the comparison of the MATLAB computational time of the baselines, the proposed schemes, and the brute-force VIA \cite{mdpcref2}.  The computational time of Baseline 1 is the smallest in all different  scenarios, but it has very poor performance.  The computational cost of our proposed schemes  is much smaller than those of Baseline 2--4, due to the closed-form approximate priority function.  Furthermore, our schemes  outperform  baselines 2--4.

\begin{table}
	\centering
\begin{tabular}{|l| c c c c c c |}
	\hline
		  {\txtblue{$\bar{F}$}}& \multicolumn{1}{c|}{\txtblue{1}}  & \multicolumn{1}{c|}{\txtblue{1.2}} & \multicolumn{1}{c|}{\txtblue{1.4}}  & \multicolumn{1}{c|}{\txtblue{1.6}} & \multicolumn{1}{c|}{\txtblue{1.8}} & \multicolumn{1}{c|}{\txtblue{2}}    \\
	\hline
		 {\txtblue{Perf. Loss}} &   \multicolumn{1}{c|}{\txtblue{2.48\%}} &   \multicolumn{1}{c|}{\txtblue{3.55\%}} &    \multicolumn{1}{c|}{\txtblue{2.63\%}}  &    \multicolumn{1}{c|}{\txtblue{2.97\%}}   &    \multicolumn{1}{c|}{\txtblue{2.21\%}}   &    \multicolumn{1}{c|}{\txtblue{2.07\%}}   	 \\ \hline
	         \end{tabular}
	\caption{ \txtblue{Performance loss under various   maximum AF gains $\bar{F}$ at    $\tau=0.05$s and $\lambda=1500$.}}	
		\label{tables1}\vspace{-0.2cm}
\end{table}
\begin{table}
	\centering
\begin{tabular}{|l | c c c c c c |}
	\hline
		  {\txtblue{$\lambda$}}& \multicolumn{1}{c|}{\txtblue{1500}}  & \multicolumn{1}{c|}{\txtblue{1600}} & \multicolumn{1}{c|}{\txtblue{1700}}  & \multicolumn{1}{c|}{\txtblue{1800}} & \multicolumn{1}{c|}{\txtblue{1900}} & \multicolumn{1}{c|}{\txtblue{2000}}    \\
	\hline
		 {\txtblue{Perf. Loss}} &   \multicolumn{1}{c|}{\txtblue{2.07\%}} &   \multicolumn{1}{c|}{\txtblue{2.90\%}} &    \multicolumn{1}{c|}{\txtblue{3.63\%}}  &    \multicolumn{1}{c|}{\txtblue{3.57\%}}   &    \multicolumn{1}{c|}{\txtblue{3.50\%}}   &    \multicolumn{1}{c|}{\txtblue{3.40\%}}   	 \\ \hline
	         \end{tabular}
	\caption{ \txtblue{Performance loss under various communication  power prices $\lambda$ at  $\tau=0.05$s and $\bar{F}=2$.}}	
		\label{tables2}\vspace{-0.5cm}
\end{table}

\begin{table}[t]
	\hspace{-0.3cm}
\begin{tabular}{|l | c c c c |}
	\hline
		  {Dimension of $\mathbf{x}$}& \multicolumn{1}{c|}{2}  & \multicolumn{1}{c|}{4} & \multicolumn{1}{c|}{6}  & \multicolumn{1}{c|}{8}    \\
	\hline
		 {Baseline 1, AP-EPDS} &  \multicolumn{4}{c|}{0.0004ms}   			\\ \hline
		 {Baseline 2, AP-EFC} &   \multicolumn{1}{c|}{2.43s} &   \multicolumn{1}{c|}{3860.2s} &    \multicolumn{2}{c|}{$> 10^6$s}   	 \\ \hline
	            {Baseline 3, AP-SPSIS} &   \multicolumn{1}{c|}{3.12s} &  \multicolumn{1}{c|}{4922.9s} &    \multicolumn{2}{c|}{$> 10^6$s} 	\\	\hline
	             \txtblue{Baseline 4, DAP-ADP} &   \multicolumn{1}{c|}{\txtblue{0.0768s}} &  \multicolumn{1}{c|}{\txtblue{0.1429s}}   &    \multicolumn{1}{c|}{\txtblue{0.3872s}} &     \multicolumn{1}{c|}{\txtblue{0.9104s}}	\\	\hline
	            {Proposed Schemes} &     \multicolumn{1}{c|}{ } &   \multicolumn{1}{c|}{ }  &    \multicolumn{1}{c|}{ } &     \multicolumn{1}{c|}{}	\\
	            {(aperiodic reset with and} &     \multicolumn{1}{c|}{$0.0012$s} &   \multicolumn{1}{c|}{$0.0018$s}  &    \multicolumn{1}{c|}{$0.0054$s} &     \multicolumn{1}{c|}{$0.088$s}	\\
	            {without controller feedback)} &     \multicolumn{1}{c|}{ } &   \multicolumn{1}{c|}{ }  &    \multicolumn{1}{c|}{ } &     \multicolumn{1}{c|}{ }	\\\hline
	             {Brute-force opt. VIA} &     \multicolumn{1}{c|}{220.5s} &   \multicolumn{3}{c|} {$> 10^6$s} 	\\
	\hline
\end{tabular}
	\caption{ {Comparison of the MATLAB computational time of the baselines, the proposed algorithm, and the brute-force optimal VIA in one decision slot.}}	
		\label{tables}\vspace{-0.5cm}
\end{table}

\vspace{-0.3cm}

\section{Summary}
In this paper, we propose a \txtblue{closed-form first-order optimal  MIMO AF  precoding solution} for the MIMO NCS by solving a weighted  average state estimation error at the remote controller subject to an average communication power gain constraint of the sensor. Using a continuous-time perturbation approach, we derive a closed-form approximate priority   function and a closed-form MIMO \txtblue{AF} precoding scheme.  The proposed  MIMO \txtblue{AF} precoding  solution   is shown to have  an event-driven control  structure. We  also give sufficient conditions for ensuring the NCS stability. Numerical results show that the proposed schemes have low complexity and much better performance compared with  the baselines.
\txtblue{\section*{Appendix A: Proof of Lemma \ref{lemmadual}} 
\vspace{0.1cm}
\hspace{-0.3cm} \emph{A. Relationship between the Original NCS and an Autonomous NCS} 
\vspace{0.1cm}}

\txtblue{We consider two NCSs. The first NCS is given as follows  for given control actions $\mathbf{u}_0^n$:
\begin{align}
	&\mathbf{x}(n+1)=\mathbf{A}\mathbf{x}(n)+\mathbf{B}\mathbf{u}(n)+\mathbf{w}(n)\notag \\
	&\mathbf{y}(n)=   \mathbf{H}(n)\mathbf{F}(n)\mathbf{x}\left(n\right)+ \mathbf{z}(n)	\label{ncs11}
\end{align}
The second NCS  is given as follows with no control actions applied (i.e., an autonomous system):
\begin{align}
	&\overline{\mathbf{x}}(n+1)=\mathbf{A}\overline{\mathbf{x}}(n)+\overline{\mathbf{w}}(n)\notag \\
	&\overline{\mathbf{y}}(n)=   \overline{\mathbf{H}}(n)\overline{\mathbf{F}}(n)\mathbf{x}\left(n\right)+ \overline{\mathbf{z}}(n)	\label{ncs22}
\end{align}
where   $\overline{\boldsymbol{\Delta}}(n)=\overline{\mathbf{x}}(n)-\hat{\overline{\mathbf{x}}}(n)$,  $\hat{\overline{\mathbf{x}}}(n)=\mathbb{E}\big[\overline{\mathbf{x}}(n)\big|\overline{I}_C(n)\big]$,  $\overline{I}_C(n)=\big\{\overline{\mathbf{E}}_0^{n}, \overline{\mathbf{y}}_0^{n}\big\}$  and  $\overline{\mathbf{E}}(n) = \overline{\mathbf{H}}(n)\overline{\mathbf{F}}(n) $. Furthermore, define $\overline{I}_S(n)=\big\{\overline{\mathbf{x}}_0, \overline{\mathbf{w}}_0^{n-1}, \overline{\mathbf{H}}_0^{n},  \overline{\mathbf{z}}_0^{n-1}\big\}$. We let  the initial conditions, system disturbances,  CSI, and channel noise   be identical in the two NCSs, i.e.,  $\mathbf{x}(0)=\overline{\mathbf{x}}(0)$, $\mathbf{w}(n) = \overline{\mathbf{w}}(n)$, $\mathbf{H}(n) = \overline{\mathbf{H}}(n)$, and $\mathbf{z}(n) = \overline{\mathbf{z}}(n)$, and assume that the two NSCs adopt the same MIMO AF precoding policy.  Then, we establish the following lemma:
\begin{Lemma}		\label{lemma6}
	For the two NCSs in (\ref{ncs11}) and (\ref{ncs22}), we have $\mathbf{x}(n)-\mathbb{E}\big[\mathbf{x}(n)|I_C(n)\big] =\overline{\mathbf{x}}(n)-\mathbb{E}\big[\overline{\mathbf{x}}(n)|I_C(n)\big]$.
\end{Lemma}
\begin{proof}
	Note that the linearity of the state dynamics for $\mathbf{x}(n)$ and $\overline{\mathbf{x}}(n)$ implies  the existence of matrices $\mathbf{J}(n)$, $\mathbf{K}(n)$ and $\mathbf{L}(n)$ such that
\begin{align}
	&\mathbf{x}(n)=\mathbf{J}(n)\mathbf{x}(0)+\mathbf{K}(n)\vec{\mathbf{u}}(n-1)+\mathbf{L}(n)\vec{\mathbf{w}}(n-1) \notag \\
	&\overline{\mathbf{x}}(n)=\mathbf{J}(n)\mathbf{x}(0)+\mathbf{L}(n)\vec{\mathbf{w}}(n-1)
\end{align}
where $\vec{\mathbf{u}}(n)=\big(\mathbf{u}^T(1), \dots, \mathbf{u}(n)\big)^T$ and $\vec{\mathbf{w}}(n)=\big(\mathbf{w}^T(1), \dots, \mathbf{w}(n)\big)^T$.  Then, we have
\begin{align}
	&\mathbf{x}(n)-\mathbb{E}\big[\mathbf{x}(n)|I_C(n)\big] \\
	=& \big(\mathbf{J}(n)\mathbf{x}(0)+\mathbf{K}(n)\vec{\mathbf{u}}(n-1)+\mathbf{L}(n)\vec{\mathbf{w}}(n-1)\big)\notag \\
	&-\big(\mathbf{J}(n)\mathbb{E}\big[\mathbf{x}(0)|I_C(n)\big]+\mathbf{K}(n)\vec{\mathbf{u}}(n-1)\notag \\
	&+\mathbf{L}(n)\mathbb{E}\big[\vec{\mathbf{w}}(n-1)|I_C(n)\big]\big)=\overline{\mathbf{x}}(n)-\mathbb{E}\big[\overline{\mathbf{x}}(n)|I_C(n)\big]\notag 
\end{align}
\end{proof}}
\vspace{0.1cm}
\hspace{-0.3cm} \emph{\txtblue{B. State Estimate of an Autonomous System}}
\vspace{0.1cm}

\txtblue{Since $\mathbf{F}(n)$ ($\overline{\mathbf{F}}(n)$) is a function of $I_S(n)$ ($\overline{I}_S(n)$) and both NCSs adopt the same MIMO AF precoding policy, we have $\mathbf{F}(n)=\overline{\mathbf{F}}(n)$. Furthermore, we have $\mathbf{H}(n)=\overline{\mathbf{H}}(n)$. Therefore, we have $\mathbf{E}(n)=\overline{\mathbf{E}}(n)$. From $\mathbf{y}(n)$ and $\overline{\mathbf{y}}(n)$ in (\ref{ncs11}) and (\ref{ncs22}), we know that 
\begin{align}
	\overline{\mathbf{y}}(n) = \mathbf{y}(n) - \mathbf{P}_k\left(\mathbf{E}_0^n\right)\vec{\mathbf{u}}(n)
\end{align}
for some matrix $\mathbf{P}_k$ that depends on $\mathbf{E}_0^n$. The above equation implies that there is a bijective relationship between $\overline{\mathbf{y}}(n)$ and $ \mathbf{y}(n)$. Therefore, given $\mathbf{E}_0^n$, the information provided by $I_C(n)$ regarding $\overline{\mathbf{x}}(n)$ is summarized in $\overline{\mathbf{y}}(n)$ (see \cite{asdacite}, Lemma 5.2.1 of \cite{mdpcref2}). Therefore, we have 
\begin{align}
	\mathbb{E}\big[\overline{\mathbf{x}}(n)|I_C(n)\big]= \mathbb{E}\big[\overline{\mathbf{x}}(n)|\overline{I}_C(n)\big] \end{align}
Furthermore, using $\overline{I}_C(n)=\sigma\left(\overline{I}_S(n)\right)=\sigma\left(I_S(n)\right)$ and the above equation, we have $\mathbf{x}(n)-\mathbb{E}\big[\mathbf{x}(n)|I_C(n)\big] =\overline{\mathbf{x}}(n)-\mathbb{E}\big[\overline{\mathbf{x}}(n)|I_S(n)\big]$.  Therefore, $\boldsymbol{\Delta}(n)=\mathbf{x}(n)-\mathbb{E}\big[\mathbf{x}(n)|I_C(n)\big]$ only depends on $I_S(n)$, which directly proves  the no dual effect property in (\ref{nodualeffecrdef}).}

\section*{Appendix B: Dynamics of the State Estimator  and the State Estimation Error} 

%
%

 We adopt the  augmented complex Kalman filter (ACKF) algorithm in \cite{ackf},  which is the   minimum MSE estimator  for complex-valued  measurement (which is $\mathbf{y}$ in our problem). Specifically,  $\hat{\mathbf{x}}(n)$ follows the following  Kalman filter equation:
\begin{align}
	\hat{\mathbf{x}}(n) =& \mathbf{A} \hat{\mathbf{x}}(n-1) + \mathbf{B} \mathbf{u}(n-1) + \mathbf{K}^a(n) \label{kalman1} \\
	&\big(\mathbf{y}^a (n)-  \mathbf{E}^a(n) \big(\mathbf{A} \hat{\mathbf{x}}(n-1) + \mathbf{B} \mathbf{u}(n-1)-\widetilde{\mathbf{x}}(n)  \big)\big)	\notag
\end{align}
with initial value $\hat{\mathbf{x}}(0)=\mathbf{x}_0$, where $\mathbf{y}^a(n) =\left( \begin{smallmatrix} 
  \mathbf{y}(n) \\
  \mathbf{y}^\ddagger (n)
\end{smallmatrix} \right)$ is an augmented  $2 N_r \times 1$ vector, $\mathbf{E}^a(n)=\left( \begin{smallmatrix} 
  \mathbf{E}(n) \\
  \mathbf{E}^\ddagger(n)
\end{smallmatrix} \right)$ is an augmented $2 N_r \times L$ matrix. Furthermore, 
\begin{align}
	 & \mathbf{K}^a(n) = \boldsymbol{\Sigma} (n) (\mathbf{E}^a (n))^\dagger \big(\mathbf{E}^a (n) \boldsymbol{\Sigma} (n)   (\mathbf{E}^a (n))^\dagger +\mathbf{I}\big)^{-1}		\label{kalman21}	\\
	 & \boldsymbol{\Sigma}(n+1) = A\big(\boldsymbol{\Sigma} (n) - \boldsymbol{\Sigma} (n)(\mathbf{E}^a (n))^\dagger \big(\mathbf{E}^a (n) \boldsymbol{\Sigma} (n)   (\mathbf{E}^a (n))^\dagger \notag \\
	& \hspace{3cm}+\mathbf{I}\big)^{-1}\mathbf{E}^a (n) \boldsymbol{\Sigma} (n)  \big) \mathbf{A}^T + \mathbf{W} \label{kalman3}
\end{align}
Based on the dynamics of $\mathbf{x}(n)$ in (\ref{plantain}) and the dynamics of $\hat{\mathbf{x}}(n)$ in (\ref{kalman1}),  the dynamics of $\boldsymbol{\Delta}(n)$ can be obtained as in (\ref{kalman2}).  The sufficient conditions for optimality in Theorem  \ref{mdpvarify} directly follows   Prop. 4.6.1 of \cite{mdpcref2} and  Lemma 1 of  \cite{mdpsurvey}.

\section*{Appendix C:  Proof of Lemma \ref{perturbPDE}}

For convenience, denote   the operators of the R.H.S. of the Bellman equation in (\ref{OrgBel}) and the PDE in (\ref{bellman2}) as 
\bs\begin{align}
	&T_{\boldsymbol{\chi}}(\theta, V, \mathbf{F})=  \frac{1}{\tau} \mathbb{E}\bigg[\big((\boldsymbol{\Delta}')^T\mathbf{S}\boldsymbol{\Delta}'  + \lambda  \mathrm{Tr}\big( \mathbf{F}^\dagger \mathbf{F}\big)  \big)\tau  \notag \\
	&  +\sum_{  \boldsymbol{\Delta}', \boldsymbol{\Sigma}'}\Pr\left[  \boldsymbol{\Delta}', \boldsymbol{\Sigma}'\big|\boldsymbol{\chi}, \mathbf{F}  \right] V^\ast \left(  \boldsymbol{\Delta}', \boldsymbol{\Sigma}'\right)-  V^\ast \left(\boldsymbol{\Delta}, \boldsymbol{\Sigma}\right)\bigg|  \boldsymbol{\chi}\bigg] - \theta  \notag	\\
	& T_{\boldsymbol{\chi}}^\dagger(\theta, V, \mathbf{F})=  \boldsymbol{\Delta}^T \mathbf{S} \boldsymbol{\Delta}   + \lambda  \mathrm{Tr}\big( \mathbf{F}^\dagger \mathbf{F}\big) \notag \\
	&  - 2\text{Re}\left\{\nabla_{\boldsymbol{\Delta}}^T V(\boldsymbol{\Delta}, \boldsymbol{\Sigma})   \boldsymbol{\Sigma} \mathbf{F}^\dagger \mathbf{H}^\dagger \mathbf{H} \mathbf{F} \boldsymbol{\Delta}/\tau \right\}   + \nabla_{\boldsymbol{\Delta}}^T V(\boldsymbol{\Delta}, \boldsymbol{\Sigma})  \widetilde{\mathbf{A}} \boldsymbol{\Delta} \notag \\
	&+ \frac{1}{2}\text{Tr}\left(\nabla_{\boldsymbol{\Delta}}^2 V(\boldsymbol{\Delta}, \boldsymbol{\Sigma}) \widetilde{\mathbf{W}} \right) + \text{Tr}\left( \frac{\partial V(\boldsymbol{\Delta}, \boldsymbol{\Sigma})}{\partial \boldsymbol{\Sigma}}  \widetilde{\mathbf{W}} \right) -\theta \notag
\end{align}\bsc

\vspace{0.3cm}
\hspace{-0.3cm} \emph{A. Relationship between $T_{\boldsymbol{\chi}}(\theta, V, \mathbf{F})$ and $T_{\boldsymbol{\chi}}^\dagger(\theta, V, \mathbf{F})$}
\vspace{0.1cm}

\begin{Lemma}	\label{applema}
	For any $\boldsymbol{\chi}$, $ T_{\boldsymbol{\chi}}(\theta, V, \mathbf{F})=T_{\boldsymbol{\chi}}^\dagger(\theta, V, \mathbf{F})+\mathcal{O}(\tau)$.
\end{Lemma}

\begin{proof}	[Proof of Lemma \ref{applema}]

\txtblue{\emph{a. Calculation of the per-stage cost:}} We first calculate the per-stage cost in (\ref{OrgBel}):
\bs\begin{align}	
	&\mathbb{E}\left[(\boldsymbol{\Delta}')^T\mathbf{S}\boldsymbol{\Delta}'\tau \big| \boldsymbol{\chi} \right]	\label{45cal} \\
	\overset{(a)}{=}& \mathbb{E}\Big[\left(\big(\mathbf{I} - \mathbf{K}^a \mathbf{E}^a \big)\big((\mathbf{I}+\widetilde{\mathbf{A}}\tau+\mathcal{O}(\tau^2)\mathbf{I})\boldsymbol{\Delta}+ \mathbf{w}\big) - \mathbf{K}^a\mathbf{z}^a\right)^T\mathbf{S} \notag \\
	& \cdot \left(\big(\mathbf{I} - \mathbf{K}^a \mathbf{E}^a \big)\big((\mathbf{I}+\widetilde{\mathbf{A}}\tau+\mathcal{O}(\tau^2)\mathbf{I}) \boldsymbol{\Delta}+ \mathbf{w}\big)- \mathbf{K}^a\mathbf{z}^a\right)\tau \Big| \boldsymbol{\chi} \Big]  \notag \\
	  \overset{(b)}{=}&\mathbb{E}\Big[ \boldsymbol{\Delta}^T\mathbf{S}\boldsymbol{\Delta}\tau+ \txtblue{\big[\mathbf{w}^T\mathbf{S}\mathbf{w} + \boldsymbol{\Delta}^T (\mathbf{S}\mathbf{K}^a \mathbf{E}^a+(\mathbf{K}^a \mathbf{E}^a)^T\mathbf{S})\boldsymbol{\Delta}}\notag \\
	  &\txtblue{-\boldsymbol{\Delta}^T(\mathbf{S}\widetilde{\mathbf{A}}+\widetilde{\mathbf{A}}^T\mathbf{S})\boldsymbol{\Delta}\tau\big]\tau+\mathcal{O}(\tau^2)} \Big| \boldsymbol{\chi} \Big]  \overset{(c)}{=}  \boldsymbol{\Delta}^T\mathbf{S}\boldsymbol{\Delta}\tau+\mathcal{O}(\tau^2) \notag
\end{align}\bsc where (a) is because $\mathbf{A}=\mathbf{I}+\widetilde{\mathbf{A}}\tau+\mathcal{O}(\tau^2)\mathbf{I}$ according to  the dynamics in (\ref{plantain}), (b) and (c) are because $\mathbb{E}[\mathbf{w}\mathbf{w}^T]=\mathbf{W}= \widetilde{\mathbf{W}}\tau+\mathcal{O}(\tau^2)\mathbf{I}$, $\mathbf{K}^a=\mathcal{O}(\tau)\mathbf{I}$, $\boldsymbol{\Sigma}=\mathcal{O}(\tau)\mathbf{I}$ according to the expression of $\mathbf{K}^a$ in  (\ref{kalman2}) and the dynamics of $\boldsymbol{\Sigma}$ in (\ref{kalman31}).

\txtblue{\emph{b. Calculation of the expectation involving the transition kernel:}  Substituting the approximate priority function $V\in \mathcal{C}^2$  into the R.H.S. of  (\ref{OrgBel}), we calculate the expectation involving the transition kernel as follows\footnote{\txtblue{Note that although the optimal priority function $V^\ast(\boldsymbol{\Delta}, \boldsymbol{\Sigma})$ may not be $\mathcal{C}^2$, the proof just requires the approximate priority function $V(\boldsymbol{\Delta}, \boldsymbol{\Sigma})$ to be $\mathcal{C}^2$. In other words, we are seeking a $\mathcal{C}^2$ approximation of $V^\ast(\boldsymbol{\Delta}, \boldsymbol{\Sigma})$ with asymptotically vanishing errors for small $\tau$.}}:}
\bs\begin{align}
	& \mathbb{E}\bigg[\sum_{\boldsymbol{\Delta}', \boldsymbol{\Sigma}'}\Pr\left[\boldsymbol{\Delta}', \boldsymbol{\Sigma}'\big|\boldsymbol{\chi}, \mathbf{F}  \right] V^\ast \left(\boldsymbol{\Delta}', \boldsymbol{\Sigma}'\right)\bigg|\boldsymbol{\chi}\bigg]  \notag \\
	 = &\mathbb{E}\bigg[V^\ast \left(\boldsymbol{\Delta}, \boldsymbol{\Sigma}\right)+  \nabla_{\boldsymbol{\Delta}}^T V(\boldsymbol{\Delta}, \boldsymbol{\Sigma}) \left(\boldsymbol{\Delta}'-\boldsymbol{\Delta}\right)   \notag \\
	 &+ \frac{1}{2}\text{Tr}\left(\nabla_{\boldsymbol{\Delta}}^2 V(\boldsymbol{\Delta}, \boldsymbol{\Sigma})\left(\boldsymbol{\Delta}'-\boldsymbol{\Delta}\right)\left(\boldsymbol{\Delta}'-\boldsymbol{\Delta}\right)^T \right)\notag \\
	&+   \text{Tr}\left( \frac{\partial V(\boldsymbol{\Delta}, \boldsymbol{\Sigma})}{\partial \boldsymbol{\Sigma}}  \left(\boldsymbol{\Sigma}'-\boldsymbol{\Sigma}\right)\right) +\txtblue{\mathcal{O}(\|\boldsymbol{\Delta}'-\boldsymbol{\Delta}\|^3) } \notag \\
	&\txtblue{+\mathcal{O}(\|\boldsymbol{\Sigma}'-\boldsymbol{\Sigma}\|^2)+\mathcal{O}(\|\boldsymbol{\Delta}'-\boldsymbol{\Delta}\|\|\boldsymbol{\Sigma}'-\boldsymbol{\Sigma}\|)}\bigg|\boldsymbol{\chi}\bigg]  \label{47equal}
\end{align}\bsc We then calculate each term in  (\ref{47equal}) as follows: using (a) of (\ref{45cal}), we have 
\begin{align}
	 & \mathbb{E}\left[\boldsymbol{\Delta}'-\boldsymbol{\Delta}\big|\boldsymbol{\chi}\right]  = \widetilde{\mathbf{A}}\boldsymbol{\Delta} \tau - \mathbf{K}^a\mathbf{E}^a \boldsymbol{\Delta}\txtblue{-\mathbf{K}^a \mathbf{E}^a  \widetilde{\mathbf{A}}\boldsymbol{\Delta} \tau+\mathcal{O}(\tau^2)\mathbf{1}}  \notag \\
	 \overset{(d)}{=} & \widetilde{\mathbf{A}}\boldsymbol{\Delta} \tau -2 \text{Re}\Big\{ \boldsymbol{\Sigma} \mathbf{F}^\dagger \mathbf{H}^\dagger \mathbf{H} \mathbf{F} \boldsymbol{\Delta}\Big\} +\mathcal{O}(\tau^2)\mathbf{1}	\label{52dsfas1}
\end{align}
where $(d)$ is because $\mathbf{K}^a \mathbf{E}^a  \widetilde{\mathbf{A}}\boldsymbol{\Delta} \tau=\mathcal{O}(\tau^2)\mathbf{1}$ and $\mathbf{K}^a=\boldsymbol{\Sigma} (\mathbf{E}^a)^\dagger +\mathcal{O}(\tau^2)\mathbf{I}$ according to (\ref{kalman2}). Then, 
\bs\begin{align}
	&  \mathbb{E}\left[(\boldsymbol{\Delta}'-\boldsymbol{\Delta})(\boldsymbol{\Delta}'-\boldsymbol{\Delta})^T\big|\boldsymbol{\chi}\right]  \notag \\
	 =&\mathbb{E}\left[\mathbf{w}\mathbf{w}^T- \txtblue{\mathbf{K}^a \mathbf{E}^a\boldsymbol{\Delta}\boldsymbol{\Delta}^T  \widetilde{\mathbf{A}}\tau -\widetilde{\mathbf{A}}\boldsymbol{\Delta}\boldsymbol{\Delta}^T( \mathbf{K}^a \mathbf{E}^a)^T\tau  +\mathcal{O}(\tau^2)\mathbf{I}}\Big|\boldsymbol{\chi}\right] \notag \\
	 =& \widetilde{\mathbf{W}}\tau + \mathcal{O}\left(\tau^2\right)\mathbf{I}
\end{align}\bsc Then, using the calculations in  (\ref{45cal}) again, we have
\begin{align}
	 & \mathbb{E}\left[\boldsymbol{\Sigma}'-\boldsymbol{\Sigma}\big|\boldsymbol{\chi}\right]\notag \\
	  =& (\mathbf{I}+\widetilde{\mathbf{A}}\tau+\mathcal{O}(\tau^2)\mathbf{I})\big(\boldsymbol{\Sigma}- \boldsymbol{\Sigma}(\mathbf{E}^a)^\dagger  \big(\mathbf{E}^a\boldsymbol{\Sigma}  (\mathbf{E}^a)^\dagger+\mathbf{I}\big)^{-1}\mathbf{E}^a \boldsymbol{\Sigma}  \big)\notag \\
	  & \cdot (\mathbf{I}+\widetilde{\mathbf{A}}\tau+\mathcal{O}(\tau^2)\mathbf{I})^T + \mathbf{W} -\boldsymbol{\Sigma}  \notag \\
	 = &\txtblue{ (\widetilde{\mathbf{A}}^T \boldsymbol{\Sigma}  + \boldsymbol{\Sigma} \widetilde{\mathbf{A}} ) \tau + \mathcal{O}(\tau^2)\mathbf{I}}+ \mathbf{W}= \widetilde{\mathbf{W}}\tau+  \mathcal{O}\left(\tau^2\right)\mathbf{I} \label{sub1}
\end{align}
\txtblue{Using the calculations in (\ref{52dsfas1})--(\ref{sub1}), we can calculate that $\mathcal{O}(\|\boldsymbol{\Delta}'-\boldsymbol{\Delta}\|^3)$ is at least $\mathcal{O}(\tau^2)$, $\mathcal{O}(\|\boldsymbol{\Sigma}'-\boldsymbol{\Sigma}\|^2)=\mathcal{O}(\tau^2)$, $\mathcal{O}(\|\boldsymbol{\Delta}'-\boldsymbol{\Delta}\|\|\boldsymbol{\Sigma}'-\boldsymbol{\Sigma}\|)=\mathcal{O}(\tau^2)$.} Substituting the above calculations results  into $T_{\boldsymbol{\chi}}(\theta, V, \mathbf{F})$, we obtain $ T_{\boldsymbol{\chi}}(\theta, V, \mathbf{F})=T_{\boldsymbol{\chi}}^\dagger(\theta, V, \mathbf{F})+\mathcal{O}(\tau)$.\end{proof}

\vspace{0.3cm}
\hspace{-0.3cm} \emph{B. Growth Rate of $T_{\boldsymbol{\chi}}(\theta, V, \mathbf{F})$}
\vspace{0.1cm}

Denote 
\begin{align}
	T_{\boldsymbol{\chi}}(\theta, V)=\min_{\mathbf{F}} T_{\boldsymbol{\chi}}(\theta, V, \mathbf{F}),  T_{\boldsymbol{\chi}}^\dagger(\theta, V)=\min_{\mathbf{F}} T_{\boldsymbol{\chi}}^\dagger(\theta, V, \mathbf{F})	\notag	
\end{align}
Suppose $(\theta^\ast,V^\ast)$ satisfies the Bellman equation in (\ref{OrgBel}) and $(\theta,V)$ satisfies the approximate Bellman equation in (\ref{bellman2}). We have for any $\boldsymbol{\Delta}, \boldsymbol{\Sigma}$,
\begin{align}
	\mathbb{E}\big[T_{\boldsymbol{\chi}}(\theta^\ast, V^\ast)\big| \boldsymbol{\Delta}, \boldsymbol{\Sigma}\big]=0, \quad \mathbb{E}\big[ T_{\boldsymbol{\chi}}^\dagger(\theta, V)\big| \boldsymbol{\Delta}, \boldsymbol{\Sigma}\big]=0	\label{zerofunc}
\end{align}
Then, we establish the following lemma:
\begin{Lemma}	\label{applemma}
$\mathbb{E}\big[T_{\boldsymbol{\chi}}  (\theta, V)\big| \boldsymbol{\Delta}, \boldsymbol{\Sigma}\big]=\mathcal{O}(\tau)$,  $\forall \boldsymbol{\Delta}, \boldsymbol{\Sigma}$.
\end{Lemma}

\begin{proof}	[Proof of Lemma \ref{applemma}]
For any $\boldsymbol{\chi}$, we have $T_{\boldsymbol{\chi}} (\theta, V)=\min_{\mathbf{F}}\left[ T_{\boldsymbol{\chi}}^\dagger(\theta, V, \mathbf{F})+\mathcal{O}(\tau)\right] \geq  \min_{\mathbf{F}} T_{\boldsymbol{\chi}}^\dagger(\theta,  V, \mathbf{F})  + \mathcal{O}(\tau)$. On the other hand, $T_{\boldsymbol{\chi}} (\theta, V) \leq T_{\boldsymbol{\chi}}^\dagger(\theta, V, \mathbf{F}^\dagger) + \mathcal{O}(\tau)$, where $\mathbf{F}^\dagger= \arg \min_{F} T_{\boldsymbol{\chi}}^\dagger(\theta, V, \mathbf{F}) $. Since $\mathbb{E}\big[\min_{\mathbf{F}} T_{\boldsymbol{\chi}}^\dagger(\theta, V, \mathbf{F})\big| \boldsymbol{\Delta}, \boldsymbol{\Sigma}\big]=0$, we have $\mathbb{E}\big[T_{\boldsymbol{\chi}}  (\theta, V)\big| \boldsymbol{\Delta}, \boldsymbol{\Sigma}\big]=\mathcal{O}(\tau)$.\end{proof}

\vspace{0.3cm}
\hspace{-0.3cm} \emph{C. Difference between $V^\ast\left(\boldsymbol{\Delta},\boldsymbol{\Sigma} \right)$ and $V\left(\boldsymbol{\Delta}, \boldsymbol{\Sigma}\right)$}
\vspace{0.1cm}

\begin{Lemma}		\label{tenlemma}
	Suppose $\mathbb{E}\big[ T_{\boldsymbol{\chi}}(\theta^\ast, V^\ast)\big| \boldsymbol{\Delta}, \boldsymbol{\Sigma}\big] = 0$ for all $\boldsymbol{\Delta}, \boldsymbol{\Sigma}$ together with the transversality condition in (\ref{transodts})  has a unique solution $(\theta^*, V^\ast)$. If $\mathbb{E}\big[ T_{\boldsymbol{\chi}}^\dagger(\theta, V)\big| \boldsymbol{\Delta}, \boldsymbol{\Sigma}\big] =0$ and $V(\boldsymbol{\Delta}, \boldsymbol{\Sigma})=\mathcal{O}(\left\|\boldsymbol{\Delta}\right\|^2)$,  then $|V^\ast\left(\boldsymbol{\Delta},\boldsymbol{\Sigma} \right)-V\left(\boldsymbol{\Delta},\boldsymbol{\Sigma} \right)|=\mathcal{O}(\tau)$ for all  $\boldsymbol{\Delta}, \boldsymbol{\Sigma}$.	
\end{Lemma}

\begin{proof}	[Proof of Lemma \ref{tenlemma}]
Since $V(\boldsymbol{\Delta}, \boldsymbol{\Sigma})=\mathcal{O}(\left\|\boldsymbol{\Delta}\right\|^2)$, we have $\lim_{n \rightarrow \infty}\mathbb{E}^{\Omega} \left[V( \boldsymbol{\Delta}(n), \boldsymbol{\Sigma}(n))\right]<\infty$ for any admissible policy $\Omega$ (according to Definition \ref{admisscontrolpol}, we have  $\mathbb{E}^{\Omega}\left[\left\|\boldsymbol{\Delta}\right\|^2 \right] < \infty$.). Then, we have $\lim_{N\rightarrow \infty} \frac{1}{N}\mathbb{E}^{\Omega}\left[ V\left( \boldsymbol{\Delta}(N),  \boldsymbol{\Sigma}(N)\right) |\boldsymbol{\chi}(0)\right]  = 0$ and the transversality condition in (\ref{transodts}) is satisfied for $V(\boldsymbol{\Delta}, \boldsymbol{\Sigma})$.

	Suppose for some $  \boldsymbol{\Delta}', \boldsymbol{\Sigma}'$, we have $V\left( \boldsymbol{\Delta}',\boldsymbol{\Sigma}' \right)=V^\ast\left( \boldsymbol{\Delta}',\boldsymbol{\Sigma}' \right)+\alpha$ for some $\alpha \neq 0$ as $\tau \rightarrow 0$. Now let $\tau \rightarrow 0$. From Lemma \ref{applemma}, we have $(\theta, V)$ satisfies $\mathbb{E}\big[T_{\boldsymbol{\chi}}(\theta, V)\big| \boldsymbol{\Delta}, \boldsymbol{\Sigma}\big]= 0$ for all $\boldsymbol{\Delta}, \boldsymbol{\Sigma}$ and satisfies  the transversality condition in (\ref{transodts}). However, $V\left( \boldsymbol{\Delta}',\boldsymbol{\Sigma}' \right) \neq V^\ast\left( \boldsymbol{\Delta}',\boldsymbol{\Sigma}' \right)$ because of the assumption that $V\left( \boldsymbol{\Delta}',\boldsymbol{\Sigma}' \right)=V^\ast\left( \boldsymbol{\Delta}',\boldsymbol{\Sigma}' \right)+\alpha$. This contradicts  the condition that $(\theta^*, V^\ast)$ is a unique solution of $\mathbb{E}\big[T_{\boldsymbol{\chi}}(\theta^\ast, V^\ast)\big| \boldsymbol{\Delta}, \boldsymbol{\Sigma}\big] =0$ for all $\boldsymbol{\Delta}, \boldsymbol{\Sigma}$  and the transversality condition in (\ref{transodts}). Hence, we must have $|V\left(\boldsymbol{\Delta},\boldsymbol{\Sigma} \right)-V^\ast\left(\boldsymbol{\Delta},\boldsymbol{\Sigma} \right)|=\mathcal{O}(\tau)$ for all  $\boldsymbol{\Delta}, \boldsymbol{\Sigma}$.
\end{proof}

\section*{Appendix D:  Proof of Theorem \ref{perfgap}}

We  calculate the performance under policy $\widetilde{\Omega}^\ast$ as follows:
\bs\begin{align}
	&\widetilde{\theta}^\ast\tau= \mathbb{E}^{\widetilde{\Omega}^\ast}\Big[\mathbb{E}\left[\big((\boldsymbol{\Delta}')^T\mathbf{S}\boldsymbol{\Delta}' + \lambda  \text{Tr}\big(\mathbf{F}^\dagger \mathbf{F} \big)  \big)\tau\big| \boldsymbol{\Delta}, \boldsymbol{\Sigma}\right]\Big]\notag \\
	&\overset{(a)}=\mathbb{E}^{\widetilde{\Omega}^\ast}\bigg[\mathbb{E}\bigg[\big((\boldsymbol{\Delta}')^T\mathbf{S}\boldsymbol{\Delta}' + \lambda  \text{Tr}\big(\mathbf{F}^\dagger \mathbf{F} \big)  \big)\tau \notag \\
	&+\sum_{  \boldsymbol{\Delta}', \boldsymbol{\Sigma}'} {\Pr}\big[   \boldsymbol{\Delta}', \boldsymbol{\Sigma}'| \boldsymbol{\chi},  \widetilde{\Omega}^\ast\left(\boldsymbol{\chi}  \right)\big]V \left(  \boldsymbol{\Delta}', \boldsymbol{\Sigma}'\right)  - V \left(\boldsymbol{\Delta}, \boldsymbol{\Sigma}  \right)   \Big| \boldsymbol{\Delta}, \boldsymbol{\Sigma}\bigg]\bigg] 	\notag \\
	&\overset{(b)}=\mathbb{E}^{\widetilde{\Omega}^\ast}\left[T_{\boldsymbol{\chi}}^\dagger(\theta, V, \mathbf{F})+\theta\tau+\mathcal{O}(\tau^2) \bigg| \boldsymbol{\Delta}, \boldsymbol{\Sigma}\right] 	\label{finalexpr}
\end{align}\bsc where ${\Pr}\big[   \boldsymbol{\Delta}', \boldsymbol{\Sigma}'|\boldsymbol{\chi},  \widetilde{\Omega}^\ast\left(\boldsymbol{\chi}\right)\big]$ is the discrete-time transition kernel under policy $\widetilde{\Omega}^\ast$. $(a)$ is due to 1) $\mathbb{E}^{\widetilde{\Omega}^\ast}\big[V(\boldsymbol{\Delta}, \boldsymbol{\Sigma} ) \big]<\infty$  (according to the conditions in  Theorem \ref{perfgap}, we have $V(\boldsymbol{\Delta}, \boldsymbol{\Sigma})=\mathcal{O}\big(\left\|\boldsymbol{\Delta}\right\|^2\big)$ and $\mathbb{E}^{\widetilde{\Omega}^\ast}\big[\|\boldsymbol{\Delta}\|^2\big]$ is bounded under admissible $\widetilde{\Omega}^\ast$) and  2) $ \mathbb{E}^{\widetilde{\Omega}^\ast}\big[\sum_{  \boldsymbol{\Delta}', \boldsymbol{\Sigma}'} \mathbb{E} [ {\Pr}\big[   \boldsymbol{\Delta}', \boldsymbol{\Sigma}'\big| \boldsymbol{\chi},    \widetilde{\Omega}^\ast\left(\boldsymbol{\chi}  \right)\big]\big| \boldsymbol{\Delta}, \boldsymbol{\Sigma}]  V \left(  \boldsymbol{\Delta}', \boldsymbol{\Sigma}'\right)   \big]  =\mathbb{E}^{\widetilde{\Omega}^\ast}\big[\mathbb{E}^{\widetilde{\Omega}^\ast}\big[V(  \boldsymbol{\Delta}', \boldsymbol{\Sigma}') \big|\boldsymbol{\Delta}, \boldsymbol{\Sigma}  \big]\big]=\mathbb{E}^{\widetilde{\Omega}^\ast}\big[V(\boldsymbol{\Delta}, \boldsymbol{\Sigma} ) \big]$, and $(b)$ is due to Lemma \ref{applema}.

Following the notation of the \emph{Bellman operators}  in  Appendix D, we define two mappings: $T_{\boldsymbol{\chi}}^\dagger( V, \mathbf{F}) = T_{\boldsymbol{\chi}}^\dagger(\theta, V, \mathbf{F}) + \theta $, $T_{\boldsymbol{\chi}}(V, \mathbf{F})=T_{\boldsymbol{\chi}}(\theta, V, \mathbf{F})+\theta$.  Let $\Omega^\ast$ be the optimal policy solving the discrete-time Bellman equation in (\ref{OrgBel}). Then we have
\begin{align}
	\mathbb{E}\big[T_{\boldsymbol{\chi}}({V^\ast}, \Omega^\ast(\boldsymbol{\chi}))\big| \boldsymbol{\Delta}, \boldsymbol{\Sigma}\big]= \theta^\ast, \quad \forall \boldsymbol{\Delta}, \boldsymbol{\Sigma}	\label{asdadadasd}
\end{align}
Furthermore, we have
\begin{align}	\label{minach}
	T_{\boldsymbol{\chi}}^\dagger( V, \widetilde{\Omega}^\ast(\boldsymbol{\chi} ))= \min_{\Omega(\boldsymbol{\chi} )} T_{\boldsymbol{\chi}}^\dagger( V, \Omega(\boldsymbol{\chi})), \quad \forall\boldsymbol{\Delta}, \boldsymbol{\Sigma}
\end{align}
Dividing $\tau$ on both sizes of (\ref{finalexpr}), we obtain
\begin{align}
	&\widetilde{\theta}^\ast =\mathbb{E}^{\widetilde{\Omega}^\ast}\big[\mathbb{E}\big[T_{\boldsymbol{\chi}}^\dagger(V, \widetilde{\Omega}^\ast(\boldsymbol{\chi} ))+ \mathcal{O}(\tau)\big|\boldsymbol{\Delta}, \boldsymbol{\Sigma}\big]\big] 	\label{finaleeeee} \\
	\overset{(c)} \leq & \mathbb{E}^{\widetilde{\Omega}^\ast}\big[\mathbb{E}\big[T_{\boldsymbol{\chi}}^\dagger(V, {\Omega^\ast}(\boldsymbol{\chi} )) + \mathcal{O}(\tau)\big| \boldsymbol{\Delta}, \boldsymbol{\Sigma} \big]\big]\notag \\
	& \overset{(d)}{=}\mathbb{E}^{\widetilde{\Omega}^\ast}\big[\mathbb{E}\big[T_{\boldsymbol{\chi}}(V, \Omega^\ast(\boldsymbol{\chi} )) + \mathcal{O}(\tau) \big|\boldsymbol{\Delta}, \boldsymbol{\Sigma}\big]\big] 	\notag \\
	\overset{(e)}=& \mathbb{E}^{\widetilde{\Omega}^\ast}\big[\mathbb{E}\big[T_{\boldsymbol{\chi}}(V, \Omega^\ast(\boldsymbol{\chi} )) - T_{\boldsymbol{\chi}}({V^\ast}, \Omega^\ast(\boldsymbol{\chi} )) + \theta^\ast+ \mathcal{O}(\tau) \big|\boldsymbol{\Delta}, \boldsymbol{\Sigma} \big]\big] 	\notag
\end{align}
where (c) is due to (\ref{minach}), (d) is due to Lemma \ref{applema},  and  (e) is due to (\ref{asdadadasd}).  \txtblue{Then, from (\ref{finaleeeee}),   we have
\begin{align}
	&\widetilde{\theta}^\ast -  \theta^\ast  \\
	&\leq \mathbb{E}^{\widetilde{\Omega}^\ast}\Big[\mathbb{E}\big[T_{\boldsymbol{\chi}}(V, {\Omega}^\ast(\boldsymbol{\chi} ))-T_{\boldsymbol{\chi}}(V^\ast, {\Omega}^\ast(\boldsymbol{\chi} ))\big|\boldsymbol{\Delta}, \boldsymbol{\Sigma}\big]\Big]+ \mathcal{O}(\tau)   \notag \\
	&\overset{(f)}{\leq} \gamma \mathbb{E}^{\widetilde{\Omega}^\ast}\Big[\mathbb{E}\big[\omega(\boldsymbol{\chi})\|\mathbf{V}^\ast-\mathbf{V}\|_{\infty}^{\overline{\omega}}\big|\boldsymbol{\Delta}, \boldsymbol{\Sigma}\big]\Big]+ \mathcal{O}(\tau)  \notag \\
	& \overset{(g)}{=}\gamma \mathbb{E}^{\widetilde{\Omega}^\ast}\Big[\mathbb{E}\big[\omega(\boldsymbol{\chi})\big(\mathcal{O}(\tau)  \big)\big|\boldsymbol{\Delta}, \boldsymbol{\Sigma}\big]\Big]+ \mathcal{O}(\tau)   \overset{(h)}{=} \mathcal{O}(\tau)  \notag
\end{align}
where $(f)$ holds because
\begin{align}	\label{comtracmapp}
	\|T_{\boldsymbol{\chi}}(V, {\Omega}^\ast(\boldsymbol{\chi} ))-T_{\boldsymbol{\chi}}(V^\ast, {\Omega}^\ast(\boldsymbol{\chi} ))\|_{\infty}^{\overline{\omega}}\leq \gamma\|\mathbf{V}^\ast-\mathbf{V}\|_{\infty}^{\overline{\omega}}
\end{align}
with $\mathbf{V}^\ast=\{V^\ast(\boldsymbol{\Delta}, \boldsymbol{\Sigma}):\forall \boldsymbol{\Delta}, \boldsymbol{\Sigma}\}$ and $\mathbf{V}=\{V(\boldsymbol{\Delta}, \boldsymbol{\Sigma}):\forall \boldsymbol{\Delta}, \boldsymbol{\Sigma}\}$, for $0<\gamma<1$ according to Lemma 3 of \cite{hinidhsd} and $\|\cdot\|_{\infty}^{\overline{\omega}}$ is a weighted sup-norm with weights $\overline{\omega}=\{0<\omega(\boldsymbol{\chi})<1:\forall \boldsymbol{\chi}\}$ chosen according to the following rule (Lemma 3 of \cite{hinidhsd}): The state space w.r.t. $\boldsymbol{\chi}$ is partitioned into non-empty subsets $\mathcal{S}_1,\dots, \mathcal{S}_r$, in which for any $\boldsymbol{\chi}\in \mathcal{S}_n$ with $n=1,\dots, r$,  there exists some $\boldsymbol{\chi}'\in \mathcal{S}_1\cup \mathcal{S}_{n-1}$ such that $\Pr[\boldsymbol{\chi}'|\boldsymbol{\chi},{\Omega}^\ast(\boldsymbol{\chi}) ]>0$. Then, we let $\rho=\min \{\Pr[\boldsymbol{\chi}'|\boldsymbol{\chi},{\Omega}^\ast(\boldsymbol{\chi}) ]:\forall \boldsymbol{\chi},\boldsymbol{\chi}' \}$ and choose $\omega(\boldsymbol{\chi})=1-\rho^{2n}$ if $\boldsymbol{\chi}\in \mathcal{S}_n$ for $n=1,\dots, r$. Therefore, based on the contraction mapping property in (\ref{comtracmapp}) and the definition of the weighted sup-norm, we have
\begin{align}
	& \frac{T_{\boldsymbol{\chi}}(V, {\Omega}^\ast(\boldsymbol{\chi} ))-T_{\boldsymbol{\chi}}(V^\ast, {\Omega}^\ast(\boldsymbol{\chi} ))}{\omega(\boldsymbol{\chi})}  \\
	\leq & \sup_{\boldsymbol{\chi} }\left\{\frac{T_{\boldsymbol{\chi}}(V, {\Omega}^\ast(\boldsymbol{\chi} ))-T_{\boldsymbol{\chi}}(V^\ast, {\Omega}^\ast(\boldsymbol{\chi} ))}{\omega(\boldsymbol{\chi})}\right\}\notag \\
	=& \|T_{\boldsymbol{\chi}}(V, {\Omega}^\ast(\boldsymbol{\chi} ))-T_{\boldsymbol{\chi}}(V^\ast, {\Omega}^\ast(\boldsymbol{\chi} ))\|_{\infty}^{\overline{\omega}}\leq \gamma\|\mathbf{V}^\ast-\mathbf{V}\|_{\infty}^{\overline{\omega}}\notag \\
	\Rightarrow &T_{\boldsymbol{\chi}}(V, {\Omega}^\ast(\boldsymbol{\chi} ))-T_{\boldsymbol{\chi}}(V^\ast, {\Omega}^\ast(\boldsymbol{\chi} ))\leq \gamma \omega(\boldsymbol{\chi})\|\mathbf{V}^\ast-\mathbf{V}\|_{\infty}^{\overline{\omega}} \notag
\end{align}
This proves (f), and (g) is because $\|\mathbf{V}^\ast-\mathbf{V}\|_{\infty}^{\overline{\omega}}=\sup_{\boldsymbol{\chi}}\left\{\frac{|V^\ast(\boldsymbol{\Delta}, \boldsymbol{\Sigma})-V(\boldsymbol{\Delta}, \boldsymbol{\Sigma})|}{\omega(\boldsymbol{\chi})}\right\}=\mathcal{O}(\tau)  $ according to Lemma \ref{perturbPDE}, and (h) is because $0<\omega(\boldsymbol{\chi})<1$ for all $\boldsymbol{\chi}$.}

\section*{Appendix E: Proof of Theorem \ref{thmpower}}

Let an eigenvalue decomposition of  the channel matrix be $\mathbf{H}^\dagger \mathbf{H}= \mathbf{U} \boldsymbol{\Pi}  \mathbf{U}^\dagger$, where  $ \mathbf{U} \in \mathbb{C}^{N_t \times N_t}$ is a unitary matrix, $\boldsymbol{\Pi} \in \mathbb{R}^{N_t \times N_t}$ is diagonal with elements being the  squared singular values in a descending   order, i.e., $\sigma^\ast>\sigma_2^2>\cdots>\sigma_d^2$ where $d=\min(N_t,N_r)$. Then,  using the transformation of $\mathbf{G}= \mathbf{U}^\dagger \mathbf{F}$, the problem in the  PDE (\ref{bellman2}) becomes:
\begin{align}	
	 \ &\min_{\mathbf{G}}\left[ \lambda \mathrm{Tr}( \mathbf{G}^\dagger \mathbf{G})  - 2\text{Re}\left\{ \text{Tr}\left(\boldsymbol{\Xi} \mathbf{G}^\dagger \boldsymbol{\Pi} \mathbf{G}  \right) \right\}\right]	\notag \\
	 & \ \text{s.t.} \  \ \mathrm{Tr}( \mathbf{G}^\dagger \mathbf{G} ) \leq \bar{F}  \hspace{4cm} 
\end{align}
where we use $\mathrm{Tr}\big( \mathbf{F}^\dagger \mathbf{F}  \big) =\mathrm{Tr}( \mathbf{G}^\dagger \mathbf{G})$ under $\mathbf{G}= \mathbf{U}^\dagger \mathbf{F}$. We further write $\mathbf{G}=\sqrt{g}\widetilde{\mathbf{G}}$ such that $g=\|\mathbf{G}\|_F$ and $\mathrm{Tr}( \widetilde{\mathbf{G}}^\dagger \widetilde{\mathbf{G}})=1$. Hence, we write the above problem in the following form:
\begin{align}	
	 \mathcal{P}_1:  \widetilde{U}(c) \ = \ &\min_{\widetilde{\mathbf{G}}, c}\left[ \lambda  - 2\text{Re}\left\{ \text{Tr}\left(\boldsymbol{\Xi} \widetilde{\mathbf{G}}^\dagger \boldsymbol{\Pi} \widetilde{\mathbf{G}}  \right) \right\}\right]c \notag \\
	 & \ \text{s.t.} \  \ \mathrm{Tr}( \mathbf{G}^\dagger \mathbf{G} )=c  \notag
\end{align} \begin{align}	\label{probtrans}
	 \mathcal{P}_2: \min_{c} \  \widetilde{U}(c)	\notag \\
	  \ \text{s.t.} \  \ 0\leq c \leq \bar{F}
\end{align}

We first solve $\mathcal{P}_1$ for given $c$. Since $\lambda g$ is a constant, the objective of $\mathcal{P}_1$ related to $\widetilde{\mathbf{G}}$ becomes  
\begin{align}
	&2 \text{Re}\left\{ \text{Tr}\left(\boldsymbol{\Xi} \widetilde{\mathbf{G}}^\dagger  \boldsymbol{\Pi}  \widetilde{\mathbf{G}} \right)  \right\}\notag \\
	=&\text{Re}\left\{ \text{Tr}\left(\left(({\boldsymbol{\Xi}+\boldsymbol{\Xi}^T})+({\boldsymbol{\Xi}-\boldsymbol{\Xi}^T})\right) \widetilde{\mathbf{G}}^\dagger  \boldsymbol{\Pi}  \widetilde{\mathbf{G}}  \right)\right\} \notag \\
	\overset{(a)}{=}& \text{Re}\left\{ \text{Tr}\left(({\boldsymbol{\Xi}+\boldsymbol{\Xi}^T}) \widetilde{\mathbf{G}}^\dagger  \boldsymbol{\Pi}  \widetilde{\mathbf{G}}   \right)\right\}		\label{54equations}
\end{align}
where (a) is because $\text{Re}\left\{ \text{Tr}\left(({\boldsymbol{\Xi}-\boldsymbol{\Xi}^T}) \widetilde{\mathbf{G}}^\dagger  \boldsymbol{\Pi}  \widetilde{\mathbf{G}}   \right)\right\} = \text{Re}\left\{ \sum_{i=1}^d   \text{Tr}\left(  \sigma_i  \widetilde{\mathbf{g}}_i  ({\boldsymbol{\Xi}-\boldsymbol{\Xi}^T})\widetilde{\mathbf{g}}_i^\dagger\right)\right\} =0$  where ${\boldsymbol{\Xi}-\boldsymbol{\Xi}^T}$ is skew-symmetric and $\widetilde{\mathbf{g}}_i$ is the $i$-th row of $\widetilde{\mathbf{G}}$. Furthermore, the  $\mathrm{Tr}( \widetilde{\mathbf{G}}^\dagger \widetilde{\mathbf{G}})=1$ is equivalent to $\sum_{i=1}^{N_t}  \|\widetilde{\mathbf{g}}_i\|^2=1$. Also, the matrix $\boldsymbol{\Xi}+\boldsymbol{\Xi}^T$ is  symmetric and we have $\boldsymbol{\Xi}+\boldsymbol{\Xi}^T= \sum_{i}^{\text{rank}(\boldsymbol{\Xi})}\nu_i\mathbf{q}_i \mathbf{q}_i^T$, where $\nu_i$ is the eigenvalue and $\mathbf{q}_i$ is  the associated $L \times 1$ orthonormal column eigenvectors.  Therefore, (\ref{54equations}) becomes $\text{Re}\left\{ \text{Tr}\left(\left(\sum_{i}^{\text{rank}(\boldsymbol{\Xi})}\nu_i\mathbf{q}_i \mathbf{q}_i^T\right) \widetilde{\mathbf{G}}^\dagger  \boldsymbol{\Pi}  \widetilde{\mathbf{G}}   \right)\right\}=\text{Re}\left\{\sum_{i}^{\text{rank}(\boldsymbol{\Xi})}\nu_i\sum_{l=1}^d \sigma_l^2 \big|\widetilde{\mathbf{g}}_l\mathbf{q}_i\big|^2\right\}	$. The optimal $\widetilde{\mathbf{G}}^\ast$ of the above problem under $\sum_{i=1}^{N_t}  \|\widetilde{\mathbf{g}}_i\|^2=1$ is  $\widetilde{\mathbf{g}}_1^\ast=\mathbf{q}_1^T$ and $\widetilde{\mathbf{g}}_i^\ast=0$ for $i\neq 1$.  Substituting $\widetilde{\mathbf{G}}^\ast$, $\mathcal{P}_2$ in (\ref{probtrans}) becomes
\begin{align}	
	&\min_{c} \ \left[ \lambda  -\sigma^\ast \nu^\ast  \right]g	\notag \\
	&\ \text{s.t.}	 \  \  0\leq g\leq \bar{F}
\end{align}
where $\nu^\ast=\nu_1$. The optimal solution of the above problem is $g^\ast = 0$ if $ \lambda >\sigma^\ast \nu^\ast   $, and $g^\ast = \bar{F}$ if $ \lambda <\sigma^\ast \nu^\ast   $.    Combining the solution of $\mathcal{P}_1$ and $\mathcal{P}_2$, the optimal precoding is summarized as follows:  if $ \lambda >\sigma^\ast \nu^\ast  $, $\mathbf{F}^\ast=0$. If $ \lambda <\sigma^\ast \nu^\ast   $,  $\mathbf{F}^\ast=\mathbf{U} \mathbf{G}^\ast=\sqrt{\bar{F}  }\mathbf{U} \boldsymbol{\Upsilon}$, where the first row is the only non-zero row of $\Upsilon$ which is given by   $\mathbf{q}_1^T$.

\section*{Appendix F: Proof of Lemma \ref{solpde}}

Substituting $\mathbf{F}^\ast$ in  Theorem \ref{thmpower} into the PDE in (\ref{bellman2}), we obtain 
\bs\begin{align}	
		\theta  = \boldsymbol{\Delta}^T \mathbf{S} \boldsymbol{\Delta} - \mathbb{E}\left[ \left[ \sigma^\ast \nu^\ast-\lambda    \right]^+\big| \boldsymbol{\Delta}, \boldsymbol{\Sigma}   \right] \bar{F}     +\nabla_{\boldsymbol{\Delta}}^T V  \widetilde{\mathbf{A}} \boldsymbol{\Delta} \notag \\
		+ \frac{1}{2}\text{Tr}\left( \nabla_{\boldsymbol{\Delta}}^2 V \widetilde{\mathbf{W}} \right)+ \text{Tr}\left( \frac{\partial V}{\partial \boldsymbol{\Sigma}}  \widetilde{\mathbf{W}} \right) 	\label{bellman21}
\end{align}\bsc
	
The difficult of solving the above PDE lies in the nonlinear expectation part. In the following, we will solve the PDE for small $\nu^\ast$ and large $\nu^\ast$ cases, and we further show that they corresponds to  small $\|\boldsymbol{\Delta} \|$ and large $\|\boldsymbol{\Delta} \|$, respectively.
	
\vspace{0.3cm}
\hspace{-0.3cm}\emph{A. Solution of (\ref{bellman2}) for small $\nu^\ast$} 
\vspace{0.1cm}

In this part, we solve the PDE  in (\ref{bellman21})  for small $\nu^\ast$. We will show later that small $\|\boldsymbol{\Delta}\|$  leads to this case. Specifically, for small $\nu^\ast$,  the expectation in (\ref{bellman21}) becomes
\begin{align}
	\mathbb{E}\left[ \left[ \sigma^\ast \nu^\ast-\lambda    \right]^+\big| \boldsymbol{\Delta}, \boldsymbol{\Sigma}   \right]=\mathcal{O}(\nu^\ast)
\end{align}
Substituting the above equation into the PDE in (\ref{bellman21}), we obtain
\begin{align}
	\theta  = \boldsymbol{\Delta}^T \mathbf{S} \boldsymbol{\Delta} - \mathcal{O}(\nu^\ast)     +\nabla_{\boldsymbol{\Delta}}^T V  \widetilde{\mathbf{A}} \boldsymbol{\Delta} + \frac{1}{2}\text{Tr}\left( \nabla_{\boldsymbol{\Delta}}^2 V \widetilde{\mathbf{W}} \right)\notag \\
	+ \text{Tr}\left( \frac{\partial V}{\partial \boldsymbol{\Sigma}}  \widetilde{\mathbf{W}} \right)
\end{align}
The solution of the above PDE has the structure $V(\boldsymbol{\Delta}, \boldsymbol{\Sigma})=\widetilde{V}_1(\boldsymbol{\Delta}, \boldsymbol{\Sigma})+J_1(\boldsymbol{\Delta}, \boldsymbol{\Sigma})$ ($J_1$ can be treated as a residual error term for $V$), where $\widetilde{V}_1$ and $J_1$ satisfy
\bs\begin{align}
	&\theta  = \boldsymbol{\Delta}^T \mathbf{S} \boldsymbol{\Delta}   +\nabla_{\boldsymbol{\Delta}}^T \widetilde{V}_1  \widetilde{\mathbf{A}} \boldsymbol{\Delta} + \frac{1}{2}\text{Tr}\left( \nabla_{\boldsymbol{\Delta}}^2 \widetilde{V}_1 \widetilde{\mathbf{W}} \right)+ \text{Tr}\left( \frac{\partial \widetilde{V}_1}{\partial \boldsymbol{\Sigma}}  \widetilde{\mathbf{W}} \right)	\label{xxx1}	\\
	&\mathcal{O}(\nu^\ast)=      \nabla_{\boldsymbol{\Delta}}^T J_1  \widetilde{\mathbf{A}} \boldsymbol{\Delta} + \frac{1}{2}\text{Tr}\left( \nabla_{\boldsymbol{\Delta}}^2 J_1 \widetilde{\mathbf{W}} \right)+ \text{Tr}\left( \frac{\partial J_1}{\partial \boldsymbol{\Sigma}}  \widetilde{\mathbf{W}} \right)\label{xxx2}
\end{align}\bsc We  obtain $\widetilde{V}_1$ and $J_1$ by solving the above two equations in the following:

\emph{1) Obtaining $\widetilde{V}_1$:}  This PDE in (\ref{xxx1}) is separable with solution of the form $\widetilde{V}_1 =\boldsymbol{\Delta}^T\Phi_1(\boldsymbol{\Sigma})\boldsymbol{\Delta}+ \psi_1(\boldsymbol{\Sigma})$ for some $\Phi_1(\boldsymbol{\Sigma})\in \mathbb{R}^{L\times L}$ and $\psi_1(\boldsymbol{\Sigma})\in \mathbb{R}$. Substituting this form into (\ref{xxx1}), we obtain
\bs\begin{align}
	&\boldsymbol{\Delta}^T \left[\mathbf{S}+ \big(\Phi_1(\boldsymbol{\Sigma})+\Phi_1^T(\boldsymbol{\Sigma})\big) \widetilde{\mathbf{A}}  + \text{Tr}\left( \frac{\partial \Phi_1(\boldsymbol{\Sigma})}{\partial \boldsymbol{\Sigma}}  \widetilde{\mathbf{W}}  \right)\right] \boldsymbol{\Delta}  \label{59veryimp} \\
	&  + \left[  \frac{1}{2}\text{Tr}\left( \big(\Phi_1(\boldsymbol{\Sigma})+\Phi_1^T(\boldsymbol{\Sigma})\big)\widetilde{\mathbf{W}} \right) + \text{Tr}\left( \frac{\partial \psi_1(\boldsymbol{\Sigma})}{\partial \boldsymbol{\Sigma}}  \widetilde{\mathbf{W}}  \right)-\theta \right]=0	\notag
\end{align}\bsc In order for the above equation to hold for any $\boldsymbol{\Delta}$ and $\boldsymbol{\Sigma}$, we require the coefficient of $\boldsymbol{\Delta}^T\boldsymbol{\Delta}$ to be zero:
\begin{align}
	\mathbf{S}+ \big(\Phi_1(\boldsymbol{\Sigma})+\Phi_1^T(\boldsymbol{\Sigma})\big) \widetilde{\mathbf{A}}  + \text{Tr}\left( \frac{\partial \Phi_1(\boldsymbol{\Sigma})}{\partial \boldsymbol{\Sigma}}  \widetilde{\mathbf{W}}  \right) ={0}	\label{deltaequation}
\end{align}
Let the eigenvalue decomposition of $\widetilde{\mathbf{A}}$  be $\widetilde{\mathbf{A}}=\mathbf{M}^{-1}\boldsymbol{\Gamma} \mathbf{M}$, where $\mathbf{M}$ is an $L \times L$  matrix and $\boldsymbol{\Gamma} =\text{diag}\left(\mu_1, \mu_2, \dots, \mu_L \right)$ and $\{\mu_l\}$ are the eigenvalues of $\widetilde{\mathbf{A}}$. Using the change of variable $\mathbf{Z}=\mathbf{M}\boldsymbol{\Delta}$, denoting $\Phi_{1}^\mathbf{M}(\boldsymbol{\Sigma})=(\mathbf{M}^{-1})^\dagger \Phi_{1}(\boldsymbol{\Sigma}) \mathbf{M}^{-1}=[\phi_{1,  kl}^\mathbf{M}(\boldsymbol{\Sigma})]$, from (\ref{deltaequation}),  we have 
\bs\begin{align}	
	S^\mathbf{M}+ \big(\Phi_{1}^\mathbf{M}(\boldsymbol{\Sigma})+(\Phi_{1}^\mathbf{M}(\boldsymbol{\Sigma}))^T\big)\boldsymbol{\Gamma}    + \text{Tr}\left( \frac{\partial \Phi_{1}^\mathbf{M}(\boldsymbol{\Sigma})}{\partial \boldsymbol{\Sigma}}  \widetilde{\mathbf{W}} \right) = 0 \label{basedonequ} 
\end{align}\bsc  where $\mathbf{S}^\mathbf{M}\triangleq (\mathbf{M}^{-1})^\dagger \mathbf{S} \mathbf{M}^{-1}=\left[s_{kl}^\mathbf{M} \right]$ and denote  $ \Phi_{1, M}(\boldsymbol{\Sigma})$. We then solve (\ref{basedonequ}). For the diagonal elements in (\ref{basedonequ}), we have
\begin{align}	\label{8964equ}
	s_{kk}^\mathbf{M}+ 2 \mu_k \phi_{1, kk}^\mathbf{M}(\boldsymbol{\Sigma})+\sum_{k=1}^L \frac{\partial \phi_{1, kk}^\mathbf{M}(\boldsymbol{\Sigma})}{\partial \Sigma_{kk}}\widetilde{w}_{kk}=0
\end{align}	
where $\mathbf{W}=\text{diag}(w_{11},\dots, w_{LL})$.  For $\phi_{1, kl}^\mathbf{M}$ and $\phi_{1, lk}^\mathbf{M}$ in $\Phi_{1,U}(\boldsymbol{\Sigma})$, they satisfy the following coupled ODEs based on (\ref{basedonequ}) for $k<l$:
\bs\begin{align}
	s_{kl}^\mathbf{M} + \big(\phi_{1,kl}^\mathbf{M}(\boldsymbol{\Sigma})+\phi_{1,lk}^\mathbf{M}(\boldsymbol{\Sigma})\big)\mu_l +\sum_{k=1}^L \frac{\partial \phi_{1, kl}^\mathbf{M} (\boldsymbol{\Sigma})}{\partial \Sigma_{kk}}\widetilde{w}_{kk} =0 \label{xxy}	\\
	s_{kl}^\mathbf{M} +\big(\phi_{1,kl}^\mathbf{M}(\boldsymbol{\Sigma})+\phi_{1,lk}^\mathbf{M}(\boldsymbol{\Sigma})\big)\mu_k+\sum_{k=1}^L \frac{\partial \phi_{1, lk} (\boldsymbol{\Sigma})}{\partial \Sigma_{kk}}\widetilde{w}_{kk}=0 \label{yyx}
\end{align}\bsc Even though (\ref{xxy}) and (\ref{yyx}) are coupled, we can first obtain $\phi_{1,kl}^\mathbf{M}(\boldsymbol{\Sigma})+\phi_{1,lk}^\mathbf{M}(\boldsymbol{\Sigma})$ by solving the  ODE by adding  (\ref{xxy}) and (\ref{yyx}) together. Then, we obtain either $\phi_{1,kl}^\mathbf{M}(\boldsymbol{\Sigma})$ or $\phi_{1,lk}^\mathbf{M}(\boldsymbol{\Sigma})$ by solving one of  them. We obtain a solution for the ODEs in (\ref{8964equ})--(\ref{yyx}) as follows for $k<l$:
\bs\begin{align}	
	&\phi_{1, kk}^\mathbf{M} (\boldsymbol{\Sigma}) = -\frac{s_{kk}^\mathbf{M}}{2 \mu_k }\label{62equ} 	\\
	&\phi_{1, kl}^\mathbf{M}(\boldsymbol{\Sigma})= \frac{s_{kl}^\mathbf{M}}{\mu_k+\mu_l}\left(\frac{\mu_l-\mu_k}{2\widetilde{w}_{ll}}\Sigma_{ll}+\frac{\mu_l-\mu_k}{2\widetilde{w}_{kk}}\Sigma_{kk}-1\right)	\label{56eq}\\
	&\phi_{1,lk}^\mathbf{M}(\boldsymbol{\Sigma})= \frac{s_{kl}^\mathbf{M}}{\mu_k+\mu_l}\left(\frac{\mu_k-\mu_l}{2\widetilde{w}_{ll}}\Sigma_{ll}+\frac{\mu_k-\mu_l}{2\widetilde{w}_{kk}}\Sigma_{kk}-1\right)	\label{57eq}
\end{align}\bsc Using (\ref{62equ})--(\ref{57eq}) and the relationship $\Phi_{1}(\boldsymbol{\Sigma})=(\mathbf{M}^{-1})^\dagger \Phi_1^\mathbf{M}(\boldsymbol{\Sigma}) \mathbf{M}^{-1}$, we can obtain $\Phi_1(\boldsymbol{\Sigma})$.  Therefore, $\nabla_{\boldsymbol{\Delta}} \widetilde{V}_1= \left(\Phi_1(\boldsymbol{\Sigma})+\Phi_1^T(\boldsymbol{\Sigma})\right)\boldsymbol{\Delta}$.

\emph{2) Obtaining $J_1$:}  We first prove the following lemma to obtain the property of $\nu^\ast$.
\begin{Lemma}	\label{lemmaappendxoosds}
	Let $\mathbf{Y}=\mathbf{x} \mathbf{y}^T+\mathbf{y} \mathbf{x}^T$, where $\mathbf{x}, \mathbf{y} \in \mathbb{R}^{L\times 1}$, then the largest eigenvalue of $\mathbf{Y}$ is $\mathbf{y}^T\mathbf{x}+\sqrt{\mathbf{x}^T\mathbf{x} \mathbf{y}^T\mathbf{y}}$ which is always positive,   and the associated eigenvector is $\frac{\mathbf{u}}{|\mathbf{u}|}$ where $\mathbf{u}=\mathbf{x}+\frac{|\mathbf{x}|}{|\mathbf{y}|}\mathbf{y}$.		
\end{Lemma}
\begin{proof}
	Consider a vector $\mathbf{v}=\mathbf{x}+z \mathbf{y}$, we $z$ is real. Then, we have
	\begin{align}
		\mathbf{y}\mathbf{v}& = (\mathbf{x} \mathbf{y}^T+\mathbf{y} \mathbf{x}^T)\mathbf{x} + z(\mathbf{x} \mathbf{y}^T+\mathbf{y} \mathbf{x}^T) \mathbf{y}	\notag \\
		  & = (\mathbf{y}^T\mathbf{x}+ z \mathbf{y}^T\mathbf{y})\mathbf{x}+( \mathbf{x}^T\mathbf{x} +z \mathbf{x}^T\mathbf{y} )\mathbf{y}
	\end{align}
	Let $\mathbf{y}^T\mathbf{x}+ z \mathbf{y}^T\mathbf{y}=\lambda$ and $\mathbf{x}^T\mathbf{x} +z \mathbf{x}^T\mathbf{y}=\lambda z$, then we have $\mathbf{Y}\mathbf{v}=\lambda \mathbf{v}$. Then, $\mathbf{x}^T\mathbf{x} +z \mathbf{x}^T\mathbf{y}=(\mathbf{y}^T\mathbf{x}+ z \mathbf{y}^T\mathbf{y}) z \Rightarrow \mathbf{x}^T\mathbf{x} = \mathbf{y}^T\mathbf{y} z^2 \Rightarrow z= \pm \sqrt{\frac{\mathbf{x}^T\mathbf{x}}{\mathbf{y}^T\mathbf{y}}}$.  Since we are interested in  the larger eigenvalue, letting $z=\sqrt{\frac{\mathbf{x}^T\mathbf{x}}{\mathbf{y}^T\mathbf{y}}}$, we have $\lambda=\mathbf{y}^T\mathbf{x}+\sqrt{\mathbf{x}^T\mathbf{x} \mathbf{y}^T\mathbf{y}}$ which is positive due to the Cauchy-Schwarz inequality. 	
\end{proof}
Letting $\mathbf{x}={\boldsymbol{\Delta}}/{\tau}$ and $y=\boldsymbol{\Sigma} \nabla_{\boldsymbol{\Delta}} V=\boldsymbol{\Sigma} \left( \left(\Phi_1(\boldsymbol{\Sigma}) +(\Phi_1(\boldsymbol{\Sigma}))^T\right) \boldsymbol{\Delta}+ \nabla_{\boldsymbol{\Delta}} J_1\right)$, then 
\begin{align}	
	\nu^\ast=& \boldsymbol{\Delta}^T  \left(\Phi_1(\boldsymbol{\Sigma}) +(\Phi_1(\boldsymbol{\Sigma}))^T\right)  \boldsymbol{\Sigma} \boldsymbol{\Delta}/\tau+ \nabla_{\boldsymbol{\Delta}}^T J_1 \boldsymbol{\Sigma} \boldsymbol{\Delta}/\tau  \notag \\
	&+\sqrt{\boldsymbol{\Delta}^T \boldsymbol{\Delta} \left(\boldsymbol{\Delta}^T \left(\Phi_1(\boldsymbol{\Sigma}) +(\Phi_1(\boldsymbol{\Sigma}))^T\right)+ \nabla_{\boldsymbol{\Delta}}^T J_1\right)\boldsymbol{\Sigma} }\notag \\
	&\sqrt{\boldsymbol{\Sigma} \left( \left(\Phi_1(\boldsymbol{\Sigma}) +(\Phi_1(\boldsymbol{\Sigma}))^T\right)\boldsymbol{\Delta}+\nabla_{\boldsymbol{\Delta}} J_1\right)}/\tau	\label{smallDelta2}
\end{align}
Substituting (\ref{smallDelta2}) into the PDE in (\ref{xxx2}) and balancing the order of $\|\boldsymbol{\Delta}\|$ on both size, we obtain
\begin{align}
	J_1 = \mathcal{O}(\|\boldsymbol{\Delta}\|^4)	\label{jexpress}
\end{align}

\emph{3) Overall solution and  small $\|\boldsymbol{\Delta}\|$  leads to small $\nu^\ast$:}  Combining part 1 and part 2, we obtain the overall solution as follows:
\begin{align}		\label{appeox12}
	\nabla_{\boldsymbol{\Delta}} {V}= \left(\Phi_1(\boldsymbol{\Sigma})+\Phi_1^T(\boldsymbol{\Sigma})\right)\boldsymbol{\Delta}+\mathcal{O}(\|\boldsymbol{\Delta}\|^3)\mathbf{1}
\end{align}
where $\Phi_1(\boldsymbol{\Sigma})=\mathbf{M}^\dagger \Phi_{1}^\mathbf{M} (\boldsymbol{\Sigma}) \mathbf{M} \in \mathbb{R}^{L\times L}$, $ \Phi_{1}^\mathbf{M}(\boldsymbol{\Sigma})=[\phi_{1, kl}^\mathbf{M}(\boldsymbol{\Sigma})]$ is given in  (\ref{62equ})--(\ref{57eq}).   Substituting (\ref{jexpress}) into (\ref{smallDelta2}), we have $\nu^\ast=\mathcal{O}(\|\boldsymbol{\Delta}\|^2)$  for as $\|\boldsymbol{\Delta}\|\rightarrow 0$. Therefore, small $\|\boldsymbol{\Delta}\|$ leads to small $\nu^\ast$.

\vspace{0.3cm}
\hspace{-0.3cm}\emph{B. Solution of (\ref{bellman2}) for large $\nu^\ast$}   
\vspace{0.1cm}

In this part, we solve the PDE  in (\ref{bellman21})  for large $\nu^\ast$. We will show later that large $\|\boldsymbol{\Delta}\|$  leads to this case. Specifically, for large $\nu^\ast$,  the expectation in (\ref{bellman21}) becomes
\bs\begin{align}
	&\mathbb{E}\left[ \left[ \sigma^\ast \nu^\ast-\lambda    \right]^+\big| \boldsymbol{\Delta}, \boldsymbol{\Sigma}   \right] \notag \\
	&=\mathbb{E}\left[ \left[ \sigma^\ast \nu^\ast-\lambda    \right]\big| \boldsymbol{\Delta}, \boldsymbol{\Sigma}   \right]-\int_0^{\lambda/\nu^\ast} (\sigma^\ast \nu^\ast-\lambda) f_{\sigma^\ast}(x) dx \notag \\
	& = \mathbb{E}\left[ \left[ \sigma^\ast \nu^\ast-\lambda    \right]\big| \boldsymbol{\Delta}, \boldsymbol{\Sigma}   \right]-\mathcal{O}\left(\frac{1}{\nu^\ast}\right)
\end{align}\bsc where $f_{\sigma^\ast}(x)$ is the PDF of $\sigma^\ast$ (given in equ. (6) of \cite{eigenvaluesds}). Substituting the above equation into the PDE in (\ref{bellman21}), we obtain
\begin{align}
	&\theta  = \boldsymbol{\Delta}^T \mathbf{S} \boldsymbol{\Delta} -  \mathbb{E}\left[ \left[ \sigma^\ast \nu^\ast-\lambda    \right]\big| \boldsymbol{\Delta}, \boldsymbol{\Sigma}   \right]\bar{F}+\mathcal{O}\left(\frac{1}{\nu^\ast}\right)  \notag \\
	 &+\nabla_{\boldsymbol{\Delta}}^T V  \widetilde{\mathbf{A}} \boldsymbol{\Delta} + \frac{1}{2}\text{Tr}\left( \nabla_{\boldsymbol{\Delta}}^2 V \widetilde{\mathbf{W}} \right)+ \text{Tr}\left( \frac{\partial V}{\partial \boldsymbol{\Sigma}}  \widetilde{\mathbf{W}} \right)
\end{align}
The solution of the above PDE has the structure $V(\boldsymbol{\Delta}, \boldsymbol{\Sigma})=\widetilde{V}_2(\boldsymbol{\Delta}, \boldsymbol{\Sigma})+J_2(\boldsymbol{\Delta}, \boldsymbol{\Sigma})$ ($J_2$ can be treated as a residual error term for $V$), where $\widetilde{V}_2$ and $J_2$ satisfy
\bs\begin{align}
	&\theta  = \boldsymbol{\Delta}^T \mathbf{S} \boldsymbol{\Delta}-\mathbb{E}\left[ \left[ \sigma^\ast \nu^\ast-\lambda    \right]\big| \boldsymbol{\Delta}, \boldsymbol{\Sigma}   \right]\bar{F}   \notag \\
	&+\nabla_{\boldsymbol{\Delta}}^T \widetilde{V}_2  \widetilde{\mathbf{A}} \boldsymbol{\Delta} + \frac{1}{2}\text{Tr}\left( \nabla_{\boldsymbol{\Delta}}^2 \widetilde{V}_2 \widetilde{\mathbf{W}} \right)+ \text{Tr}\left( \frac{\partial \widetilde{V}_2}{\partial \boldsymbol{\Sigma}}  \widetilde{\mathbf{W}} \right)	\label{xxx12}	\\
	&\mathcal{O}\left(\frac{1}{\nu^\ast}\right)  =      \nabla_{\boldsymbol{\Delta}}^T J_2  \widetilde{\mathbf{A}} \boldsymbol{\Delta} + \frac{1}{2}\text{Tr}\left( \nabla_{\boldsymbol{\Delta}}^2 J_2 \widetilde{\mathbf{W}} \right)+ \text{Tr}\left( \frac{\partial J_2}{\partial \boldsymbol{\Sigma}}  \widetilde{\mathbf{W}} \right)\label{xxx22}
\end{align}\bsc 
We then obtain $\widetilde{V}_2$ and $J_2$ by solving the above two equations in the following:

\emph{1) Obtaining $\widetilde{V}_2$:}   In this part, we solve the PDE  in (\ref{xxx12}). We calculate the expectation involved in (\ref{xxx12}) and obtain 
\begin{align}	
		\theta  = \boldsymbol{\Delta}^T \mathbf{S} \boldsymbol{\Delta} & +  \lambda \bar{F}  - \overline{\sigma} \bar{F} \nu^\ast     + \nabla_{\boldsymbol{\Delta}}^T V  \widetilde{\mathbf{A}} \boldsymbol{\Delta} \notag \\
		& + \frac{1}{2}\text{Tr}\left( \nabla_{\boldsymbol{\Delta}}^2 V \widetilde{\mathbf{W}} \right)+ \text{Tr}\left( \frac{\partial V}{\partial \boldsymbol{\Sigma}}  \widetilde{\mathbf{W}} \right)	\label{pde2equ}
\end{align}
where  $\overline{\sigma} \triangleq \mathbb{E}\left[\sigma^\ast\right]$ which depends on the distribution of the largest  singular values $\sigma^\ast$. Specifically, $\overline{\sigma}$ can be calculated as follows:
\begin{align}	\label{calculatesomm1}
	\overline{\sigma} = \int_0^\infty \Pr\left[\sigma^\ast>x\right]d x = \int_0^\infty \left(1-F_{\sigma^\ast}(x) \right)d x 
\end{align}
where $F_{\sigma^\ast}(x)$ is the CDF of $\sigma^\ast$ and is given by \cite{eigandist}: $F_{\sigma^\ast}(x) = \sum_{k=1}^b (-1)^{k-1}C_{k-1}^{r-1}\sum_{a\in \mathcal{A}_k}\text{det}(T_a(x))$,  where $\mathcal{A}_k$ represents the subset of $\{1, \dots, b\}$ with $k$ elements, $T_a(x)$ is a $k\times k$ matrix with $(i,j)$-th element of $\int_0^x\phi_{i(a)}(x)\phi_{j(a)}(x)x^{b-d}\exp(-x) dx$ (where $d\triangleq \min\{N_t, N_r\}$ and $i(a)$ is the $i$-th largest element in $a$),  $\phi_i(x)=\sqrt{\frac{(i-1)!}{(i-1+b-d)!}}Z_{i-1}^{b-d}(x)$ and $Z_{k}^c(x)=\frac{1}{k!}\exp(x)x^{-c}\frac{d^k}{d x^k}(\exp(-x)x^{c+k})$. Using (\ref{calculatesomm1}), $\overline{\sigma}$ can be calculated. We assume that $\nu^\ast     =\mathcal{O}(\|\boldsymbol{\Delta}\|^2)$   for large $\|\boldsymbol{\Delta}\|$. Therefore, we approximate $\nu^\ast     $ using $c \boldsymbol{\Delta}^T \boldsymbol{\Delta}$ for large $\|\boldsymbol{\Delta}\|$ and for some constant $c>0$. We will obtain $c$ in the  part 2 later on. Similarly, the PDE in (\ref{pde2equ}) is separable with solution of the form  $\widetilde{V}_2 =\boldsymbol{\Delta}^T\Phi_2(\boldsymbol{\Sigma})\boldsymbol{\Delta}+ \psi_2(\boldsymbol{\Sigma})$ for some $\Phi_2(\boldsymbol{\Sigma})\in \mathbb{R}^{L\times L}$ and $\psi_2(\boldsymbol{\Sigma})\in \mathbb{R}$. Substituting this form into (\ref{pde2equ}),  letting the coefficient of $\boldsymbol{\Delta}^T\boldsymbol{\Delta}$ to be zero in (\ref{pde2equ}), we obtain 
\bs\begin{align}
	\mathbf{S} -  c\overline{\sigma} \bar{F}  \mathbf{I} + \big(\Phi_2(\boldsymbol{\Sigma})+\Phi_2^T(\boldsymbol{\Sigma})\big) \widetilde{\mathbf{A}}  + \text{Tr}\left( \frac{\partial \Phi_2(\boldsymbol{\Sigma})}{\partial \boldsymbol{\Sigma}}  \widetilde{\mathbf{W}}  \right) ={0}	\label{deltaequation1}
\end{align}\bsc Similarly as solving (\ref{deltaequation}), we can solve for $\Phi_2(\boldsymbol{\Sigma})$. Denoting $\Phi_{2}^\mathbf{M}(\boldsymbol{\Sigma})=(\mathbf{M}^{-1})^\dagger \Phi_{2}(\boldsymbol{\Sigma}) \mathbf{M}^{-1}=[\phi_{2, kl}^\mathbf{M}(\boldsymbol{\Sigma})]$. The solution is given by  for $i<j$: 
\bs\begin{align}	
	&\phi_{2, kk}^\mathbf{M} (\boldsymbol{\Sigma}) = -\frac{s_{kk}^\mathbf{M}(c)}{2 \mu_k }\label{62equ1} 	\\
	&\phi_{2, kl}^\mathbf{M}(\boldsymbol{\Sigma})= \frac{s_{kl}^\mathbf{M}(c)}{\mu_k+\mu_l}\left(\frac{\mu_l-\mu_k}{2\widetilde{w}_{ll}}\Sigma_{ll}+\frac{\mu_l-\mu_k}{2\widetilde{w}_{kk}}\Sigma_{kk}-1\right)	\label{56eq1}\\
	&\phi_{2,ji}^\mathbf{M}(\boldsymbol{\Sigma})= \frac{s_{kl}^\mathbf{M}(c)}{\mu_k+\mu_l}\left(\frac{\mu_k-\mu_l}{2\widetilde{w}_{ll}}\Sigma_{ll}+\frac{\mu_k-\mu_l}{2\widetilde{w}_{kk}}\Sigma_{kk}-1\right)	\label{57eq1}
\end{align}\bsc where $\mathbf{S}^\mathbf{M}(c)  \triangleq (\mathbf{M}^{-1})^\dagger (\mathbf{S}-c\overline{\sigma} \bar{F}  \mathbf{I} ) \mathbf{M}^{-1}=\left[s_{kl}^\mathbf{M}(c)\right]$. Using (\ref{62equ1})--(\ref{57eq1}) and the relationship $\Phi_{2}(\boldsymbol{\Sigma})=(\mathbf{M}^{-1})^\dagger \Phi_2^\mathbf{M}(\boldsymbol{\Sigma}) \mathbf{M}^{-1}$, we can obtain $\Phi_2(\boldsymbol{\Sigma})$.  Therefore, $\nabla_{\boldsymbol{\Delta}} \widetilde{V}_2= \left(\Phi_2(\boldsymbol{\Sigma})+\Phi_2^T(\boldsymbol{\Sigma})\right)\boldsymbol{\Delta}$.

\emph{2) Obtaining $J_2$:}  Using Lemma \ref{lemmaappendxoosds}, Letting $x=\boldsymbol{\Delta}/\tau$ and $y=\boldsymbol{\Sigma} \left( \left(\Phi_2(\boldsymbol{\Sigma}) +(\Phi_2(\boldsymbol{\Sigma}))^T\right) \boldsymbol{\Delta}+ \nabla_{\boldsymbol{\Delta}} J_2\right)$, then 
\bs\begin{align}	
	\nu^\ast=& \boldsymbol{\Delta}^T  \left(\Phi_2(\boldsymbol{\Sigma}) +(\Phi_2(\boldsymbol{\Sigma}))^T\right)  \boldsymbol{\Sigma} \boldsymbol{\Delta}/\tau+ \nabla_{\boldsymbol{\Delta}}^T J_2 \boldsymbol{\Sigma} \boldsymbol{\Delta}/\tau  \notag \\
	&+\sqrt{\boldsymbol{\Delta}^T \boldsymbol{\Delta} \left(\boldsymbol{\Delta}^T \left(\Phi_2(\boldsymbol{\Sigma}) +(\Phi_2(\boldsymbol{\Sigma}))^T\right)+ \nabla_{\boldsymbol{\Delta}}^T J_2\right)\boldsymbol{\Sigma}}\notag \\ 			&\sqrt{\boldsymbol{\Sigma} \left( \left(\Phi_2(\boldsymbol{\Sigma}) +(\Phi_2(\boldsymbol{\Sigma}))^T\right)\boldsymbol{\Delta}+\nabla_{\boldsymbol{\Delta}} J_2\right)}/\tau	\label{smallDelta22}
\end{align}\bsc  Substituting (\ref{smallDelta2}) into the PDE in (\ref{xxx22}) and balancing the order of $\|\boldsymbol{\Delta}\|$ on both size, we obtain
\begin{align}
	J_2 = \mathcal{O}\left(\frac{1}{\|\boldsymbol{\Delta}\|^2}\right)	\label{jexpress2}
\end{align}

Based on  Lemma \ref{lemmaappendxoosds}, $\Phi_{2}(\boldsymbol{\Sigma})$ in part 1, and $J_2$ in (\ref{jexpress2}),  $c$ in  (\ref{62equ1})--(\ref{57eq1}) is determined by the following fixed-point equation $ f(\boldsymbol{\Delta},\boldsymbol{\Sigma}, c) = c\boldsymbol{\Delta}^T \boldsymbol{\Delta}$, where we define \bs$f(\boldsymbol{\Delta},\boldsymbol{\Sigma}, c)     = \boldsymbol{\Delta}^T  \left(\Phi_2(\boldsymbol{\Sigma}) +(\Phi_2(\boldsymbol{\Sigma}))^T\right)   \boldsymbol{\Sigma} \boldsymbol{\Delta}/\tau  + \sqrt{\boldsymbol{\Delta}^T \boldsymbol{\Delta} \boldsymbol{\Delta}^T \left(\Phi_2(\boldsymbol{\Sigma})  +(\Phi_2(\boldsymbol{\Sigma}))^T\right)\boldsymbol{\Sigma} \boldsymbol{\Sigma}  \left(\Phi_2(\boldsymbol{\Sigma}) +(\Phi_2(\boldsymbol{\Sigma}))^T\right)\boldsymbol{\Delta}}/\tau$\bsc.

\emph{3) Overall solution and  large $\|\boldsymbol{\Delta}\|$  leads to large $\nu^\ast$:} Combining part 1 and part 2, we obtain the overall solution as follows:
\begin{align}		\label{appeox22}
	\nabla_{\boldsymbol{\Delta}} V= \left(\Phi_2(\boldsymbol{\Sigma})+\Phi_2^T(\boldsymbol{\Sigma})\right)\boldsymbol{\Delta}+\mathcal{O}\left(\frac{1}{\|\boldsymbol{\Delta}\|^3}\right)\mathbf{1}
\end{align}
where $\Phi_2(\boldsymbol{\Sigma})=\mathbf{M}^\dagger \Phi_{2}^\mathbf{M} (\boldsymbol{\Sigma}) \mathbf{M} \in \mathbb{R}^{L\times L}$, $ \Phi_{2}^\mathbf{M}(\boldsymbol{\Sigma})=[\phi_{1, kl}^\mathbf{M}(\boldsymbol{\Sigma})]$ is given in  (\ref{62equ1})--(\ref{57eq1}).   Substituting (\ref{jexpress2}) into (\ref{smallDelta22}), we have $\nu^\ast=\mathcal{O}(\|\boldsymbol{\Delta}\|^2)$  as $\|\boldsymbol{\Delta}\|\rightarrow \infty$. Therefore, large $\|\boldsymbol{\Delta}\|$ leads to large $\nu^\ast$.

\section*{Appendix G: Proof of Theorem \ref{stab}}

Denote $\boldsymbol{\Lambda}(n)=\mathbb{E}\big[\big(\mathbf{x}(n)-\hat{\mathbf{x}}(n)\big]\big)\big(\mathbf{x}(n)-\hat{\mathbf{x}}(n)\big)^T \big|I_C(n) \big]$. According to the classical Kalman filter theory \cite{ackf}, we have
\begin{align}
	\boldsymbol{\Lambda}(n)=&\boldsymbol{\Sigma}(n) - \boldsymbol{\Sigma}(n)(\mathbf{E}^a (n))^\dagger  \notag \\
	&\big(\mathbf{E}^a (n) \boldsymbol{\Sigma}(n)   (\mathbf{E}^a (n))^\dagger +\mathbf{I}\big)^{-1}\mathbf{E}^a (n) \boldsymbol{\Sigma}(n)	 \label{77equation}
\end{align}
We first have a convergence result on $\mathbb{E}^{\widetilde{\Omega}^\ast}\big[\boldsymbol{\Sigma}(n)\big]$ as $n \rightarrow \infty$ as follows:
\begin{Lemma}	\label{asadfssss11}
	Let $\{\mathbf{P}(n)\}$ be defined as
	\begin{align}	\label{pnsds121}
		\mathbf{P}(n+1) = &\mathbb{E}^{\widetilde{\Omega}^\ast}\left[\mathbf{A} \Big(\mathbf{P}(n) - \mathbf{P}(n)(\mathbf{E}^a(n) )^\dagger \big(\mathbf{E}^a (n) \mathbf{P} (n) \right. \notag \\
		&\left.  (\mathbf{E}^a(n) )^\dagger +\mathbf{I}\big)^{-1}\mathbf{E}^a(n)  \mathbf{P}(n)  \Big) \mathbf{A}^T + \mathbf{W} \right]
	\end{align}
	where $\mathbf{P}(0)=0$. For any given $\bar{F}>0$, we have $\boldsymbol{\overline{\boldsymbol{\Sigma}}}(n)\triangleq \mathbb{E}^{\widetilde{\Omega}^\ast}\big[\boldsymbol{\Sigma}(n)\big]\leq \mathbf{P}(n)$ and  $\lim_{n \rightarrow \infty}\boldsymbol{\overline{\boldsymbol{\Sigma}}}(n)=\boldsymbol{\overline{\boldsymbol{\Sigma}}}$ and $\text{Tr}(\boldsymbol{\overline{\boldsymbol{\Sigma}}})<\infty$.  Furthermore, $\boldsymbol{\overline{\boldsymbol{\Sigma}}}\leq \mathbf{P}\triangleq \lim_{n \rightarrow \infty}\mathbf{P}(n)$, where $\mathbf{P}$ satisfies the fixed-point equation in  (\ref{fxptequ}).
\end{Lemma}
\begin{proof} [Proof of Lemma \ref{asadfssss11}]
	First, it can be verified that the dynamic system in (\ref{plantain}) and (\ref{recvsig}) under $\widetilde{\Omega}^\ast$ is weakly controllable and weakly observable (according to the definitions in Section 3 of \cite{stablity1}). Then, using Lemma 3.2. of \cite{stablity1}, for any  $\bar{F}>0$, we have $\lim_{n \rightarrow \infty}\mathbb{E}\big[\boldsymbol{\Sigma}(n)\big]=\boldsymbol{\overline{\boldsymbol{\Sigma}}}$ and $\text{Tr}(\boldsymbol{\overline{\boldsymbol{\Sigma}}})<\infty$. Using Theorem 3.3. of \cite{stablity1}, for any  $\bar{F}>0$, we have $\mathbb{E}^{\widetilde{\Omega}^\ast}\big[\boldsymbol{\Sigma}(n)\big]\leq \mathbf{P}(n)$ where $\mathbf{P}(n)$ satisfies (\ref{pnsds121}). Furthermore, using Theorem 3.3. of \cite{stablity1}, for any  $\bar{F}>0$, we have  $\boldsymbol{\overline{\boldsymbol{\Sigma}}} \leq \mathbf{P}$, where $\mathbf{P}$ satisfies the following fixed equation:
	\bs\begin{align}
		\mathbf{P}& = \mathbb{E}^{\widetilde{\Omega}^\ast}\left[\mathbf{A} \Big(\mathbf{P} - \mathbf{P}(\mathbf{E}^a )^\dagger \big(\mathbf{E}^a  \mathbf{P}   (\mathbf{E}^a )^\dagger +\mathbf{I}\big)^{-1}\mathbf{E}^a  \mathbf{P}  \Big) \mathbf{A}^T + \mathbf{W} \right] \notag \\
		&= \mathbf{A} \Big(\mathbf{P} - \mathbf{P}\mathbb{E}^{\widetilde{\Omega}^\ast}\left[(\mathbf{E}^a )^\dagger \big(\mathbf{E}^a  \mathbf{P}   (\mathbf{E}^a )^\dagger +\mathbf{I}\big)^{-1}\mathbf{E}^a   \right]\mathbf{P}  \Big) \mathbf{A}^T + \mathbf{W}	\label{pequations1}
	\end{align}\bsc	We calculate the above expectation as follows under $\widetilde{\Omega}^\ast$:
	\bs\begin{align}
		&\mathbb{E}^{\widetilde{\Omega}^\ast}\left[(\mathbf{E}^a )^\dagger \big(\mathbf{E}^a  \mathbf{P}   (\mathbf{E}^a )^\dagger +\mathbf{I}\big)^{-1}\mathbf{E}^a   \right]\notag \\
		=& \mathbb{E}^{\widetilde{\Omega}^\ast} \left[ \mathbb{E}^{\widetilde{\Omega}^\ast}\left[2 \text{Re}\left\{(\mathbf{H}\mathbf{F}^\ast )^\dagger \big(2 \mathbf{H}\mathbf{F}^\ast  \mathbf{P}   (\mathbf{H}\mathbf{F}^\ast)^\dagger +\mathbf{I}\big)^{-1}\mathbf{H}\mathbf{F}^\ast   \right\}\Big|\boldsymbol{\Delta}, \boldsymbol{\Sigma} \right]\right]\label{77eqersd} \\
		= & \mathbb{E}^{\widetilde{\Omega}^\ast} \left[ \mathbb{E}^{\widetilde{\Omega}^\ast}\left[\frac{2\bar{F}\sigma^\ast \mathbf{q}_1\mathbf{q}_1^T }{1+2\bar{F}\sigma^\ast \mathbf{q}_1^T\mathbf{P} \mathbf{q}_1}\Big|\boldsymbol{\Delta}, \boldsymbol{\Sigma} \right]\right] \label{77eqer} \\
		= & \mathbb{E}^{\widetilde{\Omega}^\ast} \left[\int_{\lambda/\nu^\ast }^\infty \left(\frac{2\bar{F}x \mathbf{q}_1\mathbf{q}_1^T }{1+2\bar{F}x \mathbf{q}_1^T\mathbf{P} \mathbf{q}_1}\right)f_{\sigma^\ast}(x)dx\right]	\label{77eqer222}
	\end{align}\bsc Denoting the above equation to be $G(\mathbf{P}, \bar{F}, \lambda)$ and substituting it into (\ref{pequations1}), we obtain the fixed-point equation for $\mathbf{P}$ as in (\ref{fxptequ}).\end{proof}

Denote $ \overline{\boldsymbol{\Lambda}}=\lim_{n \rightarrow \infty}\mathbb{E}^{\widetilde{\Omega}^\ast}\big[ \boldsymbol{\Lambda}(n)\big] $. From (\ref{77equation}) and Lemma \ref{asadfssss11}, we have
\bs\begin{align}
	&\mathbb{E}^{\widetilde{\Omega}^\ast}\left[\boldsymbol{\Lambda}(n)\right]	\notag \\
	=&\mathbb{E}^{\widetilde{\Omega}^\ast}\left[\boldsymbol{\Sigma}(n) - \boldsymbol{\Sigma}(n)\mathbb{E}^{\widetilde{\Omega}^\ast}\left[(\mathbf{E}^a (n))^\dagger\right.\right. \notag \\
	& \left.\left.  \big(\mathbf{E}^a (n) \boldsymbol{\Sigma}(n)   (\mathbf{E}^a (n))^\dagger +\mathbf{I}\big)^{-1}\mathbf{E}^a (n)\Big|\boldsymbol{\Delta}(n),  \boldsymbol{\Sigma}(n)\right] \boldsymbol{\Sigma}(n)	\right]	\notag \\
	\overset{(a)}{\leq}&\mathbf{P}(n) - \mathbf{P}(n)\mathbb{E}^{\widetilde{\Omega}^\ast}\left[(\mathbf{E}^a (n))^\dagger \big(\mathbf{E}^a (n) \mathbf{P}(n)   (\mathbf{E}^a (n))^\dagger +\mathbf{I}\big)^{-1}\right. \notag \\
	& \left. \mathbf{E}^a (n)\right] \mathbf{P}(n)	\notag \\
	\overset{(b)}{\leq}&\mathbf{P}(n) - \mathbf{P}(n) \mathbb{E}^{\widetilde{\Omega}^\ast} \left[\int_{\lambda/\nu^\ast }^\infty \left(\frac{2\bar{F}x \mathbf{q}_1(n)\mathbf{q}_1^T(n) }{1+2\bar{F}x \mathbf{q}_1^T(n)\mathbf{P}(n) \mathbf{q}_1(n)}\right)\right. \notag \\
	& \left. f_{\sigma^\ast}(x)dx\right]\mathbf{P}(n)	 \\
	&\overset{\text{taking limit}}{\Rightarrow}  \overline{\boldsymbol{\Lambda}}\leq\mathbf{P} - \mathbf{P}G(\mathbf{P}, \bar{F}, \lambda)\mathbf{P}	\label{93quationsd}
 \end{align}\bsc where $(a)$ is according to Lemma 3.1. and equ. (13) of \cite{stablity1},  $(b)$ follows the calculations in  (\ref{77eqer222}), and the last line follows the convergence of $\mathbf{P}(n)$ in Lemma \ref{asadfssss11} and the continuity of (\ref{93quationsd}) w.r.t. $\mathbf{P}(n)$. Taking trace operator on both sizes of (\ref{93quationsd}), we obtain the MSE upper bound in (\ref{msebound}).

\section*{\txtblue{Appendix H: Proof of Corollary \ref{remarkwwewxx1}}}
\txtblue{Since $\mathbf{x}(n)=\hat{\mathbf{x}}(n)+\boldsymbol{\Delta}(n)$ and we have shown the stability of $\boldsymbol{\Delta}(n)$ under $\widetilde{\Omega}^\ast$ in Theorem \ref{stab}, it is sufficient to show the stability of $\hat{\mathbf{x}}(n)$ under $\widetilde{\Omega}^\ast$ in order to show the stability of $\mathbf{x}(n)$. We then analyze the stability of $\hat{\mathbf{x}}(n)$ under $\widetilde{\Omega}^\ast$. Taking expectation on condition of $I_C(n+1)$ on both sides of  (\ref{plantain}) and substituting $\mathbf{u}^\ast(n)$ in (\ref{statsrr}),  
\begin{align}
	\hat{\mathbf{x}}(n+1) &= \mathbb{E}\big[\mathbf{A} (\boldsymbol{\Delta}(n)+\hat{\mathbf{x}}(n))+\mathbf{B}\mathbf{u}^\ast(n)+\mathbf{w}(n)\big|I_C(n+1)\big]\notag \\
	& = \big(\mathbf{A}+\mathbf{B}\boldsymbol{\Psi}\big)\hat{\mathbf{x}}(n)+\hat{\mathbf{w}}(n)	\label{xhatdyn}
\end{align}
where $\hat{\mathbf{w}}(n)=\mathbb{E}\big[\mathbf{A} \boldsymbol{\Delta}(n)+\mathbf{w}(n)\big|I_C(n+1)\big]$ and according to Section III.B of \cite{rate2}, we have $\hat{\mathbf{W}}\triangleq \lim_{n\rightarrow \infty}\mathbb{E}[\hat{\mathbf{w}}(n)\hat{\mathbf{w}}^T(n)]=\mathbf{A}\big( \lim_{t\rightarrow \infty}\mathbb{E}[\boldsymbol{\Delta}(n)\boldsymbol{\Delta}^T(n)]\big)\mathbf{A}^T+\mathbf{W}- \lim_{t\rightarrow \infty}\mathbb{E}[\boldsymbol{\Delta}(n)\boldsymbol{\Delta}^T(n)]$. Therefore, if $\lim_{n\rightarrow \infty}\mathbb{E}^{\widetilde{\Omega}^\ast}[\boldsymbol{\Delta}(n)\boldsymbol{\Delta}^T(n)]<\infty$, we have $\|\hat{\mathbf{W}}\|<\infty$. Furthermore, from (\ref{xhatdyn}), for large $n$, we have 
\begin{align}
	\mathbb{E}\left[\| \hat{\mathbf{x}}(n+1)\|^2 \right]<\|\mathbf{A}+\mathbf{B}\boldsymbol{\Psi}\|^2\mathbb{E}\left[\| \hat{\mathbf{x}}(n)\|^2 \right]+\|\hat{\mathbf{W}}\|
\end{align}
Since $\|\mathbf{A}+\mathbf{B}\boldsymbol{\Psi}\|<1$ under the optimal CE controller in (\ref{statsrr}) \cite{poweronl}, we have $\lim_{n \rightarrow \infty}\mathbb{E}^{\widetilde{\Omega}^\ast}\left[\| \hat{\mathbf{x}}(n)\|^2 \right]=\frac{\|\hat{\mathbf{W}}\|^2}{1-\|\mathbf{A}+\mathbf{B}\boldsymbol{\Psi}\|^2}<\infty$. Based on the above analysis, we conclude that $\lim_{n \rightarrow \infty}\mathbb{E}^{\widetilde{\Omega}^\ast}\left[\| {\mathbf{x}}(n)\|^2 \right]<\infty$ under  $\mathbf{u}^\ast(n)$ in (\ref{statsrr}).}

\section*{Appendix I: Proof of Corollary \ref{remarkwwewxx}}

\emph{1) MSE Upper Bound in (\ref{msebound}) vs  $\bar{F}$:} We obtain the Taylor expansion of (\ref{77eqersd}) for large $\bar{F}$ as follows:
\bs\begin{align}
	 &\mathbb{E}^{\widetilde{\Omega}^\ast} \left[ \mathbb{E}^{\widetilde{\Omega}^\ast}\left[2 \text{Re}\left\{(\mathbf{H}\mathbf{F}^\ast )^\dagger \big(2 \mathbf{H}\mathbf{F}^\ast  \mathbf{P}   (\mathbf{H}\mathbf{F}^\ast)^\dagger +\mathbf{I}\big)^{-1}\mathbf{H}\mathbf{F}^\ast   \right\}\Big|\boldsymbol{\Delta}, \boldsymbol{\Sigma} \right]\right]\notag \\
	 &=  \mathbf{P}^{-1}-\mathcal{O}\left(\frac{1}{\bar{F}}\right)  \mathbf{P}^{-1}\mathbb{E}^{\widetilde{\Omega}^\ast} \left[ \big(2 \mathbf{H}\mathbf{F}^\ast    (\mathbf{H}\mathbf{F}^\ast)^\dagger \big)^{-1}  \right]   \mathbf{P}^{-1}
\end{align}\bsc  Substituting this into the upper bound in (\ref{msebound}), we obtain 
\bs\begin{align}
		 &\text{Tr}\big(   \mathbf{P} - \mathbf{P}G(\mathbf{P}, \bar{F}, \lambda)\mathbf{P} \big)  \notag \\
		 =&\text{Tr}\left(\mathbf{P} -  \mathbf{P}\left( \mathbf{P}^{-1}-\mathcal{O}\left(\frac{1}{\bar{F}}\right)  \mathbf{P}^{-1}\mathbb{E}^{\widetilde{\Omega}^\ast} \left[ \big(2 \mathbf{H}\mathbf{F}^\ast    (\mathbf{H}\mathbf{F}^\ast)^\dagger \big)^{-1}  \right]   \right.\right. \notag \\
		 & \left.\left. \mathbf{P}^{-1}\right) \mathbf{P}\right)  =\mathcal{O}\left(\frac{1}{\bar{F}}\right) 
\end{align}\bsc

\emph{2) MSE Upper Bound in (\ref{msebound}) vs  $\lambda$:} We obtain the Taylor expansion of (\ref{77eqer222}) for large $\lambda$ as follows:
\bs\begin{align}
	&\mathbb{E}^{\widetilde{\Omega}^\ast} \left[\int_{\lambda/\nu^\ast }^\infty \left(\frac{2\bar{F}x \mathbf{q}_1\mathbf{q}_1^T }{1+2\bar{F}x \mathbf{q}_1^T \mathbf{P} \mathbf{q}_1}\right)f_{\sigma^\ast}(x)dx\right]	\notag \\
	=&\mathbb{E}^{\widetilde{\Omega}^\ast} \left[\int_{\lambda/\nu^\ast }^\infty \left(\frac{ \mathbf{q}_1\mathbf{q}_1^T }{ \mathbf{q}_1^T \mathbf{P} \mathbf{q}_1}\mathcal{O}(1)\right)f_{\sigma^\ast}(x)dx\right]\notag \\
	=&\mathbb{E}^{\widetilde{\Omega}^\ast} \left[\frac{ \mathbf{q}_1\mathbf{q}_1^T }{ \mathbf{q}_1^T \mathbf{P} \mathbf{q}_1}\mathcal{O}\left(\overline{C}_{\sigma^\ast}(\lambda/\nu^\ast)\right)\right]	\label{86equationss}
\end{align}\bsc where $\overline{C}_{\sigma^\ast}$ is the complementary cumulative distribution function of $\sigma^\ast$ and we have $\overline{C}_{\sigma^\ast}(\lambda/\nu^\ast)=\mathcal{O}\left(\frac{\lambda^d}{\exp(\lambda)}\right)$ (where $d\triangleq \min\{N_t, N_r\}$) as $\lambda$ increases according to (6) of \cite{eigenvaluesds}.  Furthermore, since $\|\mathbf{q}_1\|=1$, we have $\frac{ \mathbf{q}_1\mathbf{q}_1^T }{ \mathbf{q}_1^T \mathbf{P} \mathbf{q}_1}\leq   \mu_{max}( \mathbf{P}^{-1})\mathbf{I}$, where $\mu_{max}( \mathbf{P}^{-1})>0$ (because $ \mathbf{P}$ is positive definite and therefore $ \mathbf{P}^{-1}$ is positive definite). Substituting this  into (\ref{86equationss}), we have
\bs\begin{align}
	&\mathbb{E}^{\widetilde{\Omega}^\ast} \left[\int_{\lambda/\nu^\ast }^\infty \left(\frac{2\bar{F}x \mathbf{q}_1\mathbf{q}_1^T }{1+2\bar{F}x \mathbf{q}_1^T \mathbf{P} \mathbf{q}_1}\right)f_{\sigma^\ast}(x)dx\right]\notag \\
	\leq & c(\lambda)\mu_{max}( \mathbf{P}^{-1})\mathbf{I} \label{88eqfassdd}
\end{align}\bsc where we denote $c(\lambda)= \mathbb{E}^{\widetilde{\Omega}^\ast} \left[\mathcal{O}\left(\overline{C}_{\sigma^\ast}(\lambda/\nu^\ast)\right)\right]=\mathcal{O}\left(\frac{\lambda^d}{\exp(\lambda)}\right)$. Substituting (\ref{88eqfassdd}) into (\ref{pequations1}), we have
\begin{align}
	 \mathbf{P} &\geq \mathbf{A} \Big( \mathbf{P} - c(\lambda)  \mathbf{P} \mu_{max}( \mathbf{P}^{-1}) \mathbf{P}  \Big) \mathbf{A}^T + \mathbf{W}	 \notag \\
	 &\geq \mathbf{A} \Big( \mathbf{P} - c(\lambda)\kappa   \mathbf{P}  \Big) \mathbf{A}^T + \mathbf{W} \label{sad892211}
\end{align}
where $\kappa\triangleq \max_{x\in \mathbb{R}^{L \times 1}} \frac{x^T  \mathbf{P} \mu_{max}( \mathbf{P}^{-1}) \mathbf{P}x}{x^T  \mathbf{P} x}>0$.  From  (\ref{sad892211}), we have   $ \mathbf{P} \geq \mathbf{Z}$, where $\mathbf{Z}$ satisfies $\mathbf{Z} = \mathbf{A} \Big(\mathbf{Z} - c(\lambda)\kappa  \mathbf{Z}  \Big) \mathbf{A}^T + \mathbf{W} $ \cite{lypuncsd}. Suppose we are given an  $\lambda$ such that $\sqrt{1-c(\lambda)\kappa}\mathbf{A}$ is stable. Then, we have $\mathbf{Z}=\sum_{k=0}^\infty (1-c(\lambda)\kappa)^k\mathbf{A}^kW (\mathbf{A}^T)^k\geq \mu_{max}(\mathbf{W}) \sum_{k=0}^\infty \left(1-c(\lambda)\kappa\right)^k\mathbf{A}^k (\mathbf{A}^T)^k=\mathcal{O}\left(\frac{1}{c(\lambda)}\right)\mathbf{I}=\mathcal{O}\left(\frac{\exp(\lambda)}{\lambda^d}\right)$ as $\lambda$ increases. Substituting this into the MSE upper bound in (\ref{msebound}), we have 
\begin{align}
	 &\text{Tr}\big(   \mathbf{P} - \mathbf{P}G(\mathbf{P}, \bar{F}, \lambda)\mathbf{P} \big) \notag \\
	  & \geq   \text{Tr}\big(    \mathbf{P} - c(\lambda)\kappa   \mathbf{P}  \big) \geq (1-c(\lambda)\kappa )  \text{Tr}\big(    \mathbf{P}  \big) \notag \\
	 &\geq  (1-c(\lambda)\kappa )  \text{Tr}\big(    \mathbf{Z}  \big) =\mathcal{O}\left(\frac{\exp(\lambda)}{\lambda^d}\right)
\end{align}
where the first inequality   follows  the last inequality in  (\ref{sad892211}).


\vspace{-0.2cm}



\begin{thebibliography}{1}

\bibitem{coup1}
J. P. Hespanha, P. Naghshtabrizi, and Y. Xu, ``A survey of recent results in networked control systems," \emph{Proc. IEEE}, vol. 95, no. 1, pp. 138--162, Jan. 2007.

\bibitem{mimo1}
I. E. Telatar, ``Capacity of multi-antenna Gaussian channels," \emph{Eur. Trans. Telecommun.,}, vol. 10, no. 6, pp. 585--595, Nov. 1999.


\bibitem{selec1}
R. W. Heath and D. J. Love,  ``Multimode antenna selection for spatial multiplexing systems with linear receivers," \emph{IEEE Trans. Signal Process.}, vol. 53, no. 8, pp. 3042--3056, 2005.


\bibitem{optMISO3}
I. H. Kim and D. J. Love, ``On the capacity and design of limited feedback multiuser MIMO uplinks," \emph{IEEE Trans. Inf. Theory}, vol. 54, no. 10, pp. 4712--4724, 2008.

\bibitem{survey2}
G. N. Nair,  et al., ``Feedback control under data rate constraints: An overview,"  \emph{Proc. IEEE}, vol. 95, no. 1, pp. 108--137, 2007.
\txtblue{\bibitem{book1}
S. Y\" uksela and T. Ba\c{s}ar,  \emph{Stochastic Networked Control Systems: Stabilization and Optimization under Information Constraints}.  \ Boston, MA: Birkh\" auser, 2013.}




\bibitem{dmc2}
S. Tatikonda and S. K. Mitter, ``Control under communication constraints,"  \emph{IEEE Trans. Autom. Control}, vol. 49, no. 7, pp. 1056--1068, Jul. 2004.


\bibitem{packetloss2}
K. You, M. Fu, and L. Xie, ``Mean square stability for Kalman filtering with Markovian packet losses," \emph{Automatica}, vol. 47, no. 12, pp. 2647--2657, 2011.



\bibitem{awgn1}
J. S. Freudenberg, R. H. Middleton, and V. Solo, ``Stabilization and disturbance attenuation over a Gaussian communication channel," \emph{IEEE Trans. Automat. Contr.}, vol. 55, pp. 795--799, 2010.
\txtblue{\bibitem{ref1}
A. A.  Zaidi,  et al., ``Stabilization and control over Gaussian networks," \emph{Information and Control in Networks.}  \ Springer International Publishing, pp. 39--85, 2014.}

\bibitem{rate2}
S. Tatikonda, A. Sahai, and S. Mitter, ``Stochastic linear control over a communication channel," \emph{IEEE Trans. Autom. Control}, vol. 49, no. 9, pp. 1549--1561, 2004.



\bibitem{Qli2}
L. Qiu, G. Gu, and W. Chen, ``Stabilization of networked multi-input systems with channel resource allocation,"  \emph{IEEE Trans. Autom. Control}, vol. 58, no. 3, pp. 554--568, 2003.

\bibitem{mdpcref2}
D. P. Bertsekas, \emph{Dynamic Programming and Optimal Control, 3rd ed}. Massachusetts: Athena Scientific, 2007.


\bibitem{nodualeffect}
C. Ramesh, H. Sandberg, L. Bao, and K. H. Johansson, ``On the dual effect in state-based scheduling of networked control systems," in \emph{Proc. Amer. Control Conf.}, 2011.

\bibitem{onoff2}
A. Molin and H. Sandra,  ``On LQG joint optimal scheduling and control under communication constraints,"  in \emph{in Proc. 48th IEEE Conf. Decis. Control},, pp. 5832--5838, Dec. 2009.

\bibitem{mdp1}
K. Gatsis, A. Ribeiro, and G. J. Pappas, ``Optimal power management in wireless control systems," in \emph{Proc. Amer. Control Conf.}, pp. 1562--1569, 2013.
\txtblue{
\bibitem{ref2}
T. Ba\c{s}ar and R. Bansal, ``Optimum design of measurement channels and control policies for linear-quadratic stochastic systems," \emph{European Journal of Operational Research}, vol. 73, no. 2, pp. 226--236, 1994.}


\bibitem{hjbs}
N. V. Krylov, ``The rate of convergence of finite-difference approximations for Bellman equations with Lipschitz coefficients," \emph{Applied Mathematics \& Optimization}, vol. 52, no. 3, pp. 365--399, 2005.
\txtblue{
\bibitem{ref3}
T. Ba\c{s}ar, ``A trace minimization problem with applications in joint estimation and control under nonclassical information," 
\emph{Journal of Optimization Theory and Applications}, vol. 31, no. 3, pp. 343--359, 1980.
\bibitem{para1}
W. Chen and L. Qiu,  ``Stabilization of multirate networked control systems," \emph{2011 50th IEEE Conference on Decision and Control and European Control Conference}, pp. 5274--5280, 2011.
\bibitem{para2}
N. Xiao, X. Lihua, and L. Qiu, ``Feedback stabilization of discrete-time networked systems over fading channels," \emph{IEEE Trans. Autom. Control}, vol. 57, no. 9, pp. 2176--2189, 2012.
\bibitem{ci2}
A. Nayyar, A. Mahajan, and D. Teneketzis, ``Optimal control strategies in delayed sharing information structures," \emph{IEEE Trans. Autom. Control}, vol. 56, no. 7, pp.  1606--1620,  2011.
\bibitem{ci3}
Y. Serdar, ``Stochastic nestedness and the belief sharing information pattern," \emph{IEEE Trans. Autom. Control}, vol. 54, no. 12, pp.  2773--2786,  2009.
\bibitem{ci4}
C. Y. Chong, M. Athans, ``On the periodic coordination of linear stochastic systems," \emph{Automatica}, vol. 12, no. 4, pp. 321--335, 1976.}

\bibitem{dualeffect}
Y. Bar-Shalom and E. Tse, ``Dual effect, certainty equivalence, and separation in stochastic control," \emph{IEEE Trans. Autom. Control}, vol. 19, no. 5, pp. 494--49, 1974.


\bibitem{automobile}
L. Zhang, H. Gao, and O. Kaynak, ``Network-induced constraints in networked control systemsÑA survey," \emph{IEEE Trans. Ind. Informat.}, vol. 9, no. 1, pp. 403--416, 2013.

\bibitem{white}
T. Kailath,  ``An innovations approach to least-squares estimation--Part I: Linear filtering in additive white noise," \emph{IEEE Trans. Autom. Control}, vol. 13, no. 6, pp. 646--655, 1968.

\bibitem{sampleplant}
M. Micheli and M. I. Jordan, ``Random Sampling of a continuous- time stochastic dynamical system," in \emph{Proc. 15th Intl. Symposium on the Mathematical Theory of Networks and Systems (MTNS)}, Aug. 2002.

\bibitem{datastream}
D. J. Love and R. W. Heath, ``Limited feedback unitary precoding for spatial multiplexing systems," \emph{IEEE Trans. Inf. Theory}, vol. 51, no. 8, pp. 2967--2976, 2005.
\txtblue{\bibitem{ncsmimo2}
L. Qiu, G. Gu, and W. Chen, ``Stabilization of networked multi-input systems with channel resource allocation", \emph{IEEE Trans. Autom. Control}, vol. 58, no. 3, pp. 554--568, Mar. 2013.
\bibitem{fd1}
S. Dey, A. Leong, and J. Evans, ``Kalman filtering with faded measurements," \emph{Automatica}, vol. 45, no. 10, pp. 2223--2233, Oct. 2009.}
\txtblue{\bibitem{uplink}
J. Jose, et al., ``Pilot contamination and precoding in multi-cell TDD systems,"  \emph{IEEE Trans. Wireless Commun.}, vol. 10, no. 8, pp. 2640--2651,  2011.}

\bibitem{training}
S. Ohno and G. B. Giannakis, ``Capacity maximizing MMSE-optimal pilots for wireless OFDM over frequency-selective block Rayleigh-fading channels," \emph{IEEE Trans. Inf. Theory}, vol. 50, no. 9, pp. 2138--2145, 2004.

\bibitem{poweronl}
L. Shi, P. Cheng, and J. Chen, ``Sensor data scheduling for optimal state estimation with communication energy constraint," \emph{Automatica}, vol. 47, no. 8, pp. 1693--1698, 2011.


\bibitem{ackf}
S. L. Goh and D. P. Mandic, ``An augmented extended Kalman filter algorithm for complex-valued recurrent neural networks," \emph{Neural Computation}, vol. 19, no. 4, pp. 1039--1055, 2007.


\bibitem{mdpsurvey}
V. K. N. Lau and Y. Cui,  "Delay-optimal power and subcarrier allocation for OFDMA systems via stochastic approximation,"  \emph{IEEE Trans. Wireless Commun.}, vol. 9, no. 1, pp. 227--233, 2010.


\bibitem{dominmethod}
C. M. Bender and S. A. Steven,  \emph{Advanced Mathematical Methods for Scientists and Engineers I: Asymptotic Methods and Perturbation Theory, Vol. 1}. \emph{Springer}, 1999.






















\bibitem{baseline1}
N. Ravindran, and N. Jindal, ``Limited feedback-based block diagonalization for the MIMO broadcast channel," \emph{IEEE J. Sel. Areas Commun.}, vol.  26, no. 8, pp. 1473--1482, 2008.
\txtblue{\bibitem{adp}
H. Xu, S. Jagannathan, and F. L. Lewis, ``Stochastic optimal control of unknown linear networked control system in the presence of random delays and packet losses," \emph{Automatica}, vol.  48, no. 6, pp. 1017--1030, 2012.
\bibitem{adpp}
J. N. Tsitsiklis and B. V. Roy, ``Average cost temporal-difference learning," \emph{Automatica}, vol. 35, no. 11, pp. 1799--1808,  1999.
\bibitem{asdacite}
A. Molin and S. Hirche, ``On the optimality of certainty equivalence for event-triggered control systems," \emph{IEEE Trans. Autom. Control}, vol. 58, no. 2, pp. 470--474, 2013.}
\txtblue{\bibitem{hinidhsd}
P. Tseng, ``Solving $h$-horizon, stationary Markov decision-problems in time proportional to $\log(h)$," \emph{Oper. Res. Lett.}, vol. 9, no. 5, pp. 287--297, 1990.}

\bibitem{eigenvaluesds}
C. G. Khatri, ``Non-central distribution of i-th largest characteristic roots of three matrices concerning complex multivariate multivariate normal populations," \emph{J. Institute of Ann. Statistical Math.}, vol. 21 pp. 23--32, 1969.

\bibitem{eigandist}
R. Kwan, C. Leung, and P. Ho,``Distribution of ordered eigenvalues of Wishart matrices," \emph{IEE Electronic Letters,}, vol. 43, no. 5, pp. 277--279, Mar. 2007.

\bibitem{stablity1}
P. Bougerol,  "Almost sure stabilizability and Riccati's equation of linear systems with random parameters," \emph{SIAM J. Control Optim.}, vol. 33, no. 3, pp. 702-717, 1995.


\bibitem{lypuncsd}
J. Daafouz, , P. Riedinger, and C. Iung, ``Stability analysis and control synthesis for switched systems: a switched Lyapunov function approach," \emph{IEEE Trans. Autom. Control}, vol. 47, no. 11, pp. 1883--1887, 2002.

\end{thebibliography}
\end{document}